\documentclass[10pt, reqno]{amsart}
\usepackage{amssymb,amsthm, amsmath, latexsym, enumerate, subfig, color,graphicx}
\usepackage[colorlinks = true,
		 linktoc = page,
                 linkcolor = black,
                 citecolor = black
]{hyperref}

\usepackage[top=53pt, bottom=53pt, left=90pt, right=90pt]{geometry}
\theoremstyle{plain}%
\newtheorem{thm}{Theorem}[section]
\newtheorem{lem}{Lemma}
\newtheorem{cor}{Corollary}[section]
\newtheorem{prop}{Proposition}[section]
\newtheorem{conj}{Conjecture}[section]
\theoremstyle{definition}
\newtheorem{defn}{Definition}

\newtheorem{remark}{Remark}

\newtheorem*{wext}{Weak extension principle}
\newtheorem*{gext}{Generalized extension principle}

\newcommand{\beqn}{\begin{equation}}
\newcommand{\beqna}{\begin{eqnarray}}
\newcommand{\eeqn}{\end{equation}}
\newcommand{\eeqna}{\end{eqnarray}}
\newcommand{\m}{\mathfrak{m}}
\newcommand{\chg}{\mathcal{CH}_{\Gamma}}
\newcommand{\sgo}{\mathcal{S}_{\Gamma}^1}
\newcommand{\sgt}{\mathcal{S}_{\Gamma}^2}
\newcommand{\ii}{{\rm i}}
\newcommand{\us}{u_*}
\newcommand{\vs}{v_*}
\newcommand{\ue}{U'}
\newcommand{\ve}{V'}
\newcommand{\ut}{\tilde{u}}
\newcommand{\vt}{\tilde{v}}
\newcommand{\e}{\mathfrak{e}}
\newcommand{\M}{\mathcal{M}}

\newcommand{\Q}{\mathcal{Q}}
\newcommand{\R}{\mathbb{R}}
\newcommand{\dd}{\textrm{d}}
\newcommand{\D}{\textrm{D}}
\newcommand{\Dm}{\mathcal{D}}
\newcommand{\Qm}{Q_{\textrm{m}}}
\newcommand{\Qe}{Q_{\textrm{e}}}
\newcommand{\pu}{\partial_u}
\newcommand{\pv}{\partial_v}
\newcommand{\imp}{\mathfrak{Im}[\phi]}
\begin{document}

\author[]{\sc Jonathan Kommemi\\ Department of Applied Mathematics and Theoretical Physics\\ University of Cambridge}
\date{\today}

\title[]{The global structure of spherically symmetric charged scalar field spacetimes}

\begin{abstract}
We study the spherical collapse of self-gravitating charged scalar fields.  The main result gives a complete characterization of the future boundary of spacetime, providing a starting point for studying the cosmic censorship conjectures.  In general, the boundary includes two null components, one emanating from the center of symmetry and the other from the future limit point of null infinity, joined by an achronal component to which the area-radius function $r$ extends continuously to zero.  Various components of the boundary \emph{a priori} may be empty and establishing such generic emptiness would suffice to prove formulations of weak or strong cosmic censorship.    As a simple corollary of the boundary characterization, the present paper rules out scenarios of `naked singularity' formation by means of `super-charging' (near-)extremal Reissner-Nordstr\"om black holes.   The main difficulty in delimiting the boundary is isolated in proving a suitable global extension principle that effectively excludes a broad class of singularity formation.  This suggests a new notion of `strongly tame' matter models, which we introduce in this paper.  The boundary characterization proven here extends to any such `strongly tame' Einstein-matter system. 
\end{abstract}

\maketitle

\setcounter{tocdepth}{2}

{{\tableofcontents} }

\section{Introduction}\label{sec:intro}
A fundamental open problem in classical general relativity concerns the structure of singularities formed by the gravitational collapse of self-gravitating bodies.  There are two conjectures that one gives the utmost consideration: weak and strong cosmic censorship. Each is a statement regarding the `visibility' of singularities with respect to distinct notions of spacetime predictability: weak cosmic censorship is concerned with predictability in the sense of completeness of future null infinity, e.g.,~as given by Christodoulou \cite{DC98}; and, strong cosmic censorship is concerned with predictability in the sense of spacetime inextendibility. The choice of terminology is a bit unfortunate; the two conjectures are, in fact, strictly speaking, logically independent. 

A rigorous formulation of either conjecture is deeply rooted in the field of partial differential equations (PDEs).  Indeed, the theory of general relativity itself is only described, in a mathematical context, as an initial value problem to the (second-order quasilinear hyperbolic) Einstein equations
\beqn\label{EIN}
R_{\mu\nu}- \frac{1}{2}g_{\mu\nu}R = 8\pi T_{\mu\nu}
\eeqn
(in natural units, with $c = G = 1$).
It is well-known that given sufficiently regular initial data, there exists a unique maximal globally hyperbolic spacetime $(\mathcal{M}, g_{\mu\nu})$ that solves (\ref{EIN}) coupled to various matter models, the so-called maximal future development of initial data \cite{CB52, CB71,CBG69}.  Little is known \emph{a priori} about the structure of this spacetime in the large, as our understanding of global initial value problems for wave equations of such non-linearity is limited.

This paper considers the formation and characterization of singularities in spherically symmetric asymptotically flat (with one end) Einstein-Maxwell-Klein-Gordon spacetimes. Previously, these spacetimes have been studied numerically (for the collapse of a charged massless scalar field) by Hod and Piran \cite{SHTP98,HP98} and small data global existence has been shown by Chae \cite{DChae03}.  This model, important to understanding many phenomena relevant to gravitational collapse, generalizes the models of Christodoulou \cite{DC87} and Dafermos \cite{MD03}. The self-gravitating real-valued massless scalar field model of Christodoulou is completely understood  in regards to weak and strong cosmic censorship \cite{DC91, DC93, DC94,DC99}, but his model does not admit what is possibly the most relevant feature counter-examples to strong cosmic censorship have: a Cauchy horizon emanating from the future limit point of null infinity. In `coupling' an electromagnetic field to the model of Christodoulou, Dafermos is able to study the stability and instability of Cauchy horizon formation \cite{MD03, MD05c}, but his model is limited, in turn, by global topology\footnote{The electromagnetic field is only `coupled' to the \emph{real-valued} scalar field via its interaction with the geometry and thus requires non-trivial topology in the initial data set for the charge to be non-zero.} incompatible with the spacetime having only one asymptotically flat end. A further motivation for coupling an electromagnetic field, moreover, stems from the possibility of identifying charge with a `poor man's' notion of angular momentum (cp.~Reissner-Nordstr\"om and Kerr black hole solutions), providing a natural primer to the more difficult problem of non-spherical collapse. With the present model, we can address both cosmic censorship conjectures within a single framework admitting many of the most fascinating features of gravitational collapse.  

While we do not prove or disprove here the cosmic censorship conjectures for this model, this paper will lay a framework that will provide the necessary tools to tackle these very difficult problems in the future. The main result in Theorem \ref{thm:main} will expound on the possible global structure of spacetime, giving a characterization of its future boundary.  In general, the boundary includes two null components, one emanating from the center of symmetry and the other from the future limit point of null infinity, joined by an achronal component to which the area-radius function $r$ extends continuously to zero.  Some components of the boundary may be \emph{a priori} empty and establishing such generic emptiness would suffice to prove the conjectures. In \S \ref{sec:intro/forth}, however, we announce forthcoming results that will assert the non-emptiness of various boundary components that will affect the outcome of (versions of) cosmic censorship (cf~\S \ref{sec:intro/conj/strong}).  With respect to the state of cosmic censorship, we give an overview of results given by Christodoulou and Dafermos in \S\ref{sec:intro/cd}. Stemming from these specific models, we are then led to give a series of conjectures in \S\ref{sec:intro/conj} for the general Einstein-Maxwell-Klein-Gordon system. These conjectures and their relationship, in particular, to cosmic censorship are highlighted in \S\ref{sec:intro/conj/web}. Within the present paper, we bolster the case for the validity of weak cosmic censorship by immediately ruling out (as a corollary of Theorem \ref{thm:main}) the possibility, as entertained in the physics literature, e.g.,~\cite{BCK10, CLS10, FdFYY01, VH99,TJTS09, MS07,RS08,SS11,RW72}, of creating `naked singularities' by `super-charging' (near-)extremal Reissner-Nordstr\"om black holes.

The main difficulty in the proof of Theorem \ref{thm:main} is establishing a stronger characterization of `first singularities' than that proposed by Dafermos in \cite{MD05b}.  We will call this stronger characterization the `generalized extension principle' to distinguish it from (what we shall call) the `weak extension principle' of \cite{MD05b}. The `generalized extension principle' is formulated very generally without reference to the topology or the geometry of the spherically symmetric initial data (so as to be applicable as well in the cosmological setting or in the case with two asymptotically flat ends). See \S \ref{sec:intro/general/weak} for a formal definition.  In \S\ref{sec:general_emkg/proof_gext}, we will show that Einstein-Maxwell-Klein-Gordon satisfies both extension principles, specifically: A `first singularity' must emanate from a spacetime boundary to which the area-radius function $r$ extends continuously to zero.  

Given the `generalized extension principle', the proof of Theorem \ref{thm:main} follows from monotonicity arguments derived from the dominant energy condition\footnote{Much of Theorem \ref{thm:main}, in fact, uses the monotonicity governed by Raychaudhuri's equation, which just needs the null energy condition (cf.~the proof of Theorem \ref{thm:main} in \S\ref{sec:proof_main}).} and is independent of any other structure of the system.  This observation will motivate a notion of `strongly tame' and `weakly tame' Einstein-matter systems introduced in this paper.  In the case of a `strongly tame' system, i.e.,~one which satisfies the `generalized extension principle' and obeys the dominant energy condition, the conclusion of Theorem \ref{thm:main} holds.  In the case of a `weakly tame' system, i.e.,~one which satisfies the `weak extension principle' and obeys the dominant energy condition, parts of the conclusion of Theorem \ref{thm:main} still hold, in particular, those most relevant for the study of weak cosmic censorship but not strong cosmic censorship.  See \S\ref{sec:intro/general/tame}--\ref{sec:intro/general/weakly_tame} for a discussion.  

We note finally that many of the conjectures in \S\ref{sec:intro/conj} are not model-specific and rely simply on parts of the conclusion of Theorem \ref{thm:main}; these conjectures, in turn, can be read more generally so as to apply to `strongly tame' and `weakly tame' systems (where appropriate).  This paper, as such, provides a blueprint for establishing the global structure of spherically symmetric spacetimes. 
 
\subsection{Self-gravitating charged scalar field model}\label{sec:intro/model}

We briefly summarize the mathematical framework (see, e.g., \cite{TF97} and \cite{GN00}) necessary to describe the self-gravitating charged scalar field model.  

\subsubsection{Principal $U(1)$-bundles and associated complex line bundles}\label{prelimbundle}
Consider a 4-dimensional spacetime $(\M, g_{\mu\nu})$ and a principal $U(1)$-bundle $\pi: P\rightarrow \M$.  

Define the adjoint principal bundle $P\times_{U(1)} U(1)$ by
\beqn\nonumber
P\times_{U(1)} U(1) = \left(P\times U(1)\right)/ \sim,
\eeqn
where $(p_1, g_1)\sim (p_2, g_2)$ if there exists $g\in U(1)$ such that $p_2 = p_1\cdot g$ and $g_2 = g^{-1}\cdot g_1 \cdot g$. 
Consider a local section $s: U\subset \M \rightarrow \pi^{-1}(U)\subset P$ of the principal bundle (called a gauge choice).  This mapping induces an isomorphism between each fiber $P_x$ of the principal bundle and $\{x\}\times U(1)$ for each $x\in U$ by sending $p\in P_x$ to the unique $g\in U(1)$ such that $p = s(x)\cdot g$.  A local section of the adjoint bundle assigns the equivalence class $[s(x), g(x)]$ of a pair $(s(x), g(x))$ at each $x$. Let $f_{g(x)}$ be the left action of $g(x)$ on each fiber $\{x\}\times U(1)$.  Then, for each $p\in P_x$ there exists a unique $g\in U(1)$ such that $p = s(x)\cdot g$ and
\beqn\nonumber
f_{g(x)} (p) = s(x)\cdot (g(x) \cdot g).
\eeqn
The mapping $f_{g(x)}$ is a vertical automorphism of $P$ (called a gauge transformation).  That is, $f_{g(x)}: P \rightarrow P$ is a diffeomorphism that fixes fibers $P_x$ of $P$, i.e., $\pi \circ f_{g(x)} = \pi$, and commutes with the $U(1)$-action, i.e., for $p\in P_x$ and $h\in U(1)$,  $f_{g(x)}(p\cdot h) = f_{g(x)}(p) \cdot h$.

Let $\omega$ be a $\mathfrak{u}(1)$-valued connection defined on $P$ and let $\mathcal{F}= \dd \omega$ denote its curvature.  Under a gauge transformation $f_{g(x)}:P\rightarrow P$, the connection and curvature transform on each fiber as
\beqna\nonumber
f_{g(x)}^*\omega &=& g(x)^{-1}\cdot\omega\cdot g(x) + g(x)^{-1}\cdot\dd g(x)\\
  f_{g(x)}^*\mathcal{F} &=& g(x)^{-1}\cdot \mathcal{F}\cdot g(x).\nonumber
\eeqna
In particular, as $U(1)$ is abelian, $\mathcal{F}$ is invariant under gauge transformations, i.e.,~
\beqn\nonumber
f_{g(x)}^*\mathcal{F} = \mathcal{F}.
\eeqn
Using the gauge choice $s: U \rightarrow \pi^{-1}(U)$, we define a $\mathfrak{u}(1)$-valued 1-form $\tilde{A} = s^*\omega$ on $U$ (called the local gauge potential) and a $\mathfrak{u}(1)$-valued 2-form $\tilde{F} = s^*\mathcal{F} = \dd\tilde{A}$ on $U$ (called the electromagnetic field strength). Identifying the Lie algebra $\mathfrak{u}(1)$ with $\ii\R$, we, in turn, define a $\R$-valued 1-form $A$ on $U$ such that
\beqn\nonumber
\tilde{A} = \ii A
\eeqn
and hence define a (global) $\R$-valued 2-form $F$ on $\M$ such that
\beqn\nonumber
F =  \dd A.
\eeqn  

Fix $\e\in \mathbb{Z}$. Let $\varrho: U(1)\rightarrow GL(1,\mathbb{C})$ be a representation of $U(1)$ on the vector space $\mathbb{C}$ (over the field $\mathbb{C}$) given by $\varrho(g) = g^{\e}$.
Define
\beqn\nonumber
P\times_{\varrho}\mathbb{C} = \left(P\times \mathbb{C} \right)/ \sim,
\eeqn
where $(p_1, z_1)\sim (p_2, z_2)$ if there exists $g\in U(1)$ such that $p_2 = p_1\cdot g$ and $z_2 = \varrho(g)^{-1} \cdot z_1$. We call the fiber bundle
\beqn\nonumber
\pi_{\varrho}: P\times_{\varrho}\mathbb{C} \rightarrow \M, \hspace{1cm} [p, z]\mapsto \pi(p)
\eeqn
the associated complex line bundle through the representation $\varrho$. A charged scalar field $\phi$ is a global section of $P\times_{\varrho}\mathbb{C}$ and corresponds to an equivariant $\mathbb{C}$-valued map on $P$, i.e.,~
\beqn\nonumber
\phi(p\cdot g) = \varrho(g)^{-1}\cdot \phi(p)
\eeqn
for all $p\in P$ and all $g\in U(1)$. Let $\omega_{\varrho}= \varrho_*\omega$ be the associated $\mathfrak{gl}(1)$-valued (in fact, $\mathfrak{u}(1)$-valued)  connection on $P\times_{\varrho}\mathbb{C}$.  This connection is defined as follows: since $\omega(X)\in \mathfrak{u}(1)$ for all $X\in TP$ and $\varrho_*(\omega (X))\in \mathfrak{u}(1)\subset\mathfrak{gl}(1)$, put
\beqn\nonumber
\omega_{\varrho}(X)= \varrho_*(\omega (X)).
\eeqn
Using the gauge choice $s: U\rightarrow \pi^{-1}(U)$, we can similarly pullback this connection to a  $\mathfrak{u}(1)$-valued 1-form on $U$ by defining $\tilde{A}_{\varrho} = s^* \omega_{\varrho} = \e \ii A$. The exterior covariant derivative (called the gauge derivative) on sections of $P\times_{\varrho}\mathbb{C}$ is then defined\footnote{The definition of the exterior covariant derivative is independent of the metric on $\M$, but transforms `tensorially' under a gauge transformation.  Equivalently, we can write $\D= \nabla + \e\ii A$.} so that locally on $\M$
\beqn\nonumber
\D  = \dd + \e \ii A.
\eeqn 

When $\e = 0$, the associated complex line bundle reduces to (the trivial bundle) $P\times \mathbb{C}$ and, moreover, the only admissible gauge transformations $f_{g(x)}: P\rightarrow P$ are those for which the choice of $g(x)\in U(1)$ is independent of the base point $x$ (commonly referred to as a `global' gauge transformation).  

 \subsubsection{System of equations}\label{sec:intro/model/eqns}
Given $\m^2\in \R$, the charged scalar field model is described by the collection (cf.~\S \ref{prelimbundle})
\beqn\nonumber
\left\{(\M, g_{\mu\nu}),~ \e,~\pi_{\varrho}: P\times_{\varrho}\mathbb{C}\rightarrow 
\M,~ \omega_{\varrho}, ~\phi , ~\m^2\right\}
\eeqn
satisfying the Einstein-Maxwell-Klein-Gordon equations
\beqn\label{RMN}
R_{\mu\nu} - \frac{1}{2}g_{\mu\nu}R = 8\pi T_{\mu\nu} = 8\pi \left(T_{\mu\nu}^{\textrm{EM}} + T_{\mu\nu}^{\textrm{KG}}\right)
\eeqn
\beqn\label{TEM}
T_{\mu\nu}^{\textrm{EM}} = \frac{1}{4\pi} \left(g^{\alpha\beta}F_{\alpha\mu}F_{\beta\nu} -\frac{1}{4}g_{\mu\nu}F^{\alpha\beta}F_{\alpha\beta}\right)
\eeqn
\beqn\label{TKG}
T_{\mu\nu}^{\textrm{KG}}\ = \frac{1}{2}\D_{\mu}\phi\left(\D_{\nu}\phi\right)^{\dagger} + \frac{1}{2}\D_{\nu}\phi \left(\D_{\mu}\phi\right)^{\dagger} -\frac{1}{2}g_{\mu\nu}\left(g^{\alpha\beta}\D_{\alpha}\phi \left(\D_{\beta}\phi\right)^{\dagger}+\m^2\phi\phi^{\dagger}\right)
\eeqn
\beqn\label{max2}
\nabla^{\nu}F_{\mu\nu}=  2\pi\e\ii \left(\phi \left(\D_{\mu}\phi\right)^{\dagger}-\phi^{\dagger}\D_{\mu}\phi \right)
\eeqn
\beqn\label{eqn:kg}
g^{\mu\nu}\D_{\mu}\D_{\nu}\phi = \m^2 \phi.
\eeqn
The coupling constant $\e$ is to be interpreted as the (electromagnetic) charge of the scalar field having mass $\m^2$.
\subsubsection{Dominant energy condition}

To ensure that the matter model obeys the dominant energy condition, we must require that $\m^2\geq0$ (cf.~\S\ref{energymomentum}). We note, however, that the `generalized extension principle' for the Einstein-Klein-Gordon system ($\e= F_{\mu\nu}=0$) can be established for arbitrary $\m^2$ (cf.~footnote \ref{foot:dominant}).  

\subsubsection{Well-posedness}\label{sec:intro/model/wellposed}

An easy generalization of a theorem of Choquet-Bruhat and Geroch \cite{CBG69} gives a unique smooth maximal future development $(\M, g_{\mu\nu}, \phi, F_{\mu\nu})$ satisfying (\ref{RMN})--(\ref{eqn:kg}) for given smooth initial data defined on a Cauchy surface $\Sigma^{\scriptsize{(3)}}$ (in our convention, $\M$ is a manifold-with-boundary with (past) boundary $\Sigma^{\scriptsize{(3)}}$).  For the purpose of this paper, it will not be necessary to concern ourselves with a general discussion of the constraint equations and the construction of such initial data sets, as we shall restrict to spherical symmetry where the relevant considerations are straightforward.  See \S \ref{sec:pre/max}.

\subsection{Theorem \ref{thm:main}: global characterization of spacetime}\label{sec:intro/main}

The main result (Theorem \ref{thm:main}) of this paper concerns the general possible structure of a spherically symmetric Einstein-Maxwell-Klein-Gordon spacetime arising from gravitational collapse.  This global characterization already captures non-trivial aspects of the dynamics of (\ref{RMN})--(\ref{eqn:kg}) and can be considered a starting point for further study of the cosmic censorship conjectures, whose statements are given in \S\ref{sec:intro/ws}. 

As the proof of Theorem \ref{thm:main} will make clear, this characterization holds for a larger class of spherically symmetric Einstein-matter systems, which we call `strongly tame'. In fact, much of Theorem \ref{thm:main} will hold for a still larger class of `weakly tame' Einstein-matter systems.  With this in mind, in the statement of Theorem \ref{thm:main} below, a box will indicate the specific appeal to an Einstein-matter system being `strongly tame',  otherwise the model need only be `weakly tame' for the assertion to hold.  See \S\ref{sec:intro/general/tame}--\ref{sec:intro/general/weakly_tame} for a discussion and the statements of Theorems \ref{thm:s_t_general} and \ref{thm:w_t_general}.

\begin{thm}\label{thm:main} Let $(\mathcal{M}=\Q^+\times_r \mathbb{S}^2, g_{\mu\nu}, \phi, F_{\mu\nu})$ denote the maximal future development of smooth spherically symmetric asymptotically flat initial data with one end for the Einstein-Maxwell-Klein-Gordon system containing no anti-trapped\footnote{A point $p\in \Q^+$ is an anti-trapped sphere of symmetry if the ingoing null derivative of $r$, evaluated at $p$, is non-negative. (This is what sometimes would be called `\emph{past outer} trapped or marginally trapped'.)} spheres of symmetry, where $r: \Q^+\rightarrow [0, \infty)$ is the area-radius function.

\vspace{.5mm}
\begin{flushleft}{\textbf{\emph{I:~Penrose diagram}}}\end{flushleft}
\vspace{.5mm}
The Penrose diagram\footnote{A Penrose diagram is the range of globally defined bounded double null co-ordinates as a subset of $\R^{1+1}$ (cf.~\S\ref{sec:pre/max}).  For spherically symmetric spacetimes, the diagrams conveniently help convey global causal-geometric information about the metric. Readers unfamiliar with Penrose diagrams should consult the appendix of \cite{MDIR05}.
} of $\Q^+$ is as depicted

\begin{center}
\includegraphics[scale=1.1]{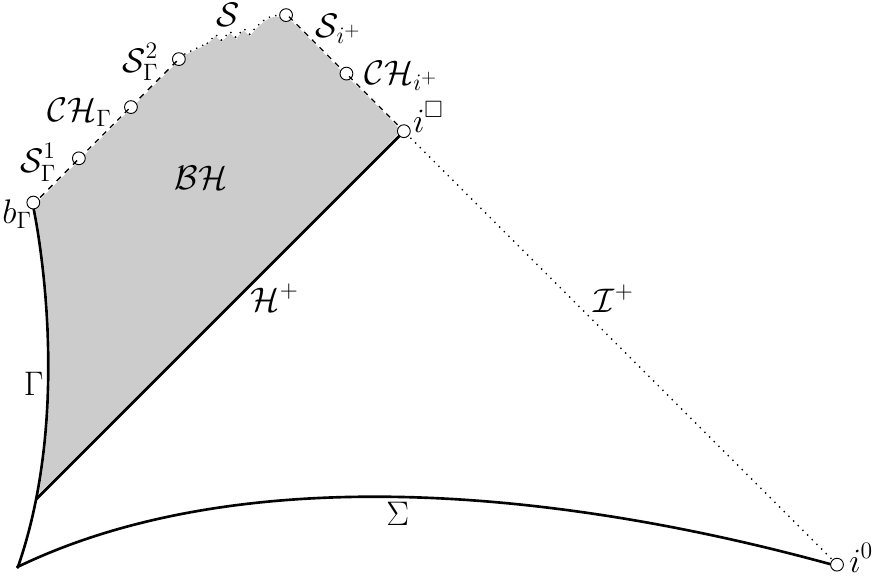}
\end{center}
with boundary $ \Sigma\cup \Gamma$ in the sense of manifold-with-boundary and boundary $\mathcal{B}^+$ induced by the ambient manifold $\R^{1+1}$ admitting a decomposition

\beqn\label{bp}
\mathcal{B}^+= b_{\Gamma}\cup \sgo \cup \chg \cup \sgt \cup\mathcal{S}\cup \mathcal{S}_{i^+}\cup\mathcal{CH}_{i^+}\cup i^{\square}\cup \mathcal{I}^+\cup i^0
\eeqn
\noindent to be enumerated immediately below.
\vspace{1mm}
\pagebreak
\begin{flushleft}{\textbf{\emph{II:~Boundary characterization}}}\end{flushleft}

The spacetime boundary is described as follows:
\vspace{1mm}
\begin{center}{\textbf{\emph{Boundary in the sense of manifold-with-boundary}}}\end{center}
\vspace{1mm}
\noindent $\Sigma$ is the spacelike past boundary of $\Q^+$ and is the projection to $\Q^+$ of the initial Cauchy hypersurface $\Sigma^{\scriptsize{(3)}}$ in $\M$.
 
\vspace{2.5mm}
\noindent
$\Gamma$ is the timelike boundary of $\Q^+$ on which $r=0$ and is the projection to $\Q^+$ of the set of fixed points of the group action $SO(3)$ on $\M$.
\vspace{1mm}
\begin{center}{\textbf{\emph{Boundary induced from the ambient $\R^{1+1}$}}}\end{center}
\vspace{1mm}
\noindent $i^0$ is the unique limit point of $\Sigma$ in $\overline{\Q^+}\backslash \Q^+$.\footnote{The closure is with respect to a bounded conformal representation of $\Q^+$ into the ambient manifold $\R^{1+1}$. See \S \ref{sec:pre/max}.  Similarly, causal-geometric constructions, e.g.,~the causal future $J^+$, the chronological past $I^-$, etc., will be with respect to the topology and the causal structure of the ambient $\R^{1+1}$.}  $r$ extends continuously\footnote{By this we mean here, and in what follows: $r$ extends continuously to a $[0, \infty]$-valued function on $\Q^+\cup i^0$ so as to yield $\infty$ on $i^0$.}  to $\infty$ on $i^0$.

\vspace{2.5mm}
\noindent $\mathcal{I}^+$ is a connected non-empty open null segment emanating from (but not including) $i^0$ characterized by the set of $p\in\overline{\Q^+}\backslash \Q^+$ that are limit points of outgoing null rays in $\Q^+$ for which $r\rightarrow \infty$. $r$ extends continuously to $\infty$ on $\mathcal{I}^+$.

\vspace{2.5mm}
\noindent $i^{\square}$ is the unique future limit point of $\mathcal{I}^+$.

\vspace{2.5mm}
\noindent
 $\mathcal{CH}_{i^+}$ is a connected (possibly empty) half-open null segment emanating\footnote{The `Cauchy horizon' $\mathcal{CH}_{i^+}$ will have special significance within the context of this paper, for it will be the only type of Cauchy horizon that is non-`first singularity'-emanating.  See \S \ref {sec:intro/gext}.}
 from (but not including) $i^{\square}$.  $r$ extends\footnote{The extension, which need not be continuous (see, however, Statement IV.4 below), is given by monotonicity along outgoing null curves.} to a function on $\mathcal{CH}_{i^+}$ that is non-zero except possibly at its future endpoint. 
 
 \vspace{2.5mm}
\noindent
$\mathcal{S}_{i^+}$ is a connected (possibly empty) half-open null segment emanating from (but not including) the future endpoint of $\mathcal{CH}_{i^+}\cup i^{\square}$. $r$ extends continuously to zero on $\mathcal{S}_{i^+}$. 

\vspace{2.5mm}

\noindent$b_{\Gamma}$ is the unique future limit point of $\Gamma$ in $\overline{\Q^+}\backslash \Q^+$. $r$ extends continuously to zero on $b_{\Gamma}\backslash \left(\mathcal{CH}_{i^+}\cup i^{\square}\right)$.
 
 \vspace{2.5mm}
\noindent
  $\sgo$ is a connected (possibly empty) half-open null segment emanating from (but not including) $b_{\Gamma}$.  $r$ extends continuously to zero on $\sgo\backslash \left(\mathcal{CH}_{i^+}\cup i^{\square}\right)$.
  
  \vspace{2.5mm}
\noindent $\chg$ is a connected (possibly empty) half-open null segment emanating from (but not including) the future endpoint of $b_{\Gamma}\cup\sgo$. $r$ extends continuously to a non-zero function on $\chg$ except possibly at its future endpoint.

\vspace{2.5mm}
\noindent$\sgt$ is a connected (possibly empty) half-open null segment emanating from (but not including) the future endpoint of $\chg$. $r$ extends continuously to zero on $\sgt$.\\

\noindent $\mathcal{S}$ is a connected (possibly empty) achronal curve that does not intersect the null rays emanating from limit points $b_{\Gamma}$ and $i^{\square}$. \fbox{ $r$ extends continuously to zero on $\mathcal{S}$.}
\vspace{1mm}
\begin{center}{\textbf{\emph{Common intersection of the boundary components}}}\end{center}
\vspace{1mm}
$\Sigma$ and $\Gamma$ intersect at a single point.

\vspace{2.5mm}
\noindent
If $\mathcal{S} = \emptyset$, the future endpoint of $b_{\Gamma}\cup\sgo\cup \chg\cup\sgt$ coincides with the future endpoint of $\mathcal{S}_{i^+}\cup \mathcal{CH}_{i^+}\cup i^{\square}$. 

\vspace{2.5mm}
\noindent
Modulo these common intersections, the boundary decomposition (\ref{bp}) is disjoint.
\vspace{1mm}
\begin{flushleft}{\textbf{\emph{III:~Completeness of $\mathcal{I}^+$}}}\end{flushleft}
\vspace{1mm}
\noindent If either of the following hold:
 \vspace{-1mm}
 \begin{itemize}
 \item [1.] $\mathcal{BH} = \Q^+ \backslash J^-(\mathcal{I}^+)\neq \emptyset$; or,
 \item [2.] $\sup_{\chg} r < \infty$,\footnote{If $\chg = \emptyset$, then condition 2 is trivially satisfied, as we take the convention $\sup \emptyset = -\infty$.}
 \end{itemize}
then $\mathcal{I}^+$ is complete in the sense of Christodoulou \cite{DC98}. 
  
{\rm [The completeness condition of \cite{DC98}, in the present context, takes the following form: Consider the parallel transport of an ingoing null vector $X$ along a fixed outgoing null segment in $\Q^+$ that has a limit point on $\mathcal{I}^+$. The affine length of integral curves of $X$ (in fact, restricted to $J^-(\mathcal{I}^+)\cap \Q^+$) tends to $\infty$ as $\mathcal{I}^+$ is approached.] }

If $\mathcal{I}^+$ is complete, we write\footnote{This explains the choice of notation $\mathcal{S}_{i^+}$ and $\mathcal{CH}_{i^+}$: If either of these sets are non-empty, then $i^{\square} = i^+$, since condition 1 is satisfied.}  
\beqn\nonumber
i^{\square} = i^+.
\eeqn

Alternatively, when $\mathcal{I}^+$ is not complete, we write
\beqn\nonumber
i^{\square} = i^{\emph{naked}}
\eeqn
and we say that $(\M, g_{\mu\nu})$ is a `naked singularity' spacetime.
\vspace{1mm}

\begin{flushleft}{\textbf{\emph{IV:~Geometry of the trapped region}}}\end{flushleft}
\vspace{1mm}
Let $\mathcal{R}$ be the non-empty `regular region' defined as the set of all $p\in \Q^+$ for which the outgoing null derivative of $r$ is positive.  Let $\mathcal{A}$ be the (possibly empty) `apparent horizon' defined as the set of all $p\in \Q^+$ for which the outgoing null derivative of $r$ vanishes. Let $\mathcal{T}$ be the (possibly empty) `trapped region' defined as the set of all $p\in \Q^+$ for which the outgoing null derivative of $r$ is negative.\footnote{The `apparent horizon' is thus the set of symmetry spheres that are marginally trapped; the `trapped region' is the set of symmetry spheres that are trapped.}

\begin{center}
\includegraphics[scale=.85]{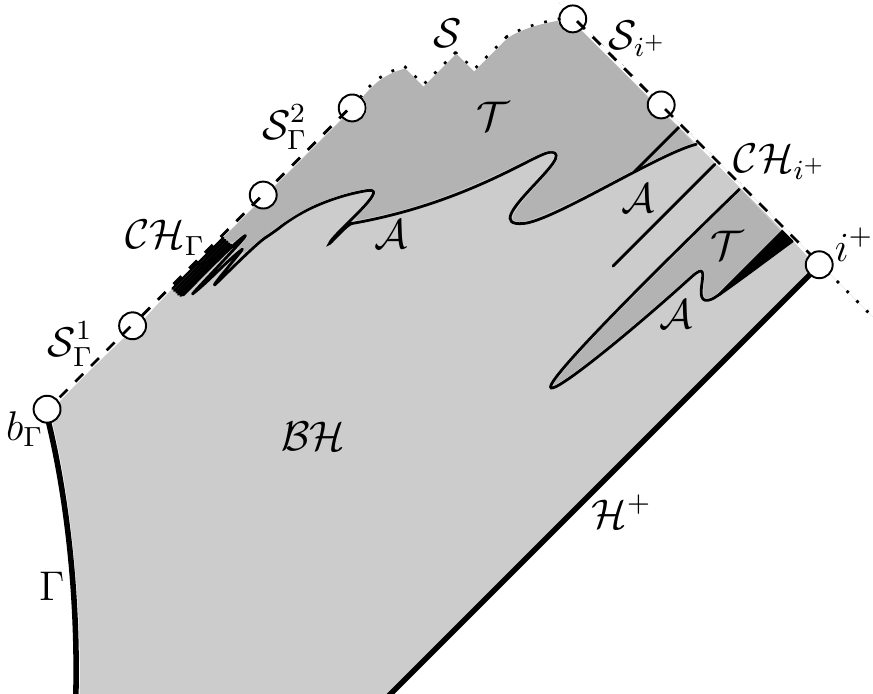}
\end{center}
\begin{itemize}
\item [1.] $\Gamma \subset \mathcal{R}$.
\item [2.] Consider $p, p'\in \Q^+$ along an outgoing null ray with $p'$ to the future of $p$. 
\begin{itemize}
\item [a.] If $p\in \mathcal{A}$, then $p'\in \mathcal{A}\cup\mathcal{T}$. 
\item [b.] If $p\in \mathcal{T}$, then $p'\in \mathcal{T}$. 
\end{itemize}
 In particular, $(\mathcal{A}\cup \mathcal{T})\cap J^-(\mathcal{I}^+)=\emptyset$.
 \item [3.] \fbox{\begin{minipage}{2.2in}
If $\sgt\cup\mathcal{S}\cup \mathcal{S}_{i^+}\neq \emptyset$, then $\mathcal{A}\cup \mathcal{T}\neq \emptyset$.\end{minipage}}
(If $\sgt\cup\mathcal{S}\cup \mathcal{S}_{i^+}=\emptyset$, then $\mathcal{A}\cup\mathcal{T}$ is possibly empty.)\item [4.] Let $p\in\mathcal{CH}_{i^+}$ that is not the future endpoint of $\mathcal{CH}_{i^+}$. If there exists a neighborhood $\mathcal{U}\subset \R^{1+1}$ of $p$ such that either $\mathcal{U}\cap \Q^+ \subset \mathcal{A}$ or $\mathcal{U}\cap \Q^+ \subset \mathcal{T}$, then $r$ extends continuously on $\mathcal{U}\cap\left(\Q^+ \cup \mathcal{CH}_{i^+}\right)$.
\item [5.] The apparent horizon $\mathcal{A}$ is clearly a closed set in $\Q^+$. Consider, however, the limit points of $\mathcal{A}$ on the boundary $\overline{\mathcal{Q}^+}\backslash\mathcal{Q}^+$ in the topology of $\R^{1+1}$.
\begin{itemize}
\item [a.] \fbox{\begin{minipage}{4.4in}
If $\mathcal{A}\neq \emptyset$, then all limit points of $\mathcal{A}$ that lie on the boundary $\overline{\mathcal{Q}^+}
\backslash\mathcal{Q}^+$ lie on $\mathcal{CH}_{i^+}\cup i^+$ and on a (possibly degenerate) closed, necessarily connected interval of $b_{\Gamma}\cup\sgo\cup\chg$.\end{minipage}}
\item [b.] If $\mathcal{A}\neq\emptyset$, then $\mathcal{A}$ has a limit point on $\mathcal{CH}_{i^+}\cup i^+$.
\item [c.]  
If $\mathcal{A}\neq\emptyset$ and $\mathcal{A}$ has a limit point on $\chg$, then there are no limit points of $\mathcal{A}$ on $b_{\Gamma}\cup\sgo$.
\item [d.] If $\mathcal{A}\neq\emptyset$ and $\mathcal{A}$ has a limit point on $b_{\Gamma}\cup\sgo$, then $\chg=\emptyset$. 
\item [e.] \fbox{\begin{minipage}{4in}
 If $\sgt\cup\mathcal{S}\cup \mathcal{S}_{i^+}\neq\emptyset$, then $\mathcal{A}$ has a limit point on $b_{\Gamma}\cup\sgo\cup \chg$.\end{minipage}}\end{itemize} 
\end{itemize}

\vspace{1mm}
\pagebreak
\begin{flushleft}{\textbf{\emph{V:~Properties of the Hawking mass}}}\end{flushleft}
\vspace{1mm}
Let $m$ be the Hawking mass function given by
\beqn\nonumber
1 - \frac{2m}{r} = g(\nabla r, \nabla r).
\eeqn
\begin{itemize}
\item [1.] If $p\in \mathcal{R}\cup \mathcal{A}$, then $m(p)$ is non-decreasing in the future-directed outgoing null direction and is non-increasing in the future-directed ingoing null direction.  Consequently,  $m$ extends (not necessarily continuously) to a non-increasing, non-negative function along $\mathcal{I}^+$. In particular, the final Bondi mass $M_f$ of the spacetime, defined by
\beqn\nonumber
M_f = \inf_{\mathcal{I}^+} m,
\eeqn
is a real (finite) number.
\item [2.] The following relations hold:
\beqna\nonumber
\frac{2m}{r} & < & 1 \hspace{.5cm}\emph{in}\hspace{.5cm}\mathcal{R}\\
\frac{2m}{r} & = & 1\hspace{.5cm}\emph{in}\hspace{.5cm}\mathcal{A}\nonumber\\
\frac{2m}{r} & >& 1\hspace{.5cm}\emph{in}\hspace{.5cm}\mathcal{T}.\nonumber
\eeqna
\end{itemize}
\vspace{1mm} 
\begin{flushleft}{\textbf{\emph{VI:~Penrose inequality}}}\end{flushleft}
\vspace{1mm} 
Let $\mathcal{H}^+\subset \mathcal{R}\cup \mathcal{A}$ denote the (possibly empty) half-open outgoing null segment forming the past-boundary of the (possibly empty) black hole region $\mathcal{BH}$, i.e.,
\beqn\nonumber
\mathcal{H}^+ =\left( \overline{J^-(\mathcal{I}^+)}\cap \Q^+\right) \backslash J^-(\mathcal{I}^+).
\eeqn
We note that if $\mathcal{H}^+\neq \emptyset$, then $\mathcal{H}^+$ has a past endpoint on $\Sigma\cup \Gamma$.\footnote{In the above diagram, we have depicted $\mathcal{H}^+$ such that $\mathcal{H}^+\cap \Sigma=\emptyset$, but, indeed, it may be that $\mathcal{H}^+\cap \Sigma \neq \emptyset$.}

If $\mathcal{BH}\neq \emptyset$, then $\mathcal{H}^+\neq \emptyset$ and, moreover, the following inequality holds: 
\beqn\nonumber
\sup_{\mathcal{H}^+}r \leq 2 M_f.
\eeqn
In particular, if $\mathcal{BH}\neq \emptyset$, then $M_f >0$.
\vspace{1mm} 
\begin{flushleft}{\textbf{\emph{VII:~Extendibility of the solution}}}\end{flushleft}
\vspace{1mm} 
\begin{itemize}
\item [1.] \fbox{\begin{minipage}{4.9in}
The Kretschmann scalar $R_{\mu\nu\alpha\beta}R^{\mu\nu\alpha\beta}$ is a continuous $[0,\infty]$-valued function on $\Q^+\cup \sgt\cup\mathcal{S}\cup \mathcal{S}_{i^+}$ that yields $\infty$ on $\sgt\cup\mathcal{S}\cup \mathcal{S}_{i^+}$. The rate of blow-up is no slower than $r^{-4}$.
\end{minipage}}
\item [2.] Let $p\in \sgo$ and consider a neighborhood $\mathcal{U}\subset \R^{1+1}$ of $p$.  
\begin{itemize}
\item [a.] There exists a sequence $\{p_j\}_{j=1}^{\infty}\subset \mathcal{U}\cap\Q^+$ with $p_j \rightarrow p$ such that
\beqn\nonumber
 \limsup_{j\rightarrow \infty} R_{\mu\nu\alpha\beta}R^{\mu\nu\alpha\beta}(p_j) = \infty.
\eeqn
The rate of blow-up is no slower than $r^{-4+\epsilon}$, for some $\epsilon>0$.
\item [b.]If $\mathcal{U}\cap\Q^+ \subset\mathcal{A}\cup \mathcal{T}$, then the Kretschmann scalar is a continuous $[0,\infty]$-valued function on $\mathcal{U}\cap\left(\Q^+\cup \sgo\right)$ that yields $\infty$ on $\sgo\cap \mathcal{U}$. The rate of blow-up is no slower than $r^{-4}$.
\end{itemize}
\item [3.] \fbox{\begin{minipage}{4.2in}
Let $m$ be the Hawking mass. If $\sup_{\mathcal{H}^+} m \neq \frac{1}{2}\sup_{\mathcal{H}^+} r$,
 then $\mathcal{CH}_{i^+}\neq \emptyset$.
\end{minipage}}
\item [4.] \fbox{\begin{minipage}{4.9in}
If $(\mathcal{M},g_{\mu\nu})$ is future-extendible as a $C^2$-Lorentzian manifold $(\widetilde{\mathcal{M}}, \widetilde{g_{\mu\nu}})$, then there exists a timelike curve $\gamma\subset \widetilde{\M}$ exiting the manifold $\mathcal{M}$ such that the closure of the projection of $\gamma|_{\M}$ to $\Q^+$ intersects $\chg\cup \mathcal{CH}_{i^+}$.   In particular, if $\chg \cup \mathcal{CH}_{i^+} = \emptyset$, then  $(\mathcal{M}, g_{\mu\nu})$ is $C^2$-future-inextendible.
\end{minipage}}\end{itemize}
\end{thm}

\subsection{Weak and strong cosmic censorship}\label{sec:intro/ws}

We will discuss known (limited) results concerning cosmic censorship for the Einstein-Maxwell-Klein-Gordon system in \S\ref{sec:intro/cd} and establish more conjectures in \S\ref{sec:intro/conj} that, if true, will imply, in particular, cosmic censorship.  To aid this study, it will be convenient to include here concise statements of the cosmic censorship conjectures (following Christodoulou in \cite{DC98})  under the framework presented by Theorem \ref{thm:main}.

\begin{conj}[\textbf{Weak cosmic censorship}]\label{conj:weak} Among all the data admissible from Theorem \ref{thm:main}, there exists a generic sub-class for which $\mathcal{I}^+$ is complete.
\end{conj}

\begin{conj}[\textbf{Strong cosmic censorship}]\label{conj:strong} Among all the data admissible from Theorem \ref{thm:main}, there exists a generic sub-class for which the maximal future development is future-inextendible as a suitably regular Lorentzian metric.
\end{conj}

Despite the historical nomenclature, there is no logical relationship between weak and strong cosmic censorship.  This should not be surprising if one thinks of weak cosmic censorship, in the language of PDEs, as ultimately a statement of global existence and strong cosmic censorship as a statement of global uniqueness (and, of course, existence and uniqueness are \emph{a priori} unrelated issues).  In this regard, one can perhaps better appreciate the importance of regularity in Conjecture \ref{conj:strong}, which will be discussed in  \S\ref{sec:intro/cd/d}, \S\ref{sec:intro/cd/d/mass}, and \S\ref{sec:intro/conj/strong}.      

\subsection{Models of Christodoulou and Dafermos}\label{sec:intro/cd}

Contained within Theorem \ref{thm:main} is the self-gravitating real-valued massless scalar field model of Christodoulou \cite{DC91}; this corresponds to taking $\m^2=\e = \imp = F_{\mu\nu}=0$.  The model of Dafermos \cite{MD03}, i.e.,~the model for which $\m^2=\e =\imp=0$, but $F_{\mu\nu}$ is not assumed to vanish, is not, however, included in the statement of Theorem \ref{thm:main} in view of the topology of the initial data.  If we impose that $\Sigma$ has one asymptotically flat end and $\m^2=\e =\imp=0$, then the model of Dafermos, necessarily, reduces to that of Christodoulou, i.e.,~it follows that $F_{\mu\nu}=0$.

\subsubsection{Christodoulou: the real-valued massless scalar field} \label{sec:intro/cd/c}
In the case $\m^2=\e=\imp=F_{\mu\nu}=0$, the system (\ref{RMN})--(\ref{eqn:kg}) exhibits stronger monotonicity properties (above and beyond Raychaudhuri; cf.~the Einstein equation (\ref{eqn:ruv})) not present in the more general case.  In particular, we can strengthen the boundary characterization of Theorem \ref{thm:main} as follows:
 
\begin{itemize}
\item [1.] $\mathcal{S}_{i^+} \cup \mathcal{CH}_{i^+}=\emptyset$.
\item [2.] $\sgt= \emptyset$ and $\mathcal{S}_{\Gamma} := \sgo$.
\item [3.] $\mathcal{S}$ is $C^1$-spacelike.
\end{itemize}
In all, there are eight possible spacetimes as depicted\footnote{We do not differentiate between the cases in which the past endpoint of $\mathcal{H}^+$ intersects $\Gamma \backslash \Sigma$, $\Sigma\backslash \Gamma$, or $\Sigma\cap \Gamma$.  This is related, however, to the important issue of dynamical formation of black holes.} below in diagrams I--VIII.

\begin{figure}[h!]
\vspace{-.25cm}
   \subfloat[\scriptsize dispersive]{\includegraphics[scale=.95]{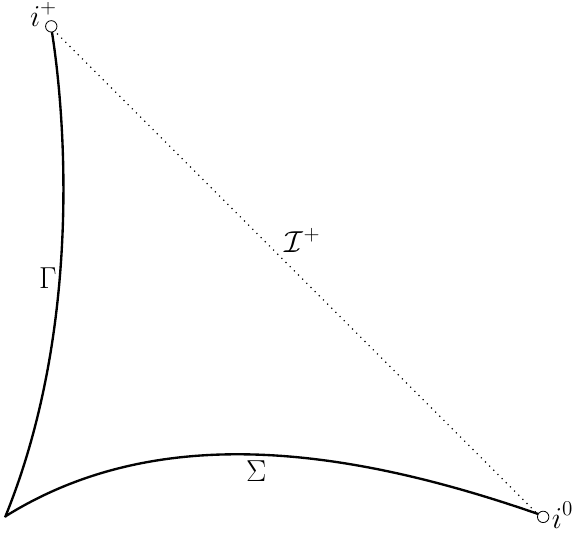}}  \quad
   \centering
  \subfloat[\scriptsize black hole with spacelike singularity]{\includegraphics[scale=.75]{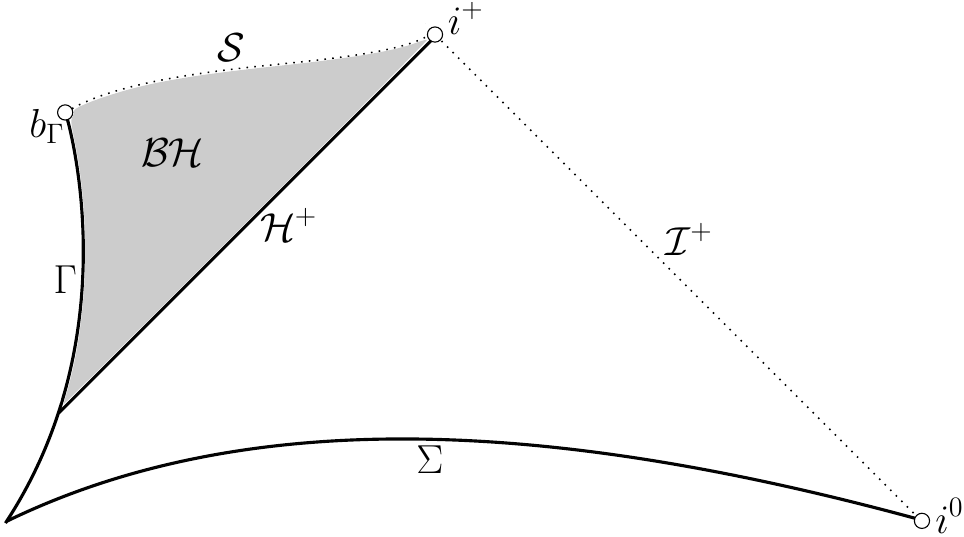}}
\end{figure}

\begin{figure}[h!]
\vspace{-.25cm}
\setcounter{subfigure}{2}
\centering
  \subfloat[\scriptsize light cone singularity]{\includegraphics[scale=.9]{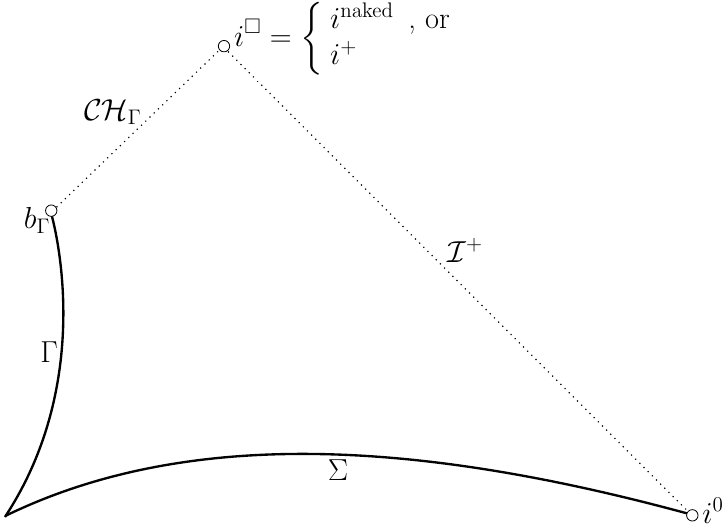}}
    \subfloat[ \scriptsize black hole with light cone singularity]{\includegraphics[scale=.7]{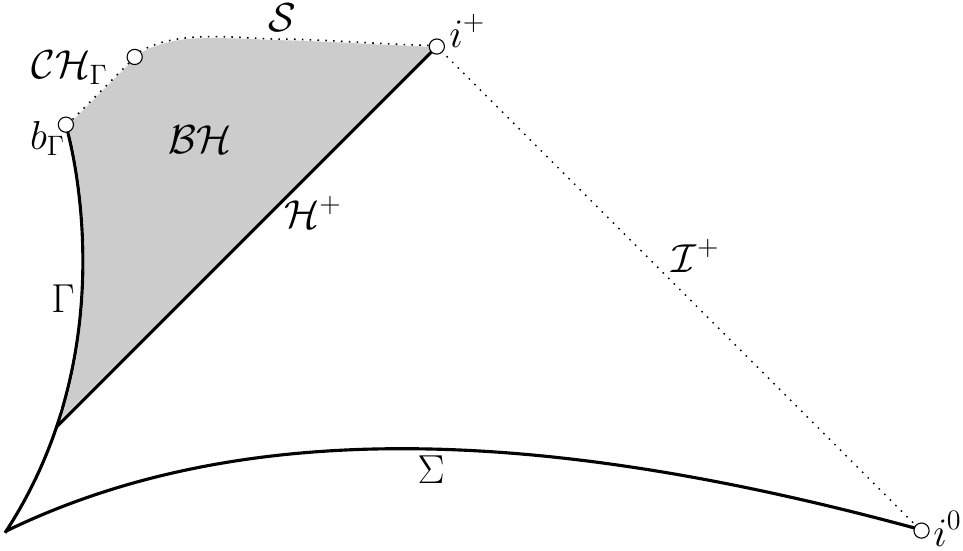}}     
\end{figure}
\begin{figure}[h!]
\vspace{-.25cm}
\setcounter{subfigure}{4}
\subfloat[\scriptsize collapsed light cone singularity]{\includegraphics[scale=.7]{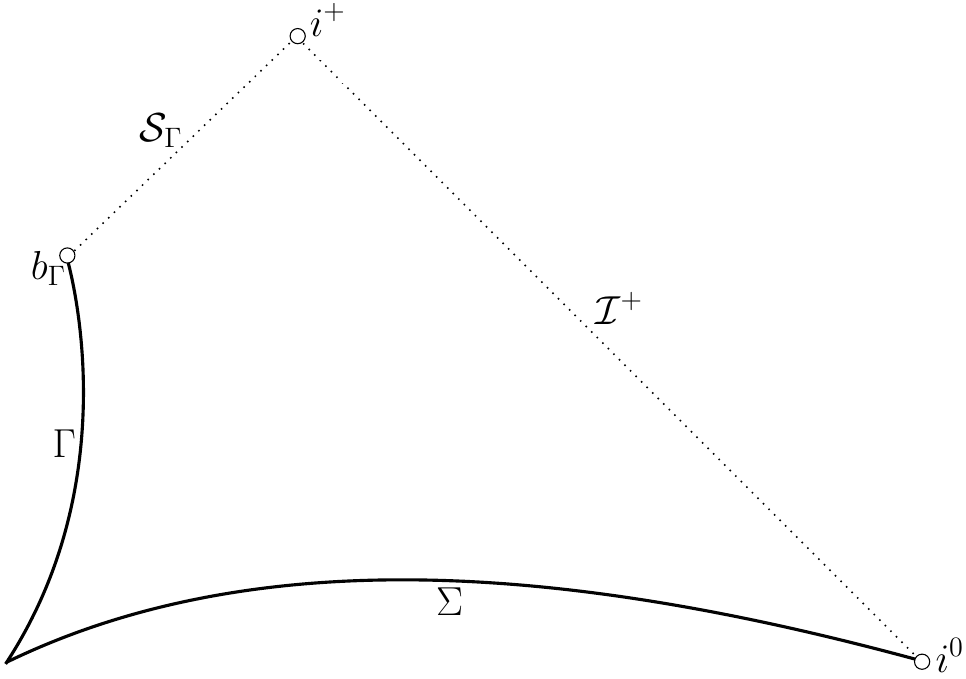}}
   \centering
  \subfloat[\scriptsize black hole with collapsed light cone singularity]{\includegraphics[scale=.7]{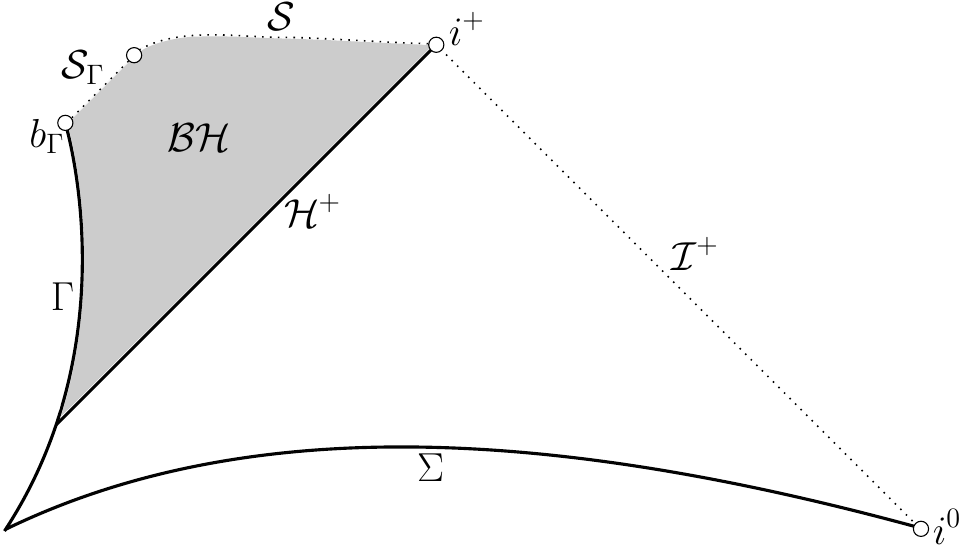}}
\end{figure}

\begin{figure}[h]
\vspace{-.25cm}
\setcounter{subfigure}{6}
 \subfloat[\scriptsize expanding collapsed light cone singularity]{\includegraphics[scale=1.05]{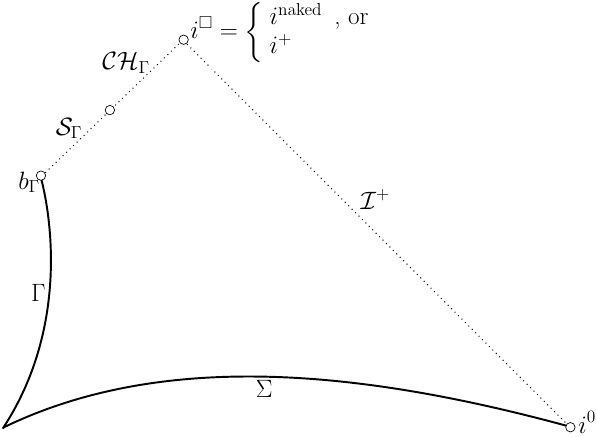}}\quad
 \centering
  \subfloat[\scriptsize black hole with expanding collapsed light cone singularity]{\includegraphics[scale=.38]{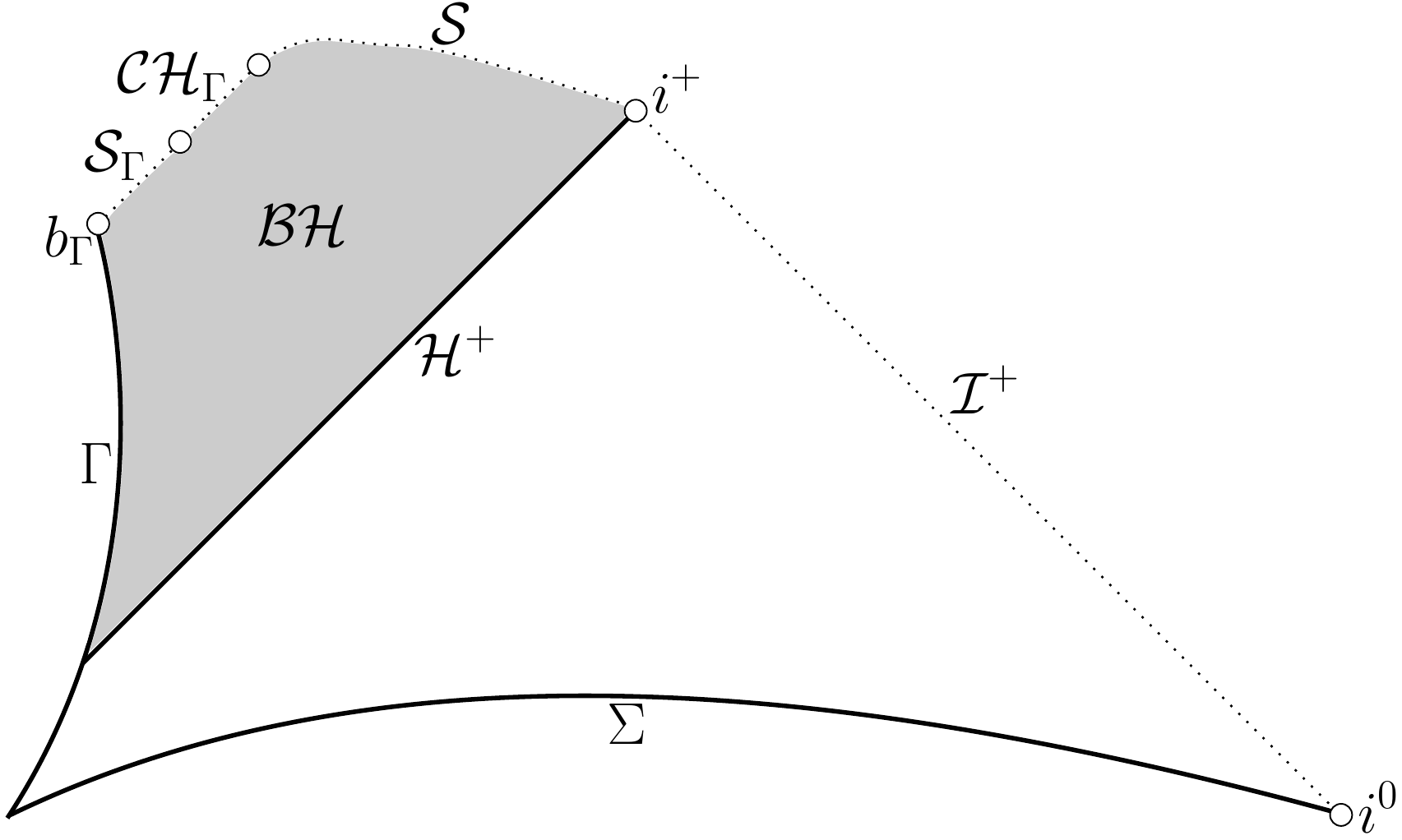}}
\end{figure}

Non-black hole (expanding collapsed) light cone singularity spacetimes, as in diagrams III and VII, \emph{a priori} may, or may not, have a complete future null infinity $\mathcal{I}^+$ and hence, may, or may not, have (in our convention) a `future timelike infinity' $i^+$.  Moreover, such spacetimes may, or may not, be future-extendible beyond $\mathcal{CH}_{\Gamma}$.  In principle, there may exist, in particular, a spacetime as in diagram III and VII where $i^{\square} = i^{\textrm{naked}}$, but for which the solution is future-inextendible. This illustrates why strong cosmic censorship does \emph{not} imply weak cosmic censorship.     

In \cite{DC94}, Christodoulou constructs, in particular, spacetimes as in diagrams III and VII with $\mathcal{I}^+$ incomplete (a non-black hole (expanding collapsed) light cone singularity with $i^{\square} = i^{\textrm{naked}}$, which we, henceforth, call a `naked singularity')  but for which the spacetime is also future-extendible. This demonstrates the necessity of having a genericity assumption in the formulation of Conjectures \ref{conj:weak} and \ref{conj:strong}. We note that the solutions constructed within the context of \cite{DC94} are not smooth but, nonetheless, lie in a `BV' class for which strong well-posedness can still be proven (see the discussion in \S\ref{sec:intro/general/regularity}). They are thus, in every sense, strong solutions.  

In his seminal work \cite{DC99}, Christodoulou shows that the set of solutions given by diagrams III--VIII are, in a suitable sense, `exceptional' (in particular, non-generic), as they form a set of  positive co-dimension in the family of all solutions as above.  This is summarized in
\pagebreak
\begin{thm}[Christodoulou \cite{DC99}]\label{thm:christo}For initial data as in Theorem \ref{thm:main} in the more general BV class with $\m^2=\e=\imp=F_{\mu\nu}=0$, there exists a sub-class for which generically $\mathcal{S}_{\Gamma}\cup \chg=\emptyset$ and the spacetimes are inextendible as $C^0$-Lorentzian metrics.  In particular, Conjectures \ref{conj:weak} and \ref{conj:strong} are \emph{true} in the case of a self-gravitating  real-valued massless scalar field and generic spacetimes are as depicted in diagrams I and II above.\footnote{With respect to Conjecture \ref{conj:strong}, $C^2$-inextendibility of the spacetime would follow by Statement VII of Theorem \ref{thm:main}, but the regularity class considered by Christodoulou in \cite{DC93} is below $C^2$ for the metric, hence the desirability of the stronger $C^0$-formulation, which Christodoulou, indeed, obtains by a separate argument.  For a general Einstein-Maxwell-Klein-Gordon spacetime, however, it is conjectured (v.~Conjecture \ref{conj:nonemptyext}) that such a strong formulation of Conjecture \ref{conj:strong} will not hold (cf.~Theorem \ref{thm:c0strong}\label{foot:christo}).}\end{thm}

For a discussion of the significance of the positive resolution of cosmic censorship for the real-valued massless scalar field in the context of other models, in particular Einstein-dust, see \S\ref{sec:intro/general/wcc_dust}.

In proving Theorem \ref{thm:christo}, Christodoulou makes use of the following result, which is also of independent interest.
 
\begin{thm}[Christodoulou \cite{DC91}]\label{thm:christo_trap_form} Let $(\M = \Q^+\times_r \mathbb{S}^2, g_{\mu\nu}, \phi, F_{\mu\nu})$ be the development of initial data as in Theorem \ref{thm:main} with $\m^2=\e=\imp=F_{\mu\nu}=0$.  For $p, p'\in \mathcal{R}$ along an outgoing null ray $C^+_0$ with $p'$ to the future of $p$, suppose the ingoing null ray $C^-_p$ that emanates from $p$ terminates on $q\in\Gamma \cup b_{\Gamma}$.  Let $\delta_0$ and $\eta_0$ be defined by\footnote{The constants $\delta_0$ and $\eta_0$ give the dimensionless size and the dimensionless mass content, respectively, of the enclosed annular region bounded by $p$ and $p'$.}
\beqn\nonumber
\delta_0 = \frac{r(p')}{r(p)}-1\hspace{.75cm}\textrm{and}\hspace{.75cm} \eta_0 = \frac{2\left(m(p')-m(p)\right)}{r(p')},
\eeqn
where $m$ is the Hawking mass function.

There are positive constants $c_1$ and $c_2$ such that if $\delta_0\leq c_1$ and 
\beqn\nonumber
\eta_0 > c_2 \delta_0 \log\left(\delta_0^{-1}\right),
\eeqn
then the region $\mathcal{X}\subset \Q^+$ given by
\beqn\nonumber
\mathcal{X} =\left( J^+(p)\cap \Q^+\right) \backslash \left(J^+(q) \cup J^+(p')\right)
\eeqn
contains a trapped surface $p_*\in \mathcal{X}$ as depicted below.
\end{thm}
\begin{center}
\includegraphics[scale=.9]{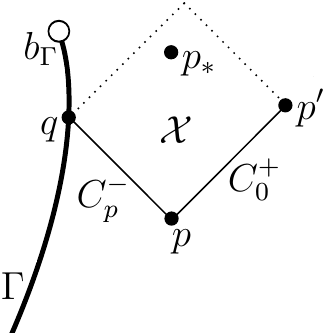}
\end{center}

Christodoulou applies Theorem \ref{thm:christo_trap_form} as an auxiliary lemma in the context of the proof of Theorem \ref{thm:christo}.  One begins with a spacetime as given by diagrams III--VIII, and the goal is to produce a 1-parameter family of spacetimes containing the given one such that \emph{all} other members of the family have $\mathcal{A}\neq \emptyset$ with limit point $b_{\Gamma}\neq i^+$.  The infinite blue-shift along $C^-_p$ plays an important role in the proof of Theorem \ref{thm:christo}, for it provides the linear mechanism for instability.\footnote{Once this property of $\mathcal{A}$ is established, the emptiness of $\mathcal{S}_{\Gamma}\cup \chg$ is a consequence of the special monotonicity $\partial_u\pv r<0$ in the trapped region.}  Using this effect, it is shown that for the perturbed spacetimes, the assumptions of Theorem \ref{thm:christo_trap_form} hold with $q=b_{\Gamma}\neq i^+$ and a sequence of $p, p'\rightarrow q$.
\begin{center}
\includegraphics[scale=.90]{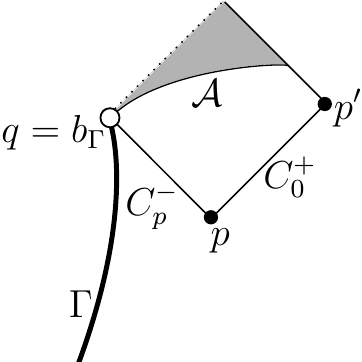}
\end{center}
Thus, Theorem \ref{thm:christo_trap_form} applies to yield $\mathcal{A}$ as desired.
It is interesting that in Christodoulou's construction the 1-parameter family of perturbations coincide with the original spacetime in the past of $q=b_{\Gamma}$.
\subsubsection{Dafermos: the real-valued massless scalar field with topological charge}\label{sec:intro/cd/d}
Dafermos considers the model for which 
$\m^2=\e=\imp=0$ and $F_{\mu\nu}\neq 0$.  Since the scalar field is itself uncharged, $F_{\mu\nu}$ can be non-trivial only if the Cauchy surface $\Sigma$ has two asymptotically flat ends. In this case, however, the electromagnetic field is only `coupled' to the scalar field via its interaction with the geometry. 
 
An analogue\footnote{In this model note that there are, in general, anti-trapped regions.  To prove the analogue of Theorem \ref{thm:main}, it suffices to assume that there exists a point $(u',v')\in \Sigma$ such that $\pv r <0$ in $\Sigma \cap \{v\leq v'\}$ and $\pu r <0$ in $\Sigma \cap \{v\geq v'\}$.\label{ft:dafext}} of Theorem \ref{thm:main}, applied to this class of initial data, yields a Penrose diagram as depicted below.
\begin{center}
\includegraphics[scale=.5]{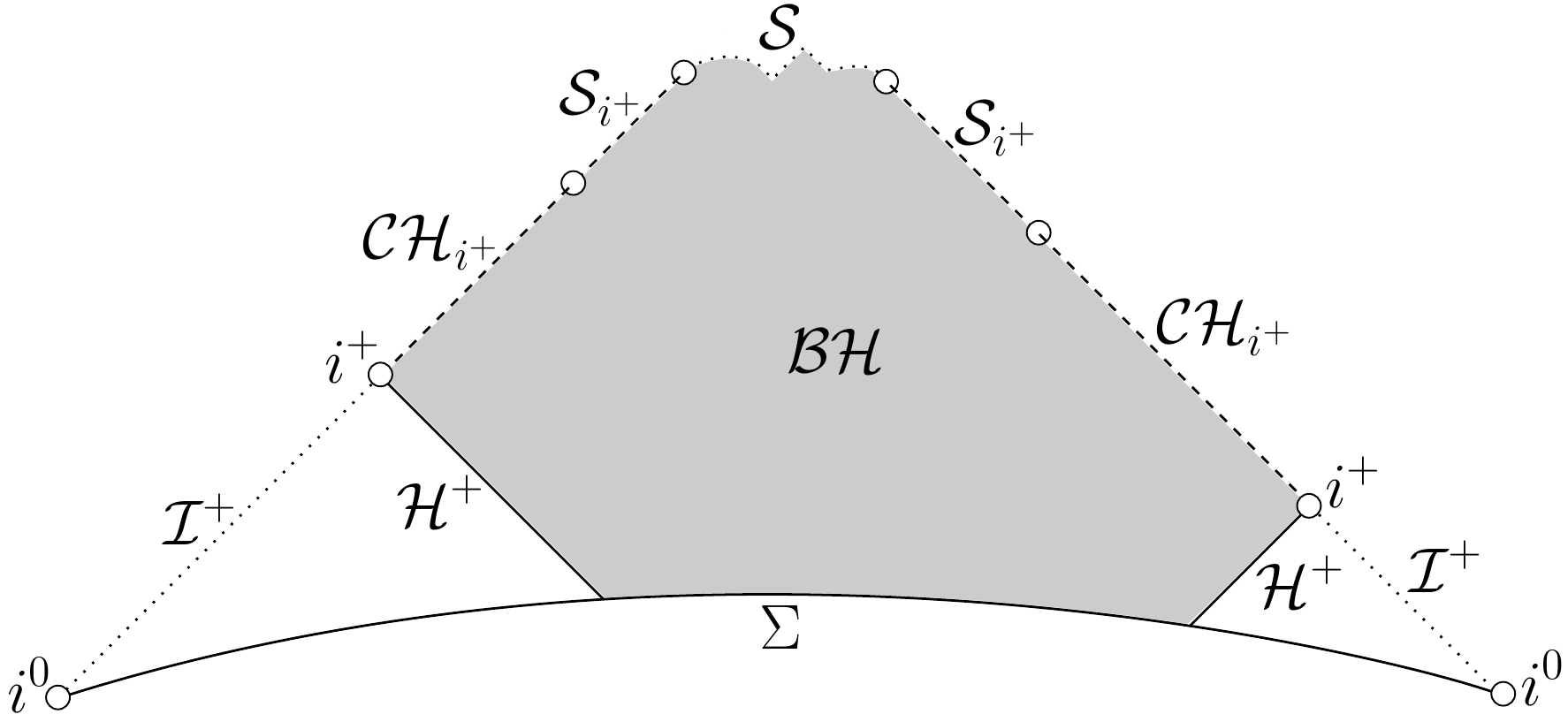}
\end{center}

One easily infers that for a spacetime having two asymptotically flat ends, the black hole region $\mathcal{BH}$ is necessarily non-empty and, therefore, by the analogue of Theorem \ref{thm:main}, both connected components of $\mathcal{I}^+$ are complete.  Thus, weak cosmic censorship is trivially true but not very physically interesting.  On the other hand, this model is well-suited for addressing strong cosmic censorship in a non-trivial context because it admits as a special solution the Reissner-Nordstr\"om family, with mass parameter $M$ and charge parameter $Q$, where, if $0< |Q| < M$, then $\mathcal{CH}_{i^+}$ is non-empty and the maximal future development is future-extendible as a smooth Lorentzian metric.  Thus, for strong cosmic censorship to be true, the Reissner-Nordstr\"om solution, in particular, must be shown to be `unstable'.

In considering this issue of stability, Dafermos shows, however, that whenever a black hole is `sub-extremal in the limit' and the black hole charge is non-vanishing, then $\mathcal{CH}_{i^+}$ is non-empty and the maximal future development is continuously extendible \cite{MD05c}.  Indeed, these assumptions can be shown to hold for solutions arising from arbitrary data in a suitable `open neighborhood' of Reissner-Nordstr\"om initial data; in particular, the spherically symmetric $C^0$-formulation\footnote{where in the analogue of Conjecture \ref{conj:strong}, `suitably regular' means `continuous'} of strong cosmic censorship is \emph{false}!  

Before presenting this result, it will be convenient to discuss asymptotic parameters of black hole solutions, i.e.,~solutions with $\mathcal{BH}\neq \emptyset$, arising when $\m^2=\e=\imp=0$, namely: area-radius, mass, and charge.

The asymptotic area-radius $r_+$ of the black hole (as measured along $\mathcal{H}^+$), given by
\beqn\nonumber
r_+ = \sup_{\mathcal{H}^+} r,
\eeqn 
is well-defined by monotonicity and is finite by Statement VI of the analogue of Theorem \ref{thm:main}. 
Similarly by monotonicity, the asymptotic mass $m_+$ of the black hole (as measured along $\mathcal{H}^+$), given by
\beqn\nonumber
m_+ = \sup_{\mathcal{H}^+} m,
\eeqn
where $m$ is the Hawking mass function, is well-defined and finite.\footnote{Indeed, since $\mathcal{H}^+\subset \mathcal{R}\cup \mathcal{A}$, one has $m_+\leq \frac{1}{2}r_+$.}  

In the case $\imp=0$ (or, more generally, $\e=0$), the scalar invariant $F_{\mu\nu}F^{\mu\nu}$ is given by\footnote{In the case $\e\neq 0$, see \S\ref{energymomentum}.}
\beqn\nonumber
F_{\mu\nu}F^{\mu\nu} = -\frac{2}{r^4}\left(\Qe^2 - \Qm^2\right)
\eeqn 
for constants $\Qe, \Qm\in \R$.  The constant $Q$ such that $Q^2 = \Qe^2 + \Qm^2$, defines, in particular, the asymptotic charge\footnote{Because $\e = 0$, this can be taken to mean, without loss of generality, `as measured along $\mathcal{H}^+$', since $Q$ is \emph{globally} constant (cf.~footnote \ref{foot:subext}).} of the black hole.

For convenience, we also define the asymptotic re-normalized mass $\varpi_+$ by
\beqn\nonumber
\varpi_+ = \sup_{\mathcal{H}^+}\varpi := \sup_{\mathcal{H}^+} \left(m + \frac{Q^2}{2r}\right) = m_+ + \frac{Q^2}{2r_+}.
\eeqn
In the case of Reissner-Nordstr\"om, $\varpi = \varpi_+ = M$.

We now state

\begin{thm}[Dafermos \cite{MD05c}]\label{thm:c0strong}
Let $(\M = \Q^+\times_r \mathbb{S}^2, g_{\mu\nu}, \phi, F_{\mu\nu})$ denote the maximal future development of compactly supported smooth spherically symmetric asymptotically flat initial data with two ends for the Einstein-Maxwell-Klein-Gordon system (for which the analogue of Theorem \ref{thm:main} holds; cf.~footnote \ref{ft:dafext}) with $\m^2=\e=\imp=0$ such that
\beqn\label{subext}
 0<|Q |< \varpi_+.
 \eeqn
Then, $\mathcal{CH}_{i^+}\neq \emptyset$.  Moreover, $(\M, g_{\mu\nu})$ is future-extendible as a $C^0$-Lorentzian manifold $(\widetilde{\M}, \widetilde{g_{\mu\nu}})$, which can be taken to be spherically symmetric, and there exists continuous functions $\widetilde{\phi}$ and $\widetilde{F_{\mu\nu}}$ defined on $\widetilde{\M}$ such that $\widetilde{\phi}$ and $\widetilde{F_{\mu\nu}}$ restricted to $\M$ coincide with $\phi$ and $F_{\mu\nu}$.
  In fact, the conclusions hold for solutions arising from arbitrary initial data in a suitable open neighborhood of Reissner-Nordstr\"om initial data.  In particular, the spherically symmetric $C^0$-formulation of strong cosmic censorship is \emph{false}.
\end{thm}

A solution of Theorem \ref{thm:c0strong} has a Penrose diagram that admits an extension as depicted below.
\begin{center}
\includegraphics[scale=.53]{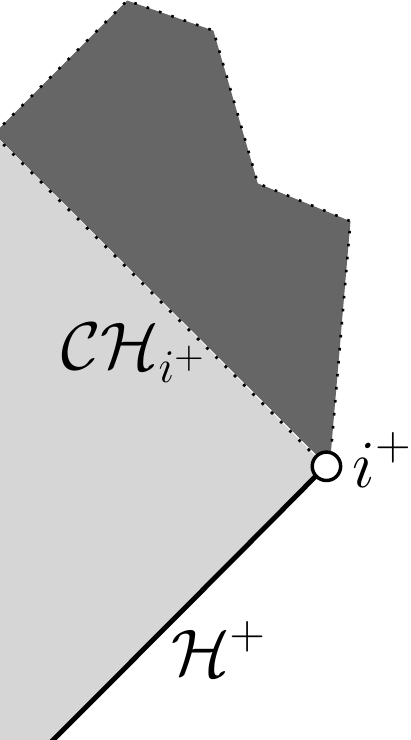}
\end{center}

To prove Theorem \ref{thm:c0strong}, Dafermos relies heavily on the decay properties of the scalar field along $\mathcal{H}^+$. This decay will be discussed in \S \ref{sec:intro/cd/d/price}.  

In order to highlight the importance of trapped surface formation to this discussion, we note that Dafermos also deduces the existence of a non-empty `asymptotically connected' component of the outermost apparent horizon\footnote{The outermost apparent horizon $\mathcal{A}'\subset \mathcal{A}$ is a (possibly empty, not necessarily connected) achronal curve defined by the set of all $p\in \mathcal{A}$ whose past-directed ingoing null segment in $\Q^+$ lies in the regular region $\mathcal{R}$ with at least one $q\in\mathcal{R}\cap J^-(\mathcal{I}^+)$. \label{foot:outer} } $\mathcal{A}'$ that terminates at $i^+$ (cf.~Williams \cite{CW08}).  This is given in
 
\begin{thm}[Dafermos \cite{MD05c}]\label{thm:trapped_i+}Let $(\M=\Q^+\times_r \mathbb{S}^2, g_{\mu\nu}, \phi, F_{\mu\nu})$ denote the maximal future development of initial data as in Theorem \ref{thm:c0strong}.  Then, there exists a non-empty `asymptotically connected' component of the outermost apparent horizon $\mathcal{A}'$ that terminates at $i^+$.  Moreover, in a sufficiently small neighborhood $\mathcal{U}\subset \R^{1+1}$ of $i^+$
\beqn\nonumber
 \mathcal{A}'\cap \mathcal{U} = \mathcal{A}\cap \mathcal{U}
\eeqn
and, in particular,
\beqn\nonumber
I^+(\mathcal{A}\cap \mathcal{U})\cap \Q^+ \subset \mathcal{T}.
\eeqn

\end{thm}
\begin{center}
\includegraphics[scale=.90]{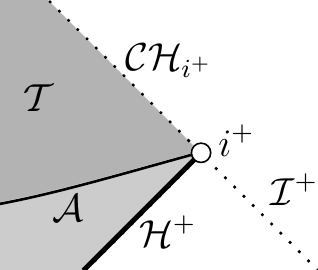}
\end{center} 

To prove Theorem \ref{thm:c0strong}, it is necessary to first establish Theorem \ref{thm:trapped_i+}.  Although the role of trapped surface formation is very different, this should be reminiscent of Theorems \ref{thm:christo} and  \ref{thm:christo_trap_form}: Deducing that $\mathcal{A}$ has a limit point on $i^+$ is necessary to prove the \emph{stability} of the Cauchy horizon $\mathcal{CH}_{i^+}$, as opposed to deducing that $\mathcal{A}$ has a limit point on $b_{\Gamma}$ to prove the \emph{instability} of the central Cauchy horizon $\mathcal{CH}_{\Gamma}$. Of course, the model of Dafermos does not admit central Cauchy horizons, nor does the model of Christodoulou admit $\mathcal{CH}_{i^+}$, but the analogy is interesting.  Within the context of the more general Einstein-Maxwell-Klein-Gordon system, this tantalizing behavior, linking both cosmic censorship conjectures to trapped surface formation, can be further explored since both types of Cauchy horizons can be admitted. Indeed, an analogue of Theorem \ref{thm:trapped_i+} has already been shown \cite{JK10c} by the author when the topology of the initial data has one asymptotically flat end.

\subsubsection{`No-hair theorem' and Price's law}\label{sec:intro/cd/d/price}  

In the study of gravitational collapse, one may ask: What are the possible `end-states' of evolution?

So-called `no-hair theorems', e.g.,~as given, in the present context, by Mayo and Bekenstein in \cite{AMJB96}, assert that if a spherically symmetric Einstein-Maxwell-Klein-Gordon black hole spacetime $(\M, g_{\mu\nu})$ is, in addition, stationary, i.e.,~the spacetime admits a Killing vector field that is asymptotically timelike in a neighborhood of $\mathcal{I}^+$, then $(\M, g_{\mu\nu})$ is a member of the Reissner-Nordstr\"om family.  

For dynamic spacetimes as given in Theorem \ref{thm:main}, if $\mathcal{BH}\neq\emptyset$ and the exterior geometry `settles down' so as to give rise to a black hole spacetime that is asymptotically stationary as $i^+$ is approached, then the above `no-hair theorem' suggests that the spacetime approaches Reissner-Nordstr\"om. The quantitative study of this decay (`settle down') mechanism is associated with the name of Price.

Formulated in \cite{RP72}, Price postulates that (massless) gravitational radiation decays polynomially with respect to the (asymptotically stationary) time co-ordinate along timelike surfaces of constant $r$.  Later, the work of Gundlach et al.~\cite{CGRPJP94} refined the heuristics so as to postulate that (massless) gravitational flux along the event horizon (resp.,~future null infinity) will have polynomial decay with respect to a suitable advanced (resp.,~retarded) time co-ordinate. In and of itself a major open problem, this decay mechanism, which we call here Price's law, is rigorously established by Dafermos and Rodnianski in the case $\m^2=\e=\imp=0$, provided that the black hole is `sub-extremal in the limit' \cite{MDIR05}.  This is summarized in
\begin{thm}[Dafermos and Rodnianski \cite{MDIR05}] \label{thm:price}Let $(\M=\Q^+\times_r \mathbb{S}^2, g_{\mu\nu}, \phi, F_{\mu\nu})$ denote the maximal future development of compactly supported smooth spherically symmetric asymptotically flat  initial data for the Einstein-Maxwell-Klein-Gordon system as in Theorem \ref{thm:c0strong} or Theorem \ref{thm:main} with $\m^2=\e=\imp=F_{\mu\nu}=0$.  Assume that $\mathcal{Q}^+\backslash J^-(\mathcal{I^+})\neq \emptyset$.  If 
\beqn\label{subex2}
0\leq |Q| <\varpi_+, 
\eeqn
then for all $\epsilon >0$ there is a positive constant $C_{\epsilon}<\infty$ such that, for a suitable normalized advanced time co-ordinate $v$, 
\beqn\label{rate}
|r\phi| + |r\pv \phi|\leq C_{\epsilon} v^{-3+\epsilon}
\eeqn
along $\mathcal{H}^+$.\footnote{Decay is also established along $\mathcal{I}^+$ and timelike curves of constant $r$, but only the decay along $\mathcal{H}^+$ is directly relevant for cosmic censorship.  We shall, therefore, only make reference to decay along $\mathcal{H}^+$ in what follows.} 
\end{thm}
We remark that the `sub-extremal in the limit' condition (\ref{subex2}) is satisfied for all black hole solutions arising in the model of Christodoulou.\footnote{One can deduce this \emph{a posteriori} from the statement of Theorem \ref{thm:christo}.  Since, in this case, $\mathcal{CH}_{i^+}=\emptyset$, if $\mathcal{BH}\neq \emptyset$, then, necessarily, $\mathcal{S}\neq \emptyset$, whence $\mathcal{A}\neq \emptyset$ and $\varpi_+ \geq \inf_{\mathcal{A}} \frac{1}{2}r >0 =Q$.  In establishing Theorem \ref{thm:christo}, however, (\ref{subex2}), i.e.,~$\varpi_+>0$, must first be shown (cf.~\cite{DC87}).}

A generalization of Price's law to the case $\m^2\neq 0$ and $\e\neq 0$ will be discussed in \S \ref{sec:intro/conj/price}.

\subsubsection{`Mass inflation' and strong cosmic censorship}\label{sec:intro/cd/d/mass}

While Theorem \ref{thm:price} gives an upper bound for the decay of a real-valued massless scalar field,
heuristic analysis \cite{LB99, LBAO99, CGRPJP94b, RP72} and numerical studies \cite{BO97, CGRPJP94} suggest that generically there is a similar lower bound.  In fact, the existence of such a generic lower bound may yet, in light of Theorem \ref{thm:c0strong}, prove significant in redeeming the validity of strong cosmic censorship, for Dafermos shows that if, indeed, such a lower bound for decay holds along $\mathcal{H}^+$ for any `sub-extremal in the limit' black hole, then the curvature must blow up along $\mathcal{CH}_{i^+}$ \cite{MD05c}.  This provides mathematical confirmation of the `mass-inflation' scenario of Israel and Poisson \cite{EPWI90}.  This is given in

\begin{thm}[Dafermos \cite{MD05c}] \label{thm:mass}
Let $(\M=\Q^+\times_r \mathbb{S}^2, g_{\mu\nu}, \phi, F_{\mu\nu})$ be as in Theorem \ref {thm:c0strong}.
Suppose along $\mathcal{H}^+$ the scalar field satisfies for some $\epsilon>0$ and some positive constant $C<\infty$
\beqn\label{lowerbound}
|r\pv \phi|\geq Cv^{-9+\epsilon} \hspace{1cm} \textrm{for all}\hspace{.5cm} v\geq V,
\eeqn
where $v$ is a suitable normalized advanced time co-ordinate. Then, the curvature blows up along $\mathcal{CH}_{i^+}$.
\end{thm}
If (\ref{lowerbound}) can be shown to hold for generic initial data, then the spherically symmetric $C^2$-formulation of strong cosmic censorship is \emph{true}!  We state this in
 
\begin{cor}\label{cor:d_strong}Let $(\M=\Q^+\times_r \mathbb{S}^2, g_{\mu\nu}, \phi, F_{\mu\nu})$ be as in Theorem \ref {thm:mass}. If (\ref{lowerbound}) holds for generic initial data, then, in particular, the spherically symmetric $C^2$-formulation of strong cosmic censorship is \emph{true}.
\end{cor}

\subsection{Conjectures}\label{sec:intro/conj}
Generalizing the results of Christodoulou and Dafermos to the full Einstein-Maxwell-Klein-Gordon system is, needless to say, no easy task.  We discuss a few conjectures here to put forthcoming results, announced in \S \ref{sec:intro/forth}, into the proper context.  For convenience, we will here formulate all conjectures in the context of smooth developments as in Theorem \ref{thm:main}.  This being said, experience with the model of Christodoulou (cf.~\S \ref{sec:intro/cd/c}) indicates that it may be more natural to consider a larger class of solutions.  The reader may wish to replace the smooth initial data and their maximal development in the statement of the conjectures with less regular initial data and their maximal development for which the conclusion of Theorem \ref{thm:main} can still be shown.  See also the discussion of regularity in \S \ref{sec:intro/general/regularity}. 

\subsubsection{`Sub-extremal in the limit' black holes}\label{sec:intro/conj/geometryH}
The Einstein-Maxwell-Klein-Gordon system admits `extremal' black hole solutions.  In \cite{SA11a, SA11b}, Aretakis shows that the wave equation exhibits both stability and instability properties on extremal Reissner-Nordstr\"om horizon geometries, suggesting that the analysis required to deal with the class of `extremal' black hole spacetimes will be delicate.  That having been said, the following conjecture would imply that `extremal' black hole solutions are non-generic and thus, in particular, one can ignore them in the context of the study of cosmic censorship.

\begin{conj}[\textbf{`Sub-extremality' Conjecture}]\label{conj:subsextremal}Among all the data admissible from Theorem \ref{thm:main}, there exists a generic sub-class for which if the maximal future development has $\mathcal{Q}^+\backslash J^-(\mathcal{I^+})\neq \emptyset$, then the black hole is `sub-extremal in the limit'.\footnote{We emphasize that `sub-extremal in the limit' is to be understood in some neighborhood of $i^+$ in $J^-(i^+)$. \label{foot:subext}}
\end{conj}

It should be noted that the asymptotic charge $Q_+$ (and hence the re-normalized mass $\varpi_+$) of the black hole is not \emph{a priori} well-defined when $\e\neq 0$. Moreover, in this case, unlike the case $\e=0$ in which (the topological) charge $Q$ is globally constant, it may be possible that the charge radiates completely to infinity.  We, however, make the following reasonable conjecture.

\begin{conj}[\textbf{Non-vanishing Charge Conjecture}]\label{conj:charge} Among all the data admissible from Theorem \ref{thm:main} with $\e\neq 0$, there exists a generic sub-class for which if the maximal future development has $\mathcal{Q}^+\backslash J^-(\mathcal{I^+})\neq \emptyset$, then the asymptotic charge $Q_+$ of the black hole is non-zero.   
\end{conj}

This conjecture is relevant in view of (\ref{subext}).

\subsubsection{Price's law}\label{sec:intro/conj/price}

With respect to the decay rate (\ref{rate}) in Theorem \ref{thm:price}, heuristics and numerical work \cite{LBGK04,SHTP98} suggest that the charged scalar massive `hairs' of a black hole will decay slower than the neutral (real) massless ones.  In particular, the following conjecture of (a version of) Price's law appears reasonable.
\begin{conj}[\textbf{Price's Law Conjecture}]\label{conj}Let $(\M=\Q^+\times_r \mathbb{S}^2, g_{\mu\nu}, \phi, F_{\mu\nu})$  be as in Theorem \ref{thm:main}.  Assume that $\mathcal{Q}^+\backslash J^-(\mathcal{I^+})\neq \emptyset$. If the black hole is `sub-extremal in the limit' and has asymptotic charge $Q_+$, then for all $\epsilon>0$ there is a positive constant $C_{\epsilon}<\infty$ such that, for a suitable normalized advanced time co-ordinate $v$, 
 \beqn\nonumber
|r\phi| + |r{\rm D}_v\phi|\leq C_{\epsilon} v^{-p + \epsilon}
\eeqn
along $\mathcal{H}^+$, where the decay rate exponent $p =p(\e , \m^2)$ satisfies\footnote{$\e Q_+$ is a dimensionsless quantity; v.~\S \ref{sec:scale}.}
\beqna\nonumber
p(\e , 0)&=&\left\{ 
  \begin{array}{l l}
    1 + \mathfrak{Re}[\sqrt{1- 4(\e Q_+)^2}], & \quad \text{if $\e \neq 0$;}\\
    3 & \quad \text{otherwise;}\\
  \end{array} \right.\\
  p(0, \m^2)&=& \left\{ 
  \begin{array}{l l}
   5/6, & \quad \text{if $\m^2 \neq 0$;}\\
    3 & \quad \text{otherwise}.\\
  \end{array} \right.\nonumber
\eeqna
\end{conj}
In the massless case, it may be surprising that $p(\e,0)$ is not continuous at $\e=0$.  The reason for this is best explained by the heuristics of \cite{HP98b}.  In short, the late-time behavior of neutral scalar `hairs' is determined by spacetime curvature ($\sim M_f$), whereas the late-time behavior of charged scalar `hairs' is dominated by scattering due to electromagnetic interactions in \emph{flat} spacetime. In particular, to (sub-)leading order, the scalar field decays like
\beqn
\frac{\Gamma(\frac{1}{2}p)}{\Gamma(1- \frac{1}{2}p)}  v^{-p} + M_f v^{-(p+1)}\nonumber
\eeqn
along the event horizon, where $\Gamma(x)$ is the (complete) Gamma function.
As $\e \rightarrow 0^+$, i.e., $p\rightarrow 2^-$, the first term tends to zero. Thus, when the scalar field is uncharged (and massless), one would expect to have  sharper decay ($p=3$) along the event horizon (cf.~Theorem \ref{thm:price}).
\subsubsection{Trapped surface formation}\label{sec:intro/conj/trapped}
As discussed in \S\ref{sec:intro/cd/c}, trapped surface formation is central to establishing that, in particular, generically $\chg=\emptyset$ (cf.~Statement IV.5d of Theorem \ref{thm:main}).  Given the nature of the argument sketched in \S \ref{sec:intro/cd/c}, Christodoulou was led to a trapped surface conjecture in \cite{DC98}, which, in the context of spherical symmetry, takes the form of

\begin{conj}[\textbf{Spherical Trapped Surface Conjecture}]\label{conj:trapped_surface_conj_christo}
Among all the data admissible from Theorem \ref{thm:main}, there exists a generic sub-class for which the maximal future development has either $b_{\Gamma} = i^+$ or $\mathcal{A}\neq \emptyset$ and $\mathcal{A}$ has a limit point on $b_{\Gamma}$ (whence \emph{a fortiori} $\chg= \emptyset$).
\end{conj}

By Statement III of Theorem \ref{thm:main}, Conjecture \ref{conj:weak} follows from Conjecture \ref{conj:trapped_surface_conj_christo}.  Moreover, in the case $\m^2=\e=\imp=F_{\mu\nu}=0$, by Statement VII of Theorem \ref{thm:main},  Conjecture \ref{conj:trapped_surface_conj_christo} also implies (the $C^0$-formulation\footnote{cf.~the discussion of regularity in footnote \ref{foot:christo}.} of) Conjecture \ref{conj:strong} since $\mathcal{CH}_{i^+}=\emptyset$. More generally (see Statement VII of Theorem \ref{thm:main}), if Conjecture \ref{conj:trapped_surface_conj_christo} were true, then the problem of strong cosmic censorship completely reduces to understanding the behavior of the solution near $\mathcal{CH}_{i^+}$.  In short, implicit in Conjecture \ref{conj:trapped_surface_conj_christo} is a partial result concerning strong cosmic censorship. 

If, however, we consign ourselves to just resolving weak cosmic censorship, then Theorem \ref{thm:main} actually allows us to state a \emph{weaker} trapped surface conjecture, from which weak cosmic censorship would also follow.  Indeed, since the presence of a single (marginally) trapped surface indicates\footnote{The converse is not true.  A black hole region need not contain a trapped surface (e.g., a spacetime with $\sgt\cup \mathcal{S}\cup \mathcal{S}_{i^+}= \emptyset$).  Note, however, Conjecture \ref{conj:outermost}. } that a spacetime has a non-empty black hole region, Theorem \ref{thm:main} immediately gives \emph{a fortiori}
\begin{cor}[Dafermos \cite{MD05b}]\label{cor:trap_w} Under the assumptions of Theorem \ref{thm:main}, if $\mathcal{A}\neq \emptyset$, then $\mathcal{I}^+$ is complete.
\end{cor}

Conjecture \ref{conj:weak} then follows from

\begin{conj}[\textbf{Weak Spherical Trapped Surface Conjecture}]\label{conj:trappedsurfaceconj}Among all the data admissible from Theorem \ref{thm:main}, there exists a generic sub-class for which the maximal future development has either $b_{\Gamma} = i^+$ or $\mathcal{A}\neq \emptyset$.  
\end{conj}
In delimiting the geometry of the trapped region (cf.~Theorem \ref{thm:trapped_i+}), we also state

\begin{conj}[\textbf{Outermost Apparent Horizon Conjecture}]\label{conj:outermost}For initial data as in Theorem \ref{thm:main}, if the maximal future development has $\mathcal{Q}^+\backslash J^-(\mathcal{I}^+)\neq \emptyset$ and the black hole is `sub-extremal in the limit', then there exists a non-empty `asymptotically connected' component of the outermost\footnote{See footnote \ref{foot:outer} for a definition.} apparent horizon $\mathcal{A}'$ that terminates at $i^+$.  Moreover,
\beqn\nonumber
\mathcal{A}' \cap \mathcal{U}= \mathcal{A}\cap \mathcal{U}
\eeqn
in a sufficiently small neighborhood $\mathcal{U}\subset \R^{1+1}$ of $i^+$ and, in particular,
\beqn\nonumber
I^+(\mathcal{A}\cap \mathcal{U}) \cap \Q^+  \subset \mathcal{T}.
\eeqn
\end{conj} 

By Statement IV of Theorem \ref{thm:main}, Conjecture \ref{conj:outermost}, in particular, implies that $r$ extends continuously to $\mathcal{CH}_{i^+}$ in a sufficiently small neighborhood of $i^+$.

If, in addition to Conjecture \ref{conj:outermost},  Conjecture \ref{conj:trapped_surface_conj_christo} holds, then $\mathcal{B}^+\backslash\left( i^{+}\cup \mathcal{I}^+\cup i^0\right)$ is always `preceded' by a trapped region.  This scenario should be compared with the assumptions of the trapped surface conjecture given by Christodoulou in \cite{DC98}.  We see that, in the terminology of \cite{DC98}, under Conjecture \ref{conj:outermost}, the terminal indecomposable past sets $I^-(p)\cap \Q^+$ for $p\in \mathcal{S}_{i^+}\cup \mathcal{CH}_{i^+}$ would correspond to sets whose trace on $\Sigma$ do \emph{not} have compact closure, but which would nonetheless satisfy Christodoulou's condition for containing a trapped surface.

\subsubsection{Cauchy horizon conjectures}\label{sec:intro/conj/strong}
In light of the results discussed in \S \ref{sec:intro/cd/d}, it seems reasonable to conjecture the following.

\begin{conj}[\textbf{Continuous Extendibility Conjecture}]\label{conj:nonemptyext} For the development of initial data as in Theorem \ref{thm:main}, if $\Q^+\backslash J^-(\mathcal{I}^+)\neq \emptyset$, the black hole is `sub-extremal in the limit', and the asymptotic charge $Q_+\neq 0$, then $\mathcal{CH}_{i^+}\neq \emptyset$ and the solution is continuously extendible beyond $\mathcal{CH}_{i^+}$.
\end{conj}

If the set of initial data for which the assumptions of Conjecture \ref{conj:nonemptyext} hold has non-empty interior, then the
 $C^0$-formulation of Conjecture \ref{conj:strong} is \emph{false}! Not all hope is lost for the fate of strong cosmic censorship, though.  For, if Conjecture \ref{conj:trapped_surface_conj_christo} can be established, then the $C^2$-formulation of Conjecture \ref{conj:strong} reduces to showing that the curvature blows up along $\mathcal{CH}_{i^+}$. 
 
 Because it is expected that a complex-valued scalar field will have late-time oscillatory behavior along $\mathcal{H}^+$ (suggested heuristically in \cite{SHTP98}), it seems unlikely, as a result, that the lower bound (\ref{lowerbound}) of Theorem \ref{thm:mass} will hold.  Heuristics, nonetheless, suggest that `mass inflation' still occurs and it is reasonable to expect that the conclusion of Theorem \ref{thm:mass} is true.  We, therefore, state

\begin{conj}[\textbf{$\mathcal{CH}_{i^+}$ Curvature Blow-up Conjecture}]\label{conj:strong/EMKG} Among all the data admissible from Theorem \ref{thm:main}, there exists a generic sub-class for which if the maximal future development has $\mathcal{Q}^+\backslash J^-(\mathcal{I^+})\neq \emptyset$, then the curvature blows up along $\mathcal{CH}_{i^+}$.
\end{conj}  

Although not relevant to weak and strong cosmic censorship, one might expect that null boundary components on which $r=0$, if they occur at all, to be unstable.  As such, the following conjecture seems reasonable.

\begin{conj}[\textbf{Spacelike $r=0$ Singularity Conjecture}]\label{conj:wandstrong}Among all the data admissible from Theorem \ref{thm:main}, there exists a generic sub-class for which if the maximal future development has $\mathcal{Q}^+\backslash J^-(\mathcal{I^+})\neq \emptyset$, then $\sgo\cup\sgt\cup \mathcal{S}_{i^+}=\emptyset$, $\mathcal{S}\neq \emptyset$, and $\mathcal{S}$ is $C^1$-spacelike.
\end{conj}

We should remark that although black hole solutions without a spacelike singularity, as depicted in diagram IX below, do not serve (in view of Conjecture \ref{conj:strong/EMKG}) as counter-examples to weak and strong cosmic censorship, it is reasonable to conjecture, as above, that they would be ruled out, nonetheless, by genericity, hence we have included the statement $\mathcal{S}\neq \emptyset$ in the above Conjecture.

It may be worthwhile to note that Conjecture \ref{conj:wandstrong} is \emph{false} in the case of two-ended data, as shown by Dafermos \cite{MD12}. 

\setcounter{subfigure}{8}
\begin{figure}[h!]
  \centering
   \subfloat[\scriptsize black hole with no spacelike singularity]{\includegraphics[scale=1.5]{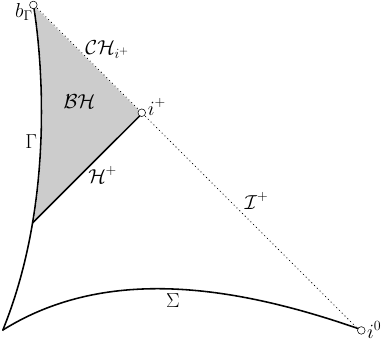}}  
\end{figure}
\vspace{-.35cm}
\pagebreak
\subsubsection{Web of implications: cosmic censorship}\label{sec:intro/conj/web}

For convenience and clarity, we collect the various conjectures and their implications in regards to weak and strong cosmic censorship below.

\begin{center}
\includegraphics[scale=.95]{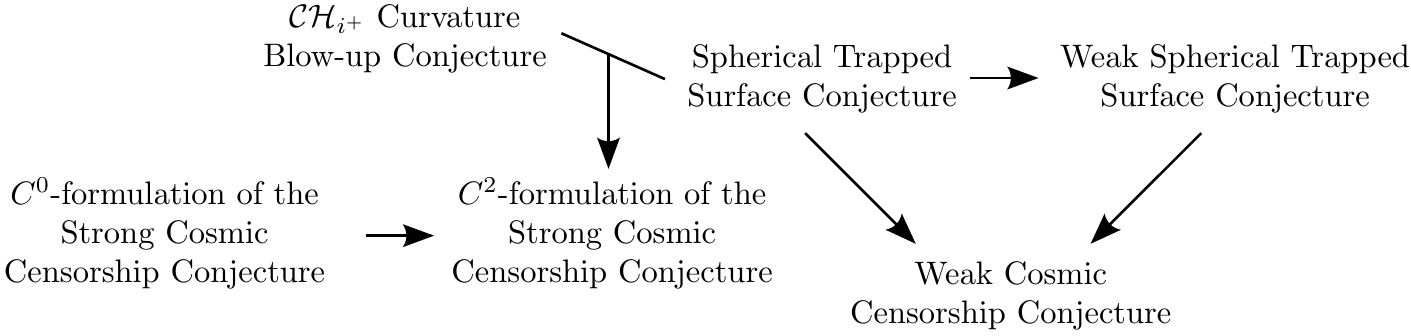}\vspace{.25cm}
\vspace{-3mm}
\includegraphics[scale=.95]{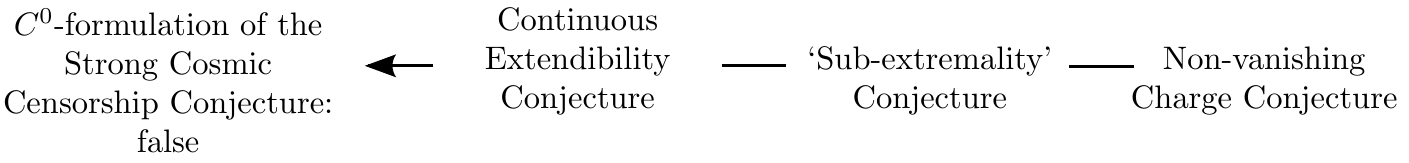}
\end{center}
\vspace{-2.9mm}\hspace{27.7mm}\footnote{In showing that the $C^0$-formulation of strong cosmic censorship is \emph{false}, it suffices to show a weaker version of Conjectures \ref{conj:subsextremal} and \ref{conj:charge}, for the statements of the latter would only need to hold on an open set of solutions not for generic solutions. }

\subsubsection{Generic Einstein-Maxwell-Klein-Gordon spacetimes}\label{sec:intro/conj/generic}
We wish to conclude this section with a summarized description of \emph{generic} spherically symmetric Einstein-Maxwell-Klein-Gordon spacetimes, as would follow from a positive resolution to Conjectures \ref{conj:subsextremal}--\ref{conj:wandstrong}. The resulting two classes of generic solutions, which we shall call the black hole case and the non-black hole case, have Penrose diagrams as depicted below.  
\setcounter{subfigure}{9}
\begin{figure}[h!]
\centering
\subfloat[\scriptsize black hole with a non-`first singularity'-emanating Cauchy horizon]{\includegraphics[scale=1]{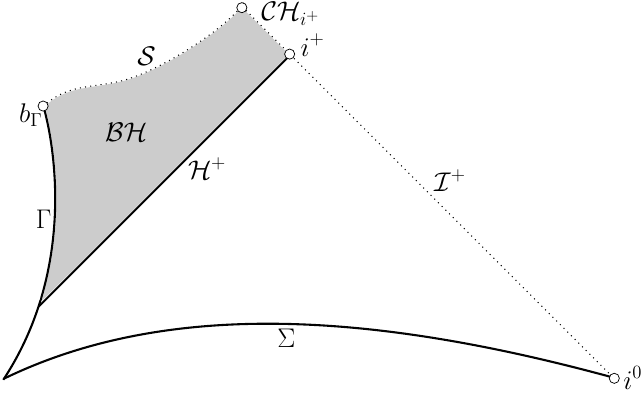}}
\quad\quad
\subfloat[\scriptsize dispersive or star-like]{\includegraphics[scale=.95]{fig_I.pdf}}
\end{figure}

\subsubsection*{The black hole case}
Conjectured generic Einstein-Maxwell-Klein-Gordon black hole solutions would have Penrose diagram depicted in diagram X and would have the following properties:

\begin{itemize}
\item [1.] the black hole is `sub-extremal in the limit'; 
\item [2.] the asymptotic charge $Q_+$ is well-defined and $Q_+\neq 0$;
\item [3.] in a neighborhood of $i^+$ in $J^-(\mathcal{I}^+)$ the spacetime asymptotically approaches Reissner-Nordstr\"om at a rate given by Price's law;
\item [4.] $\mathcal{S}\neq \emptyset$ (hence $\mathcal{A}\neq\emptyset$) and $\mathcal{S}$ is $C^1$-spacelike;
\item [5.] $\mathcal{A}$ has limit points on $b_{\Gamma}$ and $i^+$;
\item [6.] $\sgo\cup\sgt\cup\chg\cup \mathcal{S}_{i^+} = \emptyset$; and,
\item [7.] $\mathcal{CH}_{i^+}\neq \emptyset$ and the curvature blows up on $\mathcal{CH}_{i^+}$.
\end{itemize}  
\subsubsection*{The non-black hole case}
Conjectured generic Einstein-Maxwell-Klein-Gordon non-black hole solutions would have Penrose diagram as depicted in diagram XI and are of two possible types, one of which we shall call dispersive and the other `star-like'.

In the case $\m^2=0$, it is expected, as in the model of Christodoulou, that all non-black hole spacetimes will have vanishing final Bondi mass $M_f$.  These solutions will be called dispersive.  By the results of Chae \cite{DChae03}, there is an `open' set of initial data containing trivial data whose developments are dispersive.

On the other hand, when $\m^2\neq0$, the Einstein-Maxwell-Klein-Gordon system can, in general, admit charged (boson) star solutions \cite{JVDB89, FSEM03}, i.e.,~when $\m^2\neq 0$, the scalar field may not be `dispersive'; we shall call such solutions `star-like'.  Let us also note, however, that the final Bondi mass $M_f$ is not a suitable measure of dispersive phenomena, as massive scalar fields do not radiate to $\mathcal{I}^+$ (cf.~\cite{AH93}).  

The rich possibility of solutions in the $\m^2\neq 0$ case will complicate the analysis of non-black hole solutions, but its scope lies outside the context of this paper. Indeed, one may formulate a host of conjectures, as we did in the black hole case, regarding the properties of generic non-black hole solutions, but expounding on such properties would here take us too far afield.

\subsection{Generalized extension principle}\label{sec:intro/gext}

The main content of Theorem \ref{thm:main} consists of establishing an extension principle, characterizing `first singularities', considerably stronger than that proposed by Dafermos in \cite{MD05b}.  While useful for weak cosmic censorship, the extension principle of \cite{MD05b}, which concerns only the closure of the regular region $\mathcal{R}$ of spacetime, is insufficient to delve into the inner reaches of the black hole region where there are potentially trapped surfaces. Not only will we prove, in particular, the extension principle of \cite{MD05b} for our system (\ref{RMN})--(\ref{eqn:kg}), but we will give a stronger result: A `first singularity' must emanate from a spacetime boundary to which the area-radius function $r$ extends continuously to zero. 

A practical result in its own right, we wish to include this generalized extension principle  as a stand-alone statement. In view of applications to cosmological topologies or to the case of two asymptotically flat ends, it is useful to add an assumption on the finiteness of the spacetime volume, which, as we shall see (cf.~Proposition \ref{prop:finite_volume}), can be retrieved under the assumptions of Theorem \ref{thm:main}.  We thus formulate the extension principle as follows:

\begin{thm}\label{thm:gext} Let $(\M = \mathcal{Q}^+\times_r \mathbb{S}^2, g_{\mu\nu}, \phi, F_{\mu\nu})$ denote the maximal future development of smooth spherically symmetric initial data for the Einstein-Maxwell-Klein-Gordon system. For $p\in \overline{\mathcal{Q}^+}$ and $q\in \left(I^-(p)\cap \Q^+\right)\backslash \{p\}$ such that 

\[\Dm = \left(J^+(q)\cap J^-(p)\right)\backslash\{p\}\subset \Q^+,
\]
if
\begin{itemize}
\item [1.]  $\Dm$ has finite spacetime volume; and,
\item [2.] there are constants $r_0$ and $R$ such that
\begin{center}
$0 < r_0 \leq r (p')\leq R <\infty$\hspace{5mm} for all\quad$p'\in \Dm$,
\end{center}
\end{itemize}
then $p\in \Q^+$.
\end{thm}
We should re-iterate that there is neither an assumption on the global (say, asymptotically flat or hyperboloidal) geometry nor the topology of the initial data in Theorem \ref{thm:gext}.  This generality of the extension principle is made possible by the fact that the proof of Theorem \ref{thm:gext} does \emph{not} rely on any form of coercive energy integral arising from the Hawking mass\footnote{one cannot use energy conservation due to a lack of monotonicity of the Hawking mass in the trapped region}, but that it directly exploits the special null structure in the Einstein-Maxwell-Klein-Gordon system.  This null structure manifests itself as follows: To control the metric and matter fields in $C^k$ $(k\geq 2)$, it suffices to give spacetime integral estimates
\beqn\label{nullstructure}
\int\int r^2T_{uv}~\dd u\dd v \hspace{1cm}\textrm{and}\hspace{1cm}\int\int \D_u \phi\left(\D_v\phi\right)^{\dagger} + \D_v \phi \left(\D_u\phi\right)^{\dagger}~\dd u\dd v.
\eeqn
In particular, that potentially `bad' $uu$- and $vv$-components do not appear in the integrands of (\ref{nullstructure}) is a consequence of the null structure (both of the coupling of the matter equations to gravity and the matter equations themselves, respectively). This allows us to integrate by parts (\ref{nullstructure}) so as to always exploit one of the `good' ingoing or outgoing directions.  The symmetrization in 
\beqn\nonumber
\D_u \phi\left(\D_v\phi\right)^{\dagger} + \left(\D_u\phi\right)^{\dagger}\D_v \phi
\eeqn
plays an important role in being able to make use of the null structure. See \S\ref{sec:general_emkg/proof_gext/integral}--\ref{sec:tuv}.

\subsection{General spherically symmetric Einstein-matter systems}\label{sec:intro/general}
Because of the importance of a suitable extension principle in providing a global characterization of spacetime, we wish to cast the contents of \cite{MD05b} and Theorem \ref{thm:gext} in a much greater context. 
\subsubsection{Weak and generalized extension principles}\label{sec:intro/general/weak}

 We begin with the following definitions, recalling the notation introduced in Theorem \ref{thm:main}.

\begin{wext}The \emph{weak extension principle} is satisfied for an Einstein-matter system if the following condition holds: 

Let $(\M = \Q^+\times_r \mathbb{S}^2, g_{\mu\nu}, \ldots)$ denote the maximal future development of spherically symmetric asymptotically flat initial data with one end containing no anti-trapped regions.
Suppose $p\in \overline{\mathcal{R}}\backslash{\overline{\Gamma}}\subset \overline{\Q^+}$ and $q\in \left( I^-(p)\cap\overline{\mathcal{R}}\right)\backslash \{p\}$ are such that
$\left(J^-(p)\cap J^+(q)\right)\backslash \{p\}\subset \mathcal{R}\cup \mathcal{A}.$
Then, $p\in \mathcal{R}\cup \mathcal{A}$. 
\end{wext}

We emphasize that the closure and causal-geometric constructions are with respect to the topology of the ambient $\R^{1+1}$.  The weak extension principle states that a `first singularity' emanating from the closure of the regular region can only do so from the center. 

\begin{gext}The \emph{generalized extension principle} is satisfied for an Einstein-matter system if the following condition holds: 

Let $(\M = \Q^+\times_r \mathbb{S}^2, g_{\mu\nu}, \ldots)$ denote the maximal future development of spherically symmetric initial data.  For $p\in \overline{\Q^+}$ and $q\in \left(I^-(p)\cap \Q^+\right)\backslash \{p\}$ such that $\Dm = \left(J^+(q)\cap J^-(p)\right)\backslash\{p\}\subset \Q^+$, suppose that
\begin{itemize}
\item [1.]  $\Dm$ has finite spacetime volume; and,
\item [2.] there are constants $r_0$ and $R$ such that
\begin{center}
$0 < r_0 \leq r (p')\leq R <\infty$\hspace{5mm} for all\quad$p'\in \Dm$.
\end{center}
\end{itemize}
Then, $p\in \Q^+$.
\end{gext}
The generalized extension principle states that given a `first singularity', either it must emanate from a spacetime boundary to which the area-radius function $r$ can be extended to zero, or else the causal past of the `first singularity' will have infinite spacetime volume.  

While \emph{a priori} logically independent statements (nonetheless supporting our naming convention), the generalized extension principle implies the weak extension principle if the matter model obeys the null energy condition (cf.~Proposition \ref{prop: antitrapped} in \S\ref{sec:pre/noanti} and  Proposition \ref{prop:finite_volume} in \S \ref{sec:pre/volume}). See also Proposition \ref{imply} in 
\S \ref{sec:intro/general/tame}.

We also emphasize, as in Theorem \ref{thm:gext}, the generalized extension principle is stated without reference to the topology or geometry of the initial data and can be applied, for example, to the cosmological setting or the case with two asymptotically flat ends.

\subsubsection{`Tame' matter models}\label{sec:intro/general/tame}
In accordance with the above extension principles, we introduce the following notions of `tame' Einstein-matter systems.
\begin{defn} A spherically symmetric Einstein-matter system is called \emph{weakly tame} with respect to a suitable notion of maximal development if (1) the matter obeys the dominant energy condition; and, (2) the weak extension principle holds.
\end{defn}
\begin{defn} A spherically symmetric Einstein-matter system is called \emph{strongly tame} with respect to a suitable notion of maximal development if (1) the matter obeys the dominant energy condition; and, (2) the generalized extension principle holds.
\end{defn}
With this classification we have
\begin{prop}\label{imply}
A strongly tame Einstein-matter system is weakly tame.
\end{prop}
For a proof of this statement, see \S \ref{sec:proof_main/first}.

\subsubsection{Generalization of Theorem \ref{thm:main} to strongly tame matter models}\label{sec:intro/general/strong}

The proof of Theorem \ref{thm:main}, after the conclusion of Theorem \ref{thm:gext} has been established, follows from a series of monotonicity arguments governed by the dominant energy condition.\footnote{Much of Theorem \ref{thm:main}, in fact, uses the monotonicity governed by Raychaudhuri's equation, which just needs the null energy condition (cf.~the proof of Theorem \ref{thm:main} in \S\ref{sec:proof_main}).}  No structure particular to the Einstein-Maxwell-Klein-Gordon system is used.  As a result, we can state

\begin{thm}\label{thm:s_t_general}Let $(\M = \Q^+\times_r\mathbb{S}^2, g_{\mu\nu},\ldots)$ denote the maximal future development of smooth spherically symmetric asymptotically flat initial data with one end for a strongly tame Einstein-matter system containing no anti-trapped spheres of symmetry.  Then, the conclusion of Theorem \ref{thm:main} holds for this system.
\end{thm}
See the comment in \S\ref{sec:proof_general} regarding the proof of this statement.
\subsubsection{A version of Theorem \ref{thm:main} for weakly tame matter models}\label{sec:intro/general/weakly_tame}

One can deduce from the proof of Theorem \ref{thm:main} that the weak extension principle, in fact, recovers the boundary characterization of Statement II except for the characterization that $r$ vanishes on $\mathcal{S}$.  In other words, establishing the weak extension principle is not sufficient to rule out the possibility that $r$ has non-zero limit values on (part of) $\mathcal{S}$.  

To establish many of the statements of Theorem \ref{thm:main}, however, it is not important to have a characterization of $r$ on $\mathcal{S}$; these results, consequently, hold \emph{mutatis mutandis} for weakly tame matter models. In particular, we state     

\begin{thm}\label{thm:w_t_general}Let $(\M = \Q^+\times_r\mathbb{S}^2, g_{\mu\nu},\ldots)$ denote the maximal future development of smooth spherically symmetric asymptotically flat initial data with one end  for a weakly tame Einstein-matter system containing no anti-trapped spheres of symmetry.  Then, except for the statements enclosed in boxes, the conclusion of Theorem \ref{thm:main} holds for this system.
\end{thm}

It should be noted, moreover, that many of the enclosed `boxed' statements can be (trivially) re-worked as to apply even in the weakly tame case:\footnote{In the case of Statement VII.3, it is presumed that we can extend the solution into a neighborhood of $i^+$.  Since \emph{a priori} this neighborhood will contain trapped spheres, we must appeal to the generalized extension principle. Moreover, because we need to establish a positive (non-zero) lower bound on $r$ in this neighborhood, although there is no explicit reference to $\mathcal{S}$, Statement VII.3 requires, indeed, that a characterization of $r$ be given along $\mathcal{S}$.}\\

\noindent\textbf{Statement IV.3{*}}
If $\sgt\cup\mathcal{S}_{i^+}\neq \emptyset$, then $\mathcal{A}\cup \mathcal{T}\neq \emptyset$. (If $\sgt\cup\mathcal{S}_{i^+}=\emptyset$, then $\mathcal{A}\cup\mathcal{T}$ is possibly empty.)\\

\noindent\textbf{Statement IV.5a{*}}
If $\mathcal{A}\neq \emptyset$, then all limit points of $\mathcal{A}$ that lie on the boundary $\overline{\mathcal{Q}^+}\backslash\mathcal{Q}^+$ lie on $\mathcal{CH}_{i^+}\cup i^+$ and on a (possibly degenerate) closed, necessarily connected interval of $b_{\Gamma}\cup\sgo\cup\chg\cup \mathcal{S}$. 
\\

\noindent\textbf{Statement IV.5e{*}}
If $\sgt\cup\mathcal{S}_{i^+}\neq\emptyset$, then $\mathcal{A}$ has a limit point on $b_{\Gamma}\cup\sgo\cup\chg\cup \mathcal{S}$.\\

\noindent\textbf{Statement VII.1{*}}
The Kretschmann scalar $ R_{\mu\nu\alpha\beta}R^{\mu\nu\alpha\beta}$ is a continuous $[0,\infty]$-valued function on $\Q^+\cup\sgt\cup \mathcal{S}_{i^+}$ that yields $\infty$ on $\sgt\cup\mathcal{S}_{i^+}$. The rate of blow-up is no slower than $r^{-4}$.\\

\noindent\textbf{Statement VII.4{*}}
 If $(\mathcal{M},g_{\mu\nu})$ is future-extendible as a $C^2$-Lorentzian manifold $(\widetilde{\mathcal{M}}, \widetilde{g_{\mu\nu}})$, then there exists a timelike curve $\gamma\subset \widetilde{\M}$ exiting the manifold $\mathcal{M}$ such that the closure of the projection of $\gamma|_{\M}$ to $\Q^+$ intersects $\chg\cup\mathcal{S}\cup \mathcal{CH}_{i^+}$.\\

Since in a weakly tame model we know nothing \emph{a priori} about the behavior of the metric at $\mathcal{S}$, we note that, in turn, Statement VII.4{*}, in practice, tells us very little about inextendibility properties. For this reason, establishing that an Einstein-matter system is strongly tame is a crucial first step in understanding strong cosmic censorship.

\subsubsection{Examples of weakly and strongly tame models}\label{sec:intro/examples}
We now give examples of known tame Einstein-matter systems.
\subsubsection*{Strongly tame}
In the language of \S\ref{sec:intro/general/tame}, Theorem \ref{thm:gext} shows that Einstein-Maxwell-Klein-Gordon is strongly tame; Dafermos and Rendall show\footnote{Previously, Dafermos and Rendall had shown that Einstein-Vlasov is weakly tame in \cite{MDAR05}. In view of Proposition \ref{imply}, however, this is immediate from the subsequent work of \cite{MDAR07}.} in \cite{MDAR07} that Einstein-Vlasov is strongly tame. In each of these proofs, one heavily exploits relevant null structure (arising from the coupling of the matter equations to gravity and/or the matter equations themselves).  We make the following imprecise conjecture.
\begin{conj}\label{conj:null}If a spherically symmetric Einstein-matter system satisfies a suitable `null condition',  then the system is strongly tame.
\end{conj}
For a discussion of the `null condition', see Klainerman \cite{SK86}.
\subsubsection*{Weakly tame}

Dafermos shows in \cite{MD05a} that Einstein-Higgs with non-negative potential $V(\phi)$ is weakly tame. 

Narita \cite{MN08} considers `first singularity' formation in the Einstein-wave map system with target $\mathbb{S}^3$ and $H^2$.  In the language of the present paper, these models are weakly tame. 

\subsubsection{Exotic models}\label{sec:intro/examples/exotic}

The definitions of weakly and strongly tame are tailored specifically so as to apply to classical, self-gravitating matter models.  One often encounters in the physics literature, however, models that are `exotic' in some respect.  In the sequel, we will show that suitable notions of weakly and strongly tame can still be introduced for such systems. 

\subsubsection*{Exotic matter}

Immediate from the proof of Theorem \ref{thm:gext}, it follows that Einstein-Klein-Gordon ($\e=F_{\mu\nu}=0$) with $\m^2 <0$ satisfies the generalized extension principle, but because this non-classical matter model does not obey the dominant energy condition, it is not, according to our definition, strongly tame.  The matter model, however, does obey the null energy condition.  As the proof of Theorem \ref{thm:s_t_general} will make clear, many statements that hold for strongly tame Einstein-matter systems, in fact, follow from monotonicity governed by Raychaudhuri's equation, which just needs the null energy condition. As a result, much can still be said of the global structure of  Einstein-Klein-Gordon spacetimes when $\m^2 <0$. See the recent work of Holzegel and Smulevici who consider spherically symmetric asymptotically AdS Einstein-Klein-Gordon spacetimes \cite{GHJS11,GHJS11b}.

In the case of Einstein-Higgs, the proof of the weak extension principle in \cite{MD05a} can be established with the help of the flux provided by the Hawking mass.  The weak extension principle, consequently, can be given more generally for a potential that is bounded from below by a (possibly negative) constant: $V(\phi) \geq -C$.  Allowing for such a lower negative bound, \cite{MD05a} was able to disprove a scenario of `naked singularity' formation that had appeared in the high energy physics literature.
  Unless the potential is non-negative, however, this non-classical matter model is not weakly tame, by our definition, because it does not obey the dominant energy condition.

Although we will not discuss further such non-classical exotic matter, one could consider systems such as Einstein-Klein-Gordon (arbitrary $\m^2$) as being `quasi-strongly tame' and Einstein-Higgs with $V(\phi)\geq -C$ as being `quasi-weakly tame', \emph{et cetera}.

\subsubsection*{Higher dimensional case}

Dafermos and Holzegel show that the analogue of the weak extension principle holds for the maximal future development of asymptotically flat Einstein-vacuum initial data having triaxial Bianchi IX symmetry \cite{MDGH06}.  Here, spacetime $(\M, g_{\mu\nu})$ is $5$-dimensional and $SU(2)$ acts by isometry.  The Einstein-vacuum equations under this symmetry assumption can be written as a system of non-linear PDEs on a 2-dimensional Lorentzian manifold $\Q^+ = \M /SU(2)$ (possibly with boundary), whose `two dynamical degrees of freedom' correspond to the possible deformations of the group orbit $3$-spheres.  This vacuum model shares a formal similarity with a spherically symmetric $4$-dimensional Einstein-matter system, whose matter obeys the dominant energy condition.  We can, as a result, formally view triaxial Bianchi IX as being weakly tame in our present context.

One can consider more straightforward generalizations to higher dimensions by considering spherically symmetric $n$-dimensional ($n\geq 4$) spacetimes $(\M, g_{\mu\nu})$ whereby $SO(n-1)$ acts by isometry  (cf.~\cite{LR05}).

\subsubsection*{Modified gravity}
Narita \cite{MN09} has considered `first singularity' formation in the spherically symmetric Einstein-Gau\ss-Bonnet-Klein-Gordon system ($n\geq 5$) with $\m^2=\e=\imp=F_{\mu\nu}=0$.  In the language of the present paper, this system is weakly tame.
 
\subsubsection{Regularity of the maximal future development}\label{sec:intro/general/regularity}
Our notion of `tameness' attempts to classify the type of `first singularities' a given spherically symmetric Einstein-matter system will exhibit.  This classification is naturally regularity-dependent.  The discussion of weakly and strongly tame models above has been restricted to the case of smooth maximal future developments.  It will often be convenient, even necessary, to consider developments that are non-smooth.

\subsubsection*{Solutions of bounded variation for the scalar field model}

As we have mentioned before (cf.~\S \ref{sec:intro/cd/c}), in order to initiate the study of a large class of solutions sufficiently flexible to exploit the genericity assumption inherent in the formulation of cosmic censorship, Christodoulou introduces a notion of bounded variation (BV) solutions for the spherically symmetric Einstein-Klein-Gordon system with $\m^2=\e=\imp=F_{\mu\nu}=0$. In \cite{DC93}, Christodoulou establishes the well-posedness of an initial value formulation of this system for given BV initial data.  Christodoulou is, moreover, able to establish that the system is, in the language of the present paper, strongly tame.

\subsubsection*{Shell-crossing singularity formation in Einstein-dust}

Consider the Einstein-Euler system with equation of state $p = 0$, i.e.,~a pressure-free fluid.  This system is also known as Einstein-dust.  The spherically symmetric (infinite dimensional family of asymptotically flat) solutions of the Einstein-dust system were first given by Tolman in \cite{RT34} based on the work of Lema\^itre \cite{GL33}.  Beginning with the work of Oppenheimer and Snyder \cite{OS39}, which discussed in detail the gravitational collapse of a uniform density `ball of dust', this matter model had (and still does) spawn great interest in the physics community. 

In Yodzis et al.~\cite{YSM73}, it is shown that, in general, the Einstein-dust system forms (`naked') `first singularities' away from the center (in the non-trapped region), commonly referred to as shell-crossings, in the class of smooth maximal future developments.  In the language of the present paper, this shows that Einstein-dust is \emph{not} weakly tame (and hence not strongly tame) in the smooth sense.  In a way, this result is unsurprising; shell-crossings occur already in the absence of gravity, e.g.,~on a fixed Minkowski background.  It turns out, however, that the solution obtained by extending beyond the shell-crossings makes physical sense (the metric is, in particular, still continuous \cite{PH67} as long as $r>0$).  One can thus view Einstein-dust as being strongly tame in a suitable class of rough solutions.  We note, however, the negative resolution of cosmic censorship for the Einstein-dust system, even in this more appropriate class, shown by Christodoulou (cf.~\S\ref{sec:intro/general/wcc_dust}).

\subsubsection*{Shock formation in Einstein-Euler}
The breakdown of smooth solutions of the Euler system has been studied extensively in \cite{JC48, DC07,TS85}.  In coupling to gravity, Rendall and St{\aa}hl \cite{ARFS08} show that under assumption of plane symmetry, smooth solutions still break down in (arbitrarily short) finite time. It is believed that, as in the classical Euler system, the breakdown of Einstein-Euler again is a result of the discontinuities of the fluid flow, commonly referred to as shocks.  In the language of the present paper, Einstein-Euler (for general equation of state) is \emph{not} weakly tame (and hence not strongly tame) for smooth maximal future developments.  Understanding the global properties of solutions to the Euler system, in the context of less regular developments, remains a long-standing (and very difficult!) open problem, where a large-data theory is unavailable even in $1+1$-dimensions.

\subsubsection*{The two-phase fluid model of Christodoulou}  

Christodoulou sought to understand a two-phase fluid model that would capture many of the features of actual stellar gravitational collapse and at the same time would be mathematically tractable.  In \cite{DC95}, Christodoulou considers the spherically symmetric Einstein-Euler system with a two-phase barotropic equation of state given by
\[ p = \left\{ \begin{array}{ll}
         0 & \mbox{if $\rho \leq 1$};\\
        \rho -1 & \mbox{if $\rho > 1$},\end{array} \right. \]
i.e.,~the soft phase of the two-phase model coincides with that of dust ($p=0$) while the hard phase coincides with that of a massless real-valued scalar field, where
 \beqn\nonumber
p = \frac{1}{2}\left(1-g^{\mu\nu}\partial_{\mu}\phi\partial_{\nu}\phi\right)  - 1
 \eeqn
with the restriction that $-g^{\mu\nu} \partial_{\nu}\phi$ be a future-directed timelike vector field.\footnote{This requirement, i.e.,~$\phi$ is a time function, is necessary to ensure that the scalar field has a hydrodynamic interpretation.}  This model seeks to capitalize on the insight gained via his study of dust (see below)  and his subsequent work on scalar fields.  Here, shocks develop from the development of smooth initial data in the form of the boundary between the phases.  Their dynamics can then be understood as a free-boundary problem,\footnote{for which the timelike components of the phase boundary correspond to shocks} which is studied extensively in \cite{DC96, DC96b}\hspace{-1.07mm}{\color{white}\textbf{]}}\hspace{-1.1mm}, \ldots].\footnote{Part of this series by Christodoulou is unpublished.}  In this collection of work, Christodoulou shows, in particular, in the language of the present paper, that, in the context of a suitably rough maximal future development, this two-phase model is strongly tame.

\subsubsection{Cosmic censorship for Einstein-dust and the two-phase model}\label{sec:intro/general/wcc_dust}

In \S\ref{sec:intro/general/regularity}, we saw that the shell-crossing singularities of \cite{YSM73} are not `naked singularity' solutions from the point of view of the correct concept of a solution. On the other hand, that, indeed, true `naked singularities' generically form for the Einstein-dust model is later proven by Christodoulou \cite{DC84}.  Specifically, in the the collapse of an inhomogeneous ball of dust, there exists an open set of initial mass densities $\rho_0$ such that the mass density $\rho$ will become infinite at some central point before the formation of a trapped surface occurs, as opposed to the shell-crossing `first singularities' of \cite{YSM73} that are inessential.  In the language of the present paper, the Penrose diagram of such a spacetime is given by diagram III in \S\ref{sec:intro/cd/c} (light cone singularity) with $i^{\square} = i^{\textrm{naked}}$, where $\rho(b_{\Gamma}) = \infty$ and for which generically the metric is extendible across $\chg$.  In particular, Christodoulou shows that Conjectures \ref{conj:weak} and \ref{conj:strong} are \emph{false} for any suitable notion of an Einstein-dust solution.

It should be noted, however, that the work of Christodoulou on Einstein-dust is not the death knell of cosmic censorship.  The equation of state $p=0$ is a very special one, one which becomes less and less plausible as $\rho\rightarrow \infty$.  If one wishes to consider the problem of cosmic censorship for Einstein-Euler, then more realistic equations of state must be allowed.  In this sense, the two-phase model introduced by Christodoulou is perhaps the most tractable realistic model improving the pure dust case.  Because of the scalar field structure of the hard phase, and in light of Theorem \ref{thm:christo}, it is reasonable to expect that cosmic censorship will be true; this conjecture, however, remains open.

\subsubsection{Table of weakly tame and strongly tame models}\label{sec:intro/general/table}
We conclude this section with a summary of our discussion. 

Unless otherwise noted by a box, the following hold for smooth maximal future developments. In the case of rough developments, the usual caveats about uniqueness apply, as the well-posedness statement, which is formulated `downstairs', can only be discussed assuming the symmetry of the development. 

 For exotic models, a * indicates that weakly and strongly tame are to be understood in a suitable sense, e.g.,~in view of the failure of the energy condition (cf.~\S \ref{sec:intro/examples/exotic}).
\vspace{1.5mm}
{\small
\begin{center}
  \begin{tabular}{l@{\hspace{-.4cm}}c@{~~}c@{\hspace{-1.3mm}}r}
Matter model & Weakly tame & Strongly tame & Remarks     \\
     \hline  
     Maxwell-Klein-Gordon & Yes& Yes& Theorem \ref{thm:gext}\\
     Vlasov & Yes & Yes&\cite{MDAR05,MDAR07}\\
     Klein-Gordon, $\m^2=\e=\imp=F_{\mu\nu}=0$, \fbox{BV} & Yes& Yes& \cite{DC93} \\
     Higgs, $V(\phi)\geq 0$ & Yes & ? & \cite{MD05a} \\
     Wave maps: target $\mathbb{S}^3$, $H^2$ & Yes & ? & \cite{MN08} \\
     Yang-Mills, et al. & ? & ?&null structure?\\
     Euler ($p=0$) & No & No & \cite{YSM73}, shell-crossings\\
     Euler ($p=0$), \fbox{suitably rough} & Yes & Yes & \cite{PH67}, W.C.C. still \emph{false} \cite{DC84}\\
     Euler & No & No & \cite{DC07, ARFS08, TS85}, shocks\\
     Euler, \fbox{suitably rough}  & ? & ?& difficult open problem!\\
    Two-phase fluid, \fbox{free-boundary} & Yes & Yes & \cite{DC95, DC96, DC96b}\hspace{-.99mm}{\color{white}\textbf{]}}\hspace{-1.03mm}, \ldots]\\
    Vacuum, triaxial Bianchi IX & Yes* & ?& \cite{MDGH06}, $n=5$, $SU(2)$-action\\
    Klein-Gordon, $\m^2 <0$& Yes* & Yes* & Theorem \ref{thm:gext}\\
    Higgs, $V(\phi)\geq -C$& Yes* & ? & \cite{MD05a} \\
    Gau\ss-Bonnet-Klein-Gordon, $\m^2=\e=\imp=F_{\mu\nu}=0$ & Yes*& ?& \cite{MN09} 
      \end{tabular}
 \end{center}
 }
\subsection{`Sharpness' of the boundary characterization}
We briefly discuss here the `sharpness' of the boundary characterization as given in Theorem \ref{thm:main}. 

For the purpose 
of this discussion, it is helpful to define a coarser description\footnote{The set $\sgt\cup \mathcal{S}\cup \mathcal{S}_{i^+}$ can be understood as a unit in the statement of Theorem \ref{thm:main}.  See Statements IV.3, IV.5.e and VII.1.   Indeed, $\sgt\cup \mathcal{S}\cup \mathcal{S}_{i^+}$ corresponds to the set $\mathcal{B}_s$ in \cite{MDAR07}.  The reason we choose to separate the components is to highlight the fact that $r$ can be zero on the null components that emanate from $b_{\Gamma}$ and $i^{\square}$.} of the boundary.  Let us define
\beqn\nonumber
\widetilde{\mathcal{S}} = \sgt \cup \mathcal{S}\cup \mathcal{S}_{i^+}.
\eeqn
The spacetime boundary (\ref{bp}) is then given by
\beqn\label{bp2}
\mathcal{B}^+= b_{\Gamma}\cup \sgo \cup \chg \cup \widetilde{\mathcal{S}}\cup\mathcal{CH}_{i^+}\cup i^{\square}\cup \mathcal{I}^+\cup i^0.
\eeqn 
The sets $b_{\Gamma}$, $i^{\square}$, $\mathcal{I}^+$, and $i^0$ are always non-empty.  The set $\widetilde{\mathcal{S}}$ is non-empty for every black hole solution in the case $\m^2=\e=\imp=F_{\mu\nu}=0$ (cf.~\S\ref{sec:intro/cd/c}). Christodoulou \cite{DC94} constructs explicit examples in the BV class for which $\sgo$ is non-empty and examples for which $\chg$ is non-empty. On the other hand, an easy pasting argument involving the Reissner-Nordstr\"om solution yields examples in which $\mathcal{CH}_{i^+}$ is non-empty. Thus, we have the following theorem.
\begin{thm}\label{thm111} For each type of boundary component in (\ref{bp2}), there exists a development of data as in Theorem \ref{thm:main} (possibly in the more general BV class) for which that component-type is non-empty. 
\end{thm}
 
It may be useful to introduce the following nomenclature:
\begin{defn}
A strongly tame Einstein-matter model for which the conclusion of Theorem \ref{thm111} holds is called \emph{fully general}.
\end{defn}
In this language, the Einstein-Maxwell-Klein-Gordon system is fully general, in fact, the `simplest' model that is known to be fully general.\footnote{One could, of course, define alternative notions of `fully general' with respect to the original decomposition (\ref{bp}), or requiring that various component-types be simultaneously non-empty, \emph{et cetera}.  We shall not pursue this here.} This is one way to understand the importance of the charged scalar field model in studying spherically symmetric formulations of cosmic censorship.

\subsection{Is it possible to `super-charge' a black hole?}\label{sec:intro/super}

Of relevance to a discussion in the physics community, Theorem \ref{thm:main} addresses whether one can  `super-charge' a black hole.  By this we mean:  \begin{tiny}
•
\end{tiny}Consider a black hole solution that is (nearly) `extremal'.  Can one, by some (semi-)classical process, add enough additional charge to the system (e.g., by throwing in charged particles) so that the final mass and charge have a `super-extremal' relation?

This notion of `destroying' the event horizon by means of such a process has been entertained in the literature, e.g.,~\cite{BCK10, CLS10, FdFYY01, VH99,TJTS09, MS07,RS08,SS11,RW72}, suggesting that one can `transform' a black hole into a `naked singularity'. If this were possible, then weak cosmic censorship would be false.  These constructions, however, share the feature that $\mathcal{A}\neq \emptyset$.  One is thus to imagine, for example, a Penrose diagram as depicted below:
\begin{center}
\includegraphics[scale=.5]{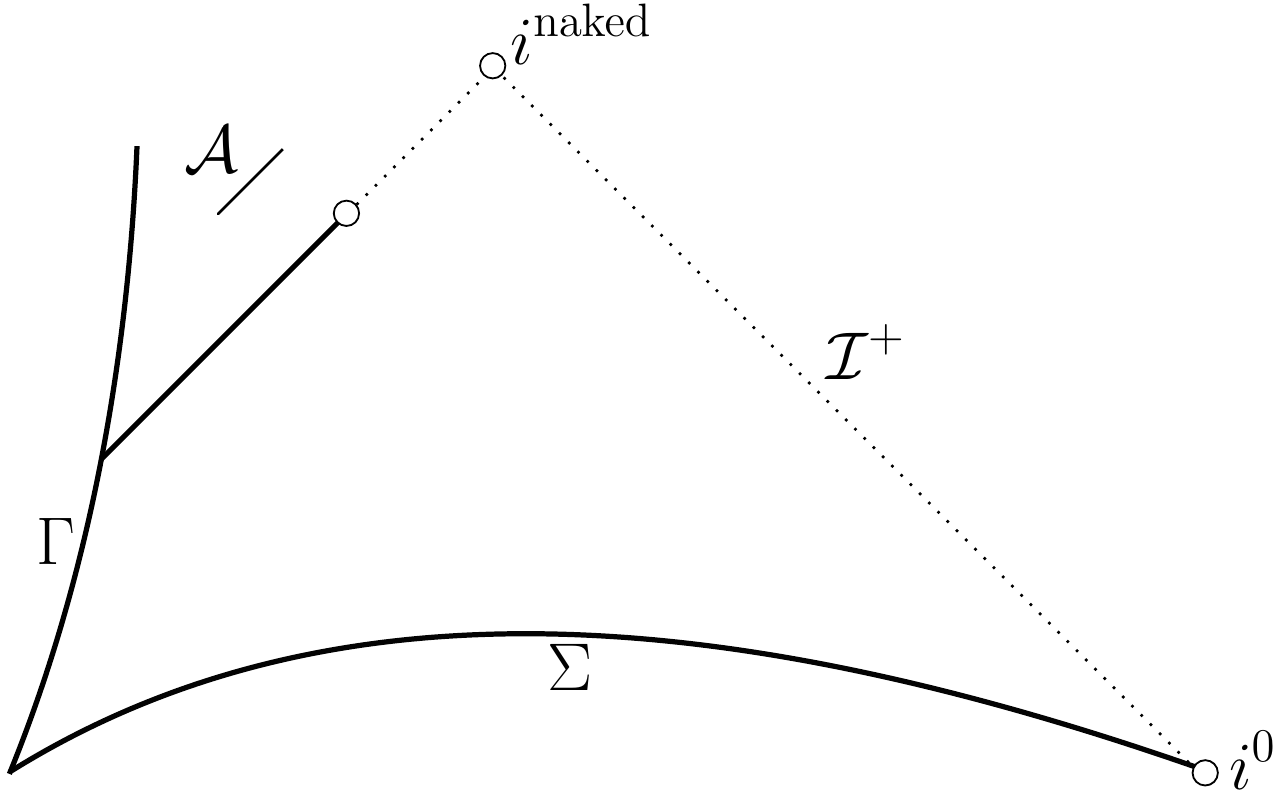}
\vspace{-.1cm}
\end{center}
In view of Theorem \ref{thm:main}, such spactimes simply do not exist.  Moreover, since $\mathcal{A}\neq \emptyset$, by Statement III of Theorem \ref{thm:main}, $\mathcal{I}^+$ is complete and the correct Penrose diagram is as depicted in Theorem \ref{thm:main} with $\mathcal{BH}\neq \emptyset$ and $i^{\square} = i^+$.  Interestingly, one need not understand the precise long-time behavior of $\mathcal{H}^+$ to infer this.  
In particular,  no `naked singularity' can be created in the present model by `super-charging' a black hole.  This confirms the original intuition of Wald \cite{RW72}.
\subsection{Forthcoming results}\label{sec:intro/forth} Lastly, we would like to announce that in a series of forthcoming papers \cite{JK10c, JK10a}, we will complement the result of Theorem \ref{thm:main} by giving criteria on initial data sufficient to yield the non-emptiness of certain boundary components. We will show, in particular, that
the trapped surface formation result of Christodoulou (Theorem \ref{thm:christo}) and the Cauchy horizon stability result of Dafermos (Theorem \ref{thm:c0strong}) can be extended to the full Einstein-Maxwell-Klein-Gordon system ($\e\in \R, \m^2\geq 0$), the latter generalization imposing decay along the event horizon compatible with Conjecture \ref{conj}.\\

\noindent\textbf{Acknowledgements.} I thank my Ph.D. advisor Mihalis Dafermos for providing invaluable insight and guidance.  I would also like to thank Jared Speck and Willie Wong for helpful discussions regarding this project and two anonymous referees for carefully reading the manuscript and providing many insightful remarks and suggestions.  I am supported through the generosity of the Cambridge Overseas Trust.

\section{Preliminaries}\label{sec:pre}
We begin by introducing a few mathematical preliminaries. In what follows, causal-geometric constructions, e.g.,~the causal future $J^+$, the causal past $J^-$, the chronological future $I^+$, the chronological past $I^-$, etc., will refer to the underlying flat Minkowski metric $\eta_{\mu\nu}$ and its induced topology on $\R^{1+1}$.\footnote{We take the convention $p\in J^+(p)\cap J^-(p)$, but $p\not\in I^+(p)\cup I^-(p)$.} 

\subsection{Spacetime geometry of the maximal future development}\label{sec:pre/max}
We consider those spacetimes that, as in Theorems \ref{thm:main} and \ref{thm:gext}, are the maximal future developments of spherically symmetric initial data. In light of the $U(1)$ gauge freedom highlighted in \S\ref{prelimbundle}, we need to be explicit in what we mean by `spherically symmetric' initial data. Let us recall that initial data consists of a Riemannian 3-manifold $(\Sigma^{\scriptsize{(3)}}, h_{ij})$, a symmetric 2-tensor $K_{ij}$ (extrinsic curvature), two 1-forms $E_i$ and $B_i$ (electric and magnetic field, respectively), a connection 1-form $A_i$ (local gauge potential), and functions $\phi$ and $\phi'$ (scalar field and its normal derivative) satisfying the constraint equations. 

\begin{defn}\label{defnspherical} Let $(\Sigma^{\scriptsize{(3)}}, h_{ij}, K_{ij}, E_i, B_i, A_i, \phi, \phi')$ be a smooth initial data set for the Einstein-Maxwell-Klein-Gordon system.  We say the initial data are \emph{spherically symmetric} if
\begin{itemize}
\item [1.] the Lie group $SO(3)$ acts by isometry on $(\Sigma^{\scriptsize{(3)}}, h)$;
\item [2.] the $SO(3)$-action preserves $K$, $|\phi|^2$, $E$, and $B$;
\item [3.] $\Sigma = \Sigma^{\scriptsize{(3)}}/SO(3)$ inherits the structure of a connected 1-dimensional Riemannian manifold (possibly with boundary); and,
\item [4.] when $\e\neq 0$, the $SO(3)$-action additionally preserves $\phi$ and $A$.
\end{itemize}
\end{defn}

In constructing our developments, a local existence result is needed.  We apply an easy generalization of a theorem of Choquet-Bruhat and Geroch \cite{CBG69}, together with standard preservation of symmetry arguments, to obtain 

\begin{prop}\label{choquet}Let $(\Sigma^{\scriptsize{(3)}}, \ldots)$ be a smooth spherically symmetric initial data set for the Einstein-Maxwell-Klein-Gordon system, where, topologically, $\Sigma^{\scriptsize{(3)}}$ is homeomorphic to either $\R^3$, $\mathbb{S}^2\times \R$, or $\mathbb{S}^3$. Then, there exists a unique (up to diffeomorphism of the base manifold and vertical automorphism of its principal $U(1)$-bundle) smooth collection $(\M, g_{\mu\nu}, \phi, F_{\mu\nu})$ such that

\begin{itemize}
\item [1.] $g_{\mu\nu}$, $\phi$, and $F_{\mu\nu}$ satisfy the Einstein-Maxwell-Klein-Gordon equations (\ref{RMN})--(\ref{eqn:kg});
\item [2.] $(\M, g_{\mu\nu})$ is globally hyperbolic and $\Sigma^{\scriptsize{(3)}}$ is a Cauchy surface;
\item [3.] $(\M, g_{\mu\nu}, \phi, F_{\mu\nu})$ induces the initial data set $(\Sigma^{\scriptsize{(3)}}, \ldots)$; and,
\item[4.] for any other collection $(\widetilde{\M}, \widetilde{g_{\mu\nu}}, \widetilde{\phi}, \widetilde{F_{\mu\nu}})$ satisfying Properties 1--3, there is an isometry of $(\widetilde{\M}, \widetilde{g_{\mu\nu}})$ onto a subset of $(\M, g_{\mu\nu})$ that preserves the matter fields and fixes the Cauchy surface $\Sigma^{\scriptsize{(3)}}$.
\end{itemize}
Moreover, $SO(3)$ acts smoothly by isometry on $\M$, preserves $|\phi|^2$ (in fact, $\phi$ itself) and $F_{\mu\nu}$, and the quotient manifold $\M/SO(3)$ inherits the structure of a 2-dimensional, time-oriented Lorentzian manifold (possibly with boundary) that can be conformally embedded as a bounded subset of $(\R^{1+1}, \eta_{\mu\nu})$.  

Let $\Gamma$ denote the projection to $\M/SO(3)$ of the set of fixed points of the $SO(3)$-action on $\M$.  $\Gamma$ is then the (possibly empty, not necessarily connected) timelike boundary of $\M/SO(3)$, called the center of symmetry.

If $\Sigma^{\scriptsize{(3)}}\simeq \R^3$, then $\Gamma$ is non-empty and has one connected component. 

If  $\Sigma^{\scriptsize{(3)}}\simeq \mathbb{S}^2\times \R$, then $\Gamma$ is empty. 

 If $\Sigma^{\scriptsize{(3)}}\simeq \mathbb{S}^3$, then $\Gamma$ is non-empty and has two connected components. 
\end{prop}

To extend Proposition \ref{choquet} to the case in which $\Sigma^{\scriptsize{(3)}}\simeq\mathbb{S}^2\times \mathbb{S}^1$, we consider the universal cover of $\Sigma^{\scriptsize{(3)}}$, which has topology $\mathbb{S}^2\times \R$. 

Since our spacetime manifold $\M$ admits a spherically symmetric spacelike Cauchy hypersurface $\Sigma^{\scriptsize{(3)}}$, if $\pi: \M \rightarrow \M/SO(3)$ is the standard projection map, then $\Sigma = \pi(\Sigma^{\scriptsize{(3)}})$  is a connected spacelike curve `downstairs' in $\M/SO(3)$, as $\Sigma^{\scriptsize{(3)}}$ is preserved under the $SO(3)$-action.

Let $\Q^+=D_{\M}^+(\Sigma^{\scriptsize{(3)}})/SO(3)$, i.e., the quotient by $SO(3)$ of the future domain of dependence of $\Sigma^{\scriptsize{(3)}}$ in $\M$. For the asymptotically flat initial data with one end  considered in Theorem \ref{thm:main}, we have $\Sigma^{\scriptsize{(3)}}\simeq \R^3$.  In this case, past-directed causal curves issuing from a point $p \in \Q^+$ will terminate at a single point on $\Sigma\cup \Gamma$.  Thus, $\Sigma$ is the past (spacelike) boundary of $\Q^+$ as depicted below.

\begin{figure}[h]
\begin{center}
\includegraphics[scale=.5]{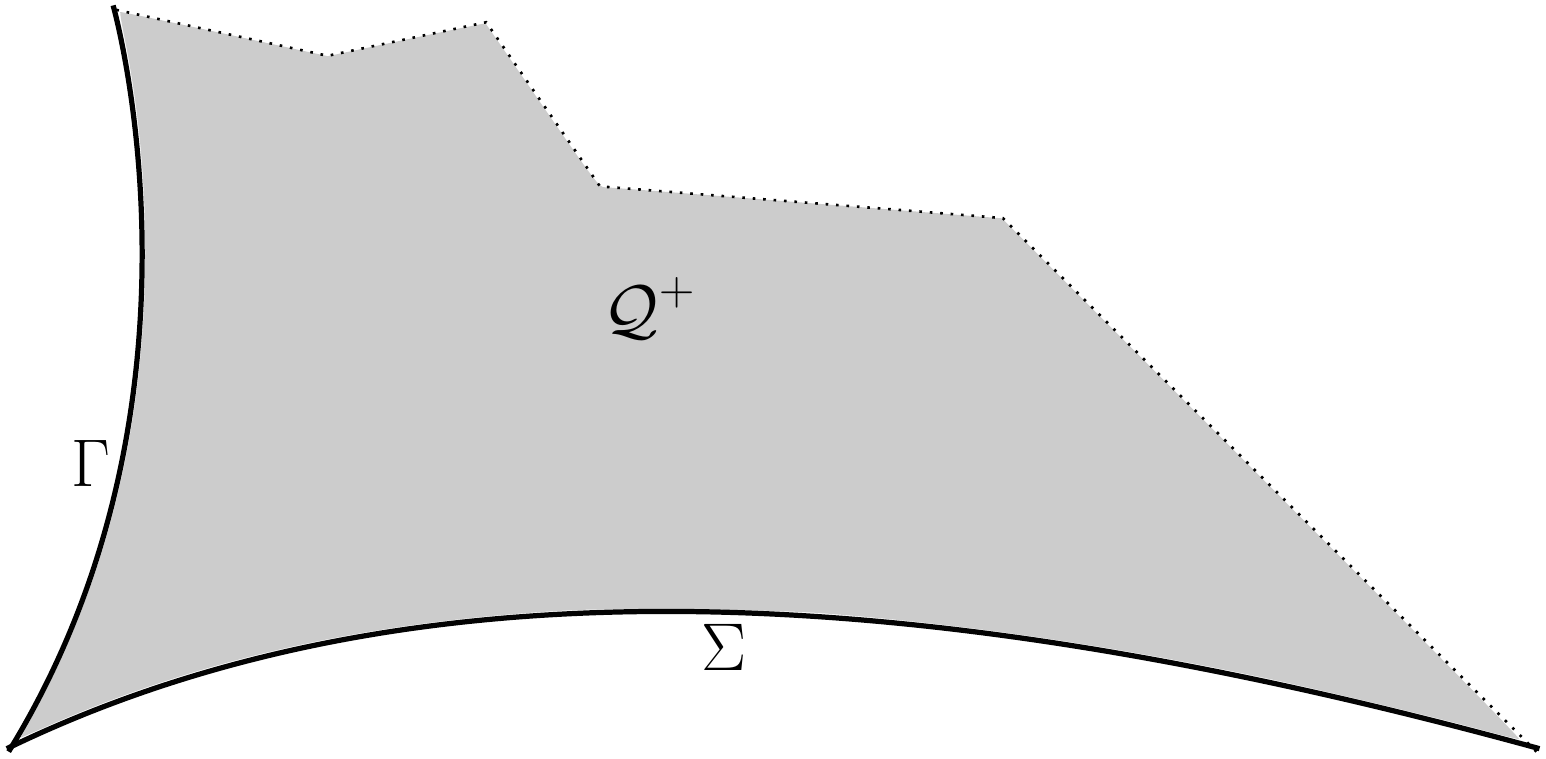}
\end{center}
\end{figure}

As conformal embeddings preserve causal structure, the standard double null co-ordinates $(u,v)$ of the ambient $\R^{1+1}$ provide a global chart on $\Q^+$,  and its metric can be written $g_{ab}\dd x^a\dd x^b = -\Omega^2\dd u\dd v$.  In particular, the full metric on $\M$ is given by
\beqn\label{metric}
g_{\mu\nu}\dd x^{\mu}\dd x^{\nu} = -\Omega^2\dd u\dd v + r^2 \dd h^2
\eeqn
where $\dd h^2= \dd \theta^2 + \sin^2\theta\dd \varphi^2$ is the standard metric on $\mathbb{S}^2$ and $r:\Q^+\rightarrow [0, \infty)$ is the area-radius function given by
\beqn\nonumber
r(p) = \sqrt{\frac{\textrm{Area}(\pi^{-1}(p))}{4\pi}}.
\eeqn

\noindent The embedding into $\R^{1+1}$ shall be chosen, in the case that $\Gamma$ is non-empty and connected, such that $\partial/\partial u$ points `inward' towards $\Gamma$ and $\partial/\partial v$ points `outward' away from $\Gamma$.

In what follows, $\Q^+$ will be identified with its image under the embedding into $\R^{1+1}$.

\subsection{Determined system}\label{sec:pre/determined}
In order to work with a determined Einstein-Maxwell-Klein-Gordon system of equations, we must fix the inherent gauge freedom that arises.

\subsubsection{Choice of gauge}\label{sec:pre/determined/gauge}

For $\e\neq 0$, as seen from \S\ref{prelimbundle}, the system of equations (\ref{RMN})--(\ref{eqn:kg}) is invariant under local $U(1)$ gauge transformations
\beqna\nonumber
\phi &\rightarrow& e^{-\e\ii \chi} \phi\\
A_{\mu} &\rightarrow& A_{\mu} + \partial_{\mu}\chi,\nonumber
\eeqna
for any smooth real-valued function $\chi$. Fixing a gauge amounts to fixing $\chi$. 

Let us recall that the electromagnetic field strength $2$-form $F$ can be expressed as
\beqn\nonumber
F_{\mu\nu} = \partial_{\mu}A_{\nu} - \partial_{\nu}A_{\mu}.
\eeqn 
As the $SO(3)$-action preserves $F$, there can be only two non-vanishing components:
\beqn\nonumber
F = F_{uv} \dd u\wedge \dd v + F_{\theta\varphi}\dd \theta\wedge \dd\varphi.
\eeqn
Since every $2$-form is proportional to the volume element on $\mathbb{S}^2$, it then follows from the fact that $\dd F = 0$, i.e., $F$ is closed, that
\beqn\label{defmagn}
F_{\theta\varphi} = \Qm \sin\theta
\eeqn
for some constant $\Qm$ such that (cf.~the Dirac quantization condition)
\beqn\nonumber
\e\Qm  \in \frac{1}{2}\mathbb{Z},
\eeqn
which is to be interpreted as the (topological) magnetic charge of $\M$.  In the case, moreover, that $A$ is spherically symmetric, then $A$ must be of the form
\beqn\nonumber
A = A_u(u,v) \dd u + A_v(u,v) \dd v,
\eeqn
from which it follows that $F_{\theta\varphi} = 0$, i.e., $\Qm= 0$. We deduce that if $\e \neq 0$, then $F$ has trivial (de Rham) cohomology, i.e., there is a globally defined 1-form $A$ (on $\M$) such that $\dd A=F$, i.e., $F$ is exact.  Indeed, given $F$, we define uniquely a gauge, i.e., define $\chi$, on $\M$ by imposing the following conditions:
\beqna
A_u(u,v)\big|_{\pi^{-1}(\Sigma\cup\Gamma)} &=&0\nonumber\\
A_v(u,v) &=&0\nonumber.
\eeqna 

\subsubsection{Quantization condition}\label{quantize}
In view of the fact that spherical symmetry requires that $\Qm=0$ whenever $\e\neq 0$, it is no longer necessary to impose the restriction $\e\in \mathbb{Z}$.  In particular, as a consequence of spherical symmetry, we may view electromagnetism not as a \emph{compact} $U(1)$ gauge theory (for which charge is quantized) but as a \emph{non-compact} $\R$ gauge theory (for which charge is unquantized). In the sequel, we will take $\e \in \R$.   

\subsubsection{Energy-momentum tensor}\label{energymomentum}
Let us note that the energy-momentum tensor (\ref{TEM})--(\ref{TKG}) immediately gives
\beqna
T_{uu} &=& \D_u\phi\left(\D_u\phi\right)^{\dagger}\nonumber\\
T_{vv} &=& \pv\phi\left(\pv\phi\right)^{\dagger}\nonumber.\nonumber
\eeqna
To compute the other components of the energy-momentum tensor, let us consider the scalar invariant $F_{\mu\nu}F^{\mu\nu}$. We note that in spherical symmetry,
\beqna\nonumber
F_{\mu\nu}F^{\mu\nu} &=& g^{\mu\alpha}g^{\nu\beta}F_{\mu\nu}F_{\alpha\beta}\\
&=&2g^{\theta\theta}g^{\varphi\varphi}\left(F_{\theta\varphi}\right)^2 - 2\left(g^{uv}\right)^2\left(F_{uv}\right)^2\nonumber\\
&=&\frac{2}{r^4}\Qm^2 - 8\Omega^{-4}\left(F_{uv}\right)^2.\nonumber
\eeqna
Let us define the quantity $\Qe$ (called the electric charge) by
\beqn\label{defnbeta}
\Qe =2r^2\Omega^{-2}F_{uv}. 
\eeqn
Noting the identity
\beqn\nonumber
F_{\mu\nu}F^{\mu\nu} = -\frac{2}{r^4}\left(\Qe^2 - \Qm^2\right),
\eeqn
the $(u,v)$-component of (\ref{TEM})--(\ref{TKG}) then reads
\beqn\label{tuv}
T_{uv} = \frac{1}{4}\Omega^2 \left(\m^2\phi\phi^{\dagger} + \frac{\Qe^2 + \Qm^2}{4\pi r^4}\right).
\eeqn
Lastly, we note that the spherical components of the energy-momentum tensor are given by
\beqn\nonumber
T_{\varphi\varphi}\sin^{-2}\theta = T_{\theta\theta} = r^2\Omega^{-2}\left(\D_u\phi \left(\pv \phi\right)^{\dagger} + \pv\phi\left(\D_u\phi\right)^{\dagger} - \frac{1}{2}\m^2\Omega^2\phi\phi^{\dagger}\right) + \frac{\Qe^2+\Qm^2}{8\pi r^2}.
\eeqn

Since $T_{uu}\geq0$ and $T_{vv}\geq0$, the Einstein-Maxwell-Klein-Gordon system obeys the null energy condition.  If, in addition, $\m^2\geq0$, then $T_{uv}\geq0$ and thus the model obeys the dominant energy condition.  We shall assume $\m^2\geq 0$ henceforth.\footnote{In the case of the Einstein-Klein-Gordon system ($\e = F_{\mu\nu}=0$), the assumption that $\m^2\geq0$ can be dropped in the proof of Theorem \ref{thm:gext}. See footnote \ref{foot:dominant}.}

\subsubsection{Reduced system of equations}\label{sec:pre/determined/eqns}
The spherically symmetric Einstein-Maxwell-Klein-Gordon system reduces (in the gauge of \S \ref{sec:pre/determined/gauge}) to the following equations on $\Q^+$ for $C^2$-functions $(r, \Omega, \phi, A_u)$:

 \beqn\label{eqn:ruv}
r\pv\pu r= -\frac{1}{4}\Omega^2 -\pv r\pu r+ \m^2\pi r^2 \Omega^2\phi\phi^{\dagger} + \frac{1}{4} \Omega^2 r^{-2}Q^2
 \eeqn
 \beqn\label{logO}
r^2 \pu\pv \log \Omega = -2\pi r^2\left(\D_u\phi\left(\pv\phi\right)^{\dagger}+\pv\phi\left(\D_u\phi\right)^{\dagger}\right)- \frac{1}{2}\Omega^{2}r^{-2}Q^2+\frac{1}{4}\Omega^2 + \pu r\pv r
 \eeqn
 \beqn\label{eqn:conuu}
  \partial_u(\Omega^{-2}\partial_ur) = -4\pi r  \Omega^{-2}\D_u\phi\left(\D_u\phi\right)^{\dagger} 
 \eeqn
 \beqn\label{eqn:convv}
 \partial_v(\Omega^{-2}\partial_vr) = -4\pi r  \Omega^{-2}\pv\phi\left(\pv\phi\right)^{\dagger} \eeqn
 
\beqn\label{betaa}
\partial_u\Qe = 2\pi\e\ii r^2 \left(\phi\left(\D_u\phi\right)^{\dagger}-\phi^{\dagger}\D_u\phi\right)
\eeqn
\beqn\label{beta}
\partial_v\Qe = 2\pi\e\ii r^2\left(\phi^{\dagger}\pv\phi-\phi\left(\pv\phi\right)^{\dagger}\right)
\eeqn
 \beqn\label{KG}
 \partial_u\partial_v \phi +r^{-1}( \pu r\partial_v\phi + \pv r\partial_u\phi ) + \e\ii \Psi(A)
= -\frac{1}{4}\m^2\Omega^{2}\phi
\eeqn
\beqn\label{last}
\Psi (A)=A_u\left(\phi r^{-1}\pv r + \pv \phi\right) - \frac{1}{4}\Omega^2r^{-2}\Qe\phi
\eeqn
\beqn\label{totalcharge}
Q^2 = \Qe^2+\Qm^2
\eeqn
\beqn\label{fuv}
\Qe = -2 r^2 \Omega^{-2} \pv A_u.
\eeqn

\subsubsection{Scaling properties}\label{sec:scale}  Lastly, we consider the scaling properties of the reduced system (\ref{eqn:ruv})--(\ref{fuv}). The choice of co-ordinates induces a natural scaling of the metric (\ref{metric}), whereby we take
\beqn\nonumber
[\Omega] = [L]^0\quad\quad \textrm{and}\quad\quad [r] = [u]=[v] = [L]^1.
\eeqn
The wave equation (\ref{KG}) then requires
\beqn\nonumber
[\m] = [L]^{-1},
\eeqn
and (\ref{eqn:ruv}), in turn, gives
\beqn\nonumber
[\phi] = [L]^0\quad\quad \textrm{and}\quad\quad [Q]= [\Qe] = [\Qm] = [L].
\eeqn
From (\ref{betaa}) and (\ref{fuv}) we then have
\beqn\nonumber
[A_u] = [L]^0 \quad\quad\textrm{and}\quad\quad [\e] = [L]^{-1}.
\eeqn
Note, in particular, that
\beqn\nonumber
[\e Q] = [L]^0.
\eeqn

It may seem somewhat surprising that the coupling parameter $\e$ (resp., $\m$) does not have the same dimension of charge (resp., mass). This is a consequence of imposing natural units ($c = G= 1$) at the level of the equations of motions.  While this is a commonly employed mathematical convenience (within the context of general relativity), at the level of the Einstein-Hilbert action, however, $G$ is a dimension\emph{ful} quantity. 
\section{Rudimentary boundary characterization}\label{spaceboundary}\label{sec:rud}
Let $\Q^+$ be as in Theorem \ref{thm:main}.
The maximal future development $\mathcal{Q}^+$ will have boundary $\mathcal{B}^+  = \overline{\Q^+}\backslash\Q^+$, which we aim to now characterize.\footnote{This boundary is induced by the conformal embedding of $\Q^+$ into $\R^{1+1}$ and is \emph{not} the boundary in the sense of a manifold-with-boundary, which is $\Sigma\cup \Gamma$.}  

The following description of $\mathcal{B}^+$ holds very generally for globally hyperbolic spherically symmetric spacetimes with one asymptotically flat end where the matter model obeys the null energy condition. The boundary $\mathcal{B}^+$ can be decomposed into two components: the boundary `emanating from spacelike infinity' and the boundary `emanating from first singularities'.

We begin with a few preliminary results.

\subsection{Nowhere anti-trapped initial data}\label{sec:pre/noanti}
The statement of Theorem \ref{thm:main} assumes that initial data are prescribed such that no anti-trapped regions are present, i.e.,~
\beqn\nonumber
\pu r <0
\eeqn
along $\Sigma$.  We thus assume this  throughout \S \ref{sec:rud}. Data satisfying this property, motivated by Christodoulou in \cite{DC95}, preclude anti-trapped regions from forming in their future development; that is, anti-trapped regions are non-evolutionary. This is given in
\begin{prop}\label{prop: antitrapped} In the notation of Theorem \ref{thm:main}, if $\pu r<0$ along $\Sigma$, then $\pu r<0$ everywhere in $\Q^+$.
\end{prop}
\begin{proof}
By global hyperbolicity, the past-directed null segment of constant-$v$ issuing from $(u,v)\in \Q^+$ will terminate on $\Sigma$.  In particular, there exists $u' \leq u$ such that $\pu r(u', v) <0$, by assumption.  The result follows by integrating (\ref{eqn:conuu}) along $[u', u]\times\{v\}$.
\end{proof}
We emphasize that Proposition \ref{prop: antitrapped} only relies on the null energy condition.
\subsection{Spacetime volume}\label{sec:pre/volume}
Let us define the sets
\beqna\nonumber
\mathcal{U}&=& \{u: \exists v ~~\textrm{s.t.}  ~~(u,v)\in \Q^+\}\\
\mathcal{V} &=& \{v: \exists u ~~\textrm{s.t.}  ~~(u,v)\in \Q^+\}.\nonumber
\eeqna
Since $\Q^+$ is a bounded subset of $\R^{1+1}$, we have that
\beqna\nonumber
U &=& \sup \mathcal{U} <\infty\nonumber \\
V &=& \sup \mathcal{V}<\infty\nonumber,
\eeqna
and
\beqna
U_0 &=& \inf \mathcal{U} > -\infty\nonumber \\
V_0 &=& \inf \mathcal{V}> -\infty\nonumber.
\eeqna

For any $u'\leq U$ and $v'< V$ such that $(u', v')\in \Q^+$, it follows by Proposition \ref{prop: antitrapped} and monotonicity (\ref{eqn:conuu}) that the spacetime volume of the region
\beqn\nonumber
\Q^+  (u', v') = \Q^+\cap \{u\leq u'\}\cap\{v\leq v'\}
\eeqn
satisfies
\beqna\nonumber
\int_{\Q^+(u', v')}2~\dd \textrm{Vol} &= &\int_{\mathcal{Q}^+\cap\{v\leq v'\}}\int_{\mathcal{Q}^+\cap \{u\leq u'\}}(-\Omega^{-2}\pu r)^{-1}(-\pu r)~\dd u \dd v\nonumber\\
&\leq& \sup_{\Sigma\cap \{v\leq v'\}}\Omega^2(-\pu r)^{-1}\int_{\Q^+\cap\{v\leq v'\}}\int_{\Q^+\cap \{u\leq u'\}}(-\pu r)~\dd u\dd v\nonumber\\
&\leq& \sup_{\Sigma\cap \{v\leq v'\}}\Omega^2(-\pu r)^{-1}\cdot \sup_{\Sigma\cap \{v\leq v'\}} r\int_{\Q^+\cap\{v\leq v'\}}~\dd v\nonumber\\
&\leq&\left(\sup_{\Sigma\cap \{v\leq v'\}}\Omega^2(-\pu r)^{-1}\cdot \sup_{\Sigma\cap \{v\leq v'\}} r \right)\left(V - V_0\right)<\infty.\nonumber
\eeqna
We have thus established, in particular, 
\begin{prop}\label{prop:finite_volume} If $p\in \overline{\Q^+}\backslash\{v=V\}$, then the region $J^-(p)\cap \Q^+$ has finite spacetime volume. 
\end{prop}

In fact, we note that the domain of dependence of any compact subset of  $\Sigma^{\scriptsize{(3)}}$ has finite spacetime volume `upstairs' as well. To see this, we simply compute that
\beqn\nonumber
\int_{\M} 2~\dd \textrm{Vol}_{\M} = \int_{\M}r^2\Omega^2 \sin\theta~\dd u\dd v\dd \theta\dd\varphi= 8\pi \int_{\Q^+}r^2~\dd \textrm{Vol}_{\Q^+}.
\eeqn
Since $\pu r<0$ in $\Q^+$, we have for any $v' <V$
\beqn\nonumber
\sup_{\Q^+\cap \{v\leq v'\}} r\leq \sup_{\Sigma\cap \{v\leq v'\}} r< \infty.
\eeqn
This establishes the claim.

\subsection{Boundary `emanating from spacelike infinity'}\label{sec:rud/inf}
The initial data curve $\Sigma$ acquires a unique limit point $i^0\in\mathcal{B}^+$ called \emph{spacelike infinity}.  

Recall the notation of \S \ref{sec:pre/volume}.  Let $\mathcal{U}_{\infty}$ denote the set of all $u$ given by
\beqn\nonumber
\mathcal{U}_{\infty} = \left\{u: \sup_{\mathcal{V}} r(u,v) = \infty\right\},
\eeqn
which may be \emph{a priori} empty (even if $r\rightarrow \infty$ along $\Sigma$).
For each $u\in \mathcal{U}_{\infty}$, there exists a unique $v_{\infty}(u)$ such that
\beqn\nonumber
(u, v_{\infty}(u)) \in \mathcal{B}^+.
\eeqn
Define \emph{future null infinity} $\mathcal{I}^+$ by
\beqn\nonumber
\mathcal{I}^+ = \bigcup_{u\in \mathcal{U}_{\infty}} (u, v_{\infty}(u)).
\eeqn
For data of compact support along $\Sigma$, an application of Birkhoff's theorem (cf.~\cite{MH96}) and the domain of dependence property will ensure that $\mathcal{I}^+$ is non-empty (for the geometry of $\Sigma$ coincides with Reissner-Nordstr\"om outside the support of $\phi$), allowing us then to appeal to 
\begin{prop}\label{iplus} If non-empty, $\mathcal{I}^+$ is a connected ingoing null segment with past limit point $i^0$.
\end{prop}
\begin{proof} Let $i^0 = (U,V)$.  Since, by Proposition \ref{prop: antitrapped}, $\pu r<0$ in $\Q^+$, for all $v_0 <V$ we have
\beqn\nonumber
\sup_{\Q^+ \cap \{v \leq v_0\}} r \leq \sup_{\Sigma\cap \{v \leq v_0\}} r.
\eeqn
Thus, $\mathcal{I}^+ \cap \{v = v_0\} = \emptyset$.  In particular, $\mathcal{I}^+ \subset \{v= V\}$.

Consider $(u_0, V) \in \mathcal{I}^+$, and let $u <u_0$ be such that $u\in \mathcal{U}$.  Using, again, the fact that $\pu r<0$ in $\Q^+$, we obtain
\beqn\nonumber
\lim_{v\rightarrow V} r(u,v)\geq \lim_{v\rightarrow V} r(u_0, v) = \infty.
\eeqn
That is, $(u, V) \in \mathcal{I}^+$.
\end{proof}

Alternatively, the non-emptiness of $\mathcal{I}^+$ easily follows if the fall-off rate of data at spacelike infinity is suitably tame.\footnote{We will not discuss the weakest possible notion of an admissible data set here, but the reader may wish to consider the issue further.  It is certainly sufficient for our purpose to consider the case where $\phi$ is compactly supported.}  An appropriate notion, therefore, of `asymptotically flat'  in Theorem \ref{thm:main} can \emph{ab initio} ensure that $\mathcal{I}^+$ is non-empty.  We will assume that we have built in the requirement  $\mathcal{I}^+\neq\emptyset$ into our definition of asymptotically flat in Theorem \ref{thm:main}.

We define $i^{\square}\in \mathcal{B}^+$ as the future limit point of $\mathcal{I}^+$, noting \emph{a priori} we could have  $i^{\square}\in \mathcal{I}^+$.  

By causality, we have a (possibly empty) half-open null segment\footnote{We choose the notation $\mathcal{N}_{i^+}$ as this set can be non-empty only if $i^{\square} = i^+$ (cf.~Statement III of Theorem \ref{thm:main}).} $\mathcal{N}_{i^+}$ emanating from (but not including) $i^{\square}$.
\begin{center}
\includegraphics[scale = .55]{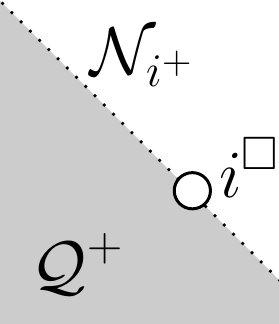} 
\end{center}

\subsection{Boundary `emanating from first singularities'}\label{sec:rud/finite}

We now introduce the notion of a point $p \in \overline{\Q^+}\backslash \overline{\mathcal{I}^+}$ being a `first singularity'.

\begin{defn} Let $p \in \overline{\Q^+}$.  The causal set $J^-(p)\cap \Q^+\subset \Q^+$ is said to be \emph{compactly generated} if there exists a compact subset $X\subset \Q^+$ such that
\beqn\nonumber
J^-(p) \subset \pi(D^+_{\M}(\pi^{-1}(X))) \cup J^-(X).
\eeqn
\end{defn}

If $p\in \overline{\mathcal{Q}^+}\backslash{\overline{\Gamma}}$, then for $J^-(p)\cap \Q^+$ to be compactly generated means that there exists a causal rectangle
\beqn\nonumber
\Dm = \left(J^-(p)\cap J^+(q)\right)\backslash\{p\} \subset \Q^+
\eeqn
for some $q\in\left( I^-(p)\cap \Q^+\right)\backslash\{p\}$.

If $p\in \Gamma$, $J^-(p)\cap \Q^+$ is always compactly generated.  Let $b_{\Gamma}$ denote the unique future limit point of $\Gamma$ in $\overline{\Q^+}\backslash \Q^+$.  If $p = b_{\Gamma}$, then $J^-(p)\cap \Q^+$ is compactly generated as long as $b_{\Gamma}\not\in \mathcal{N}_{i^+}\cup i^{\square}$.

\begin{defn}
$p\in \mathcal{B}^+$ is said to be a \emph{first singularity} if the causal set $J^-(p)\cap \Q^+$ is compactly generated and if any compactly generated proper causal subset of $J^-(p)\cap \Q^+$ is of the form $J^-(q)$ for a $q\in \Q^+$.  
\end{defn}

\begin{defn} Let $\mathcal{B}_1^+\subset \mathcal{B}^+$ be the set of all first singularities.
A first singularity $p\in \mathcal{B}_1^+ \backslash b_{\Gamma}$ will be called \emph{non-central}.  
\end{defn}

If $\mathcal{B}_1^+=\emptyset$, then $\mathcal{B}^+$ is given by
\beqn\nonumber
\mathcal{B}^+ = b_{\Gamma}\cup \mathcal{N}_{i^+}\cup i^{\square} \cup \mathcal{I}^+\cup i^0,
\eeqn
where $b_{\Gamma}$ and the future endpoint of $\mathcal{N}_{i^+}\cup i^{\square}$ coincide.

We note that if $\mathcal{B}_1^+\neq \emptyset$, then $b_{\Gamma}\in \mathcal{B}_1^+$. We then define the boundary `emanating from first singularities' as the union of

\begin{itemize}
\item [1.] $b_{\Gamma}$ and a (possibly empty) half-open null segment $\mathcal{N}_{\Gamma}$ emanating from (but not including) $b_{\Gamma}$; and,
\begin{center}
\includegraphics[scale=.55]{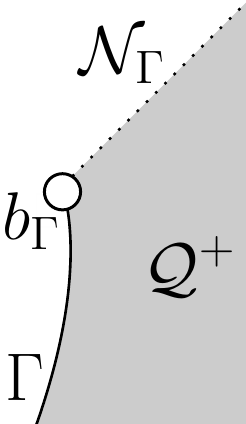}
\end{center}
\end{itemize}
\begin{itemize}
\item [2.] the set $\mathcal{B}_1^+\backslash b_{\Gamma}$ and for every $p\in \mathcal{B}_1^+\backslash b_{\Gamma}$, two (possibly empty) half-open null segments $\mathcal{N}_p^j$ emanating from (but not including) $p$.
\begin{center}
\includegraphics[scale = .55]{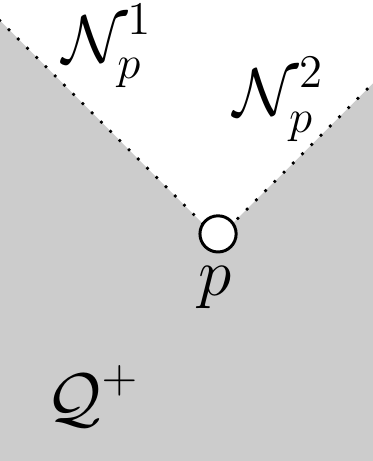} 
\end{center}
\end{itemize}

\subsection{Summary of preliminary boundary decomposition}\label{sec:rud/sum}
We summarize the boundary decomposition thus far in
\begin{prop}\label{rudb}   Let $(\M = \Q^+\times_r \mathbb{S}^2, g_{\mu\nu}, \phi, F_{\mu\nu})$ denote the maximal future development of initial data as in Theorem \ref{thm:main}. The Penrose diagram has boundary $\mathcal{B}^+$ admitting the (not necessarily disjoint) decomposition
\beqn\nonumber
\mathcal{B}^+ = b_{\Gamma} \cup \mathcal{N}_{\Gamma}\cup \left(\bigcup_{p\in \mathcal{B}_1^+\backslash b_{\Gamma}} \{p\}\cup \mathcal{N}_p^1 \cup \mathcal{N}_p^2\right)\cup \mathcal{N}_{i^+} \cup i^{\square} \cup \mathcal{I}^+ \cup i^0,
\eeqn
where the sets are as described previously. Moreover, the following hold:
\begin{itemize}
\item [1.] If $\mathcal{B}_1^+ = \emptyset$, then $\mathcal{N}_{\Gamma} = \emptyset$.
\item [2.] If $\mathcal{B}_1^+\backslash b_{\Gamma} = \emptyset$, then the future endpoint of $b_{\Gamma}\cup \mathcal{N}_{\Gamma}$ coincides with the future endpoint of $\mathcal{N}_{i^+}\cup i^{\square}$.
\end{itemize}
\end{prop}

\section{The generalized extension principle}\label{sec:general_emkg}
In this section, we shall briefly shift gears and prove the generalized extension principle of Theorem \ref{thm:gext}. Here, $\Q^+$ will denote the maximal development of the more general initial data as in the statement of Theorem \ref{thm:gext}. In \S\ref{sec:proof_main}, we will return to considering the development of initial data as in the statement of Theorem \ref{thm:main}.

\subsection{Another local existence result}\label{sec:general_emkg/local}
The extension criterion, given in the next subsection, relies on an auxiliary (to Proposition \ref{choquet}) local existence statement, which is stated directly at the level of the spherically symmetric reduction. Formulated as a double characteristic initial value problem, the following result is suitable for our purpose.   
    
\begin{prop}\label{extension} Let $k\geq 0$ and consider a set $X = [0, d]\times\{0\} \cup \{0\}\times[0, d]$.  On $X$, let $r$ be a positive function and $A_u$ be a function that are $C^{k+2}$; let $\Omega$ be a positive function and $\phi$ be a function that are $C^{k+1}$.  Suppose that equations (\ref{eqn:conuu}) and (\ref{eqn:convv}) hold initially on $[0, d]\times\{0\}$ and $\{0\}\times[0, d]$, respectively.  Let $|\cdot|_{n, u}$ denote the $C^n(u)$-norm on $[0, d]\times\{0\}$; similarly, let $|\cdot|_{n, v}$ denote the $C^n(v)$-norm on $\{0\}\times[0, d]$. Define
\beqn\nonumber
N = \sup \{|\Omega|_{1, u}, |\Omega|_{1, v} |\Omega^{-1}|_0, |r|_{2, u},|r|_{2, v},  |r^{-1}|_0, |\phi|_{1, u}, |\phi|_{1, v}, |A_u|_0, |A_u|_{2,v}, |\pv A_u|_{1, u} \}.
\eeqn
Then, there exists a $\delta>0$ depending only on $N$, $C^{k+2}$-functions (unique amongst $C^2$-functions) $r$ and $A_u$, and $C^{k+1}$-functions (unique amongst $C^1$-functions) $\Omega$ and $\phi$ satisfying equations (\ref{eqn:ruv})--(\ref{fuv}) in $[0, \delta^*]\times [0, \delta^*]$, where $\delta^* = \min\{d, \delta\}$, such that the restriction of these functions to $X$ is as prescribed.   
\end{prop}
The above result is quite general in the sense that it does not depend on the structure of the non-linear terms in (\ref{RMN})--(\ref{eqn:kg}).  The proof is omitted, but can be obtained via standard arguments  (cf.~the appendix of \cite{MDAR05}).

\subsection{An extension criterion}\label{sec:general_emkg/ext}
Let $p=(U,V)\in \overline{\Q^+}$, $q = (U',V')\in \left(I^-(p)\cap\Q^+\right)\backslash \{p\}$ and 
\beqn\nonumber
\mathcal{D} = \left(J^+(q) \cap J^-(p)\right)\backslash \{p\} \subset \Q^+
\eeqn 
be as in the statement of Theorem \ref{thm:gext}.  Then, the compact set
\beqn\label{x}
X =  [\ue, U] \times \{\ve\}\cup \{\ue\}\times [\ve, V] 
\eeqn
satisfies $X\subset \Q^+$.

\begin{center}
\includegraphics[scale = .55]{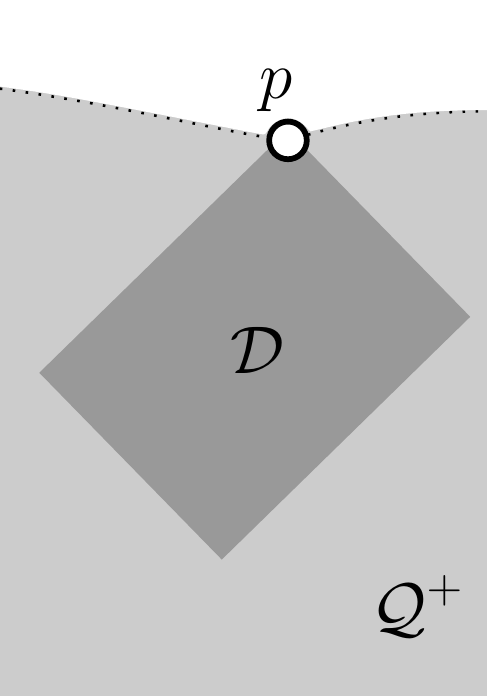}
\end{center}

Given a subset $Y \subset \Q^+$, we define a `norm' $N: 2^{\Q^+}\rightarrow [0, \infty]$ given by
\beqn\nonumber
N(Y) = \sup \{|\Omega|_{1}, |\Omega^{-1}|_0, |r|_{2}, |r^{-1}|_0, |\phi|_{1},  |A_u|_{0}, |\Qe|_0  \},
\eeqn
where $|f|_n$ denotes the restriction of the $C^n$-norm on $\Q^+$ to $Y$.

In view of Proposition \ref{extension}, we formulate an extension criterion in

\begin{prop}\label{naext}
Let $p = (U,V)\in \overline{\Q^+}\backslash \Q^+$ and  $q = (U',V')\in \left(I^-(p)\cap\Q^+\right)\backslash \{p\}$ be such that
\beqn\nonumber
\mathcal{D} = \left(J^+(q) \cap J^-(p)\right)\backslash \{p\} \subset \Q^+.
\eeqn
Then, 
\beqn\nonumber
N(\Dm) = \infty.
\eeqn 
\end{prop}

\begin{proof} We prove the contrapositive.  Suppose that $N = 2N(\Dm) < \infty$.  We will show that $p = (U,V)\in \Q^+$. Corresponding to the value $N$, let $\delta>0$ be given as in Proposition \ref{extension}.  Consider the point $(U - \frac{1}{2}\delta, V - \frac{1}{2}\delta)$.  Taking $\delta$ suitably small, we can assume that this point is in $\Q^+$.  Translate the co-ordinates so that this point is called $(0,0)$.  Since $\Q^+$ is, by definition, an open set, there is by continuity a $\delta^*\in (\frac{1}{2}\delta, \delta)$ such that
\beqn\nonumber
X^* = \{0\}\times[0, \delta^*] \cup [0, \delta^*]\times \{0\}\subset \Q^+.
\eeqn
Moreover, the assumptions of Proposition \ref{extension} hold on this set $X^*$.  Hence, there exists a unique solution in
\beqn\nonumber
\tilde{\Dm} = [0, \delta^*]\times [0, \delta^*],
\eeqn  
which coincides with the previous solution on $\Dm\cap \tilde{\Dm}$ by uniqueness.

\begin{center}
\includegraphics[scale=.55]{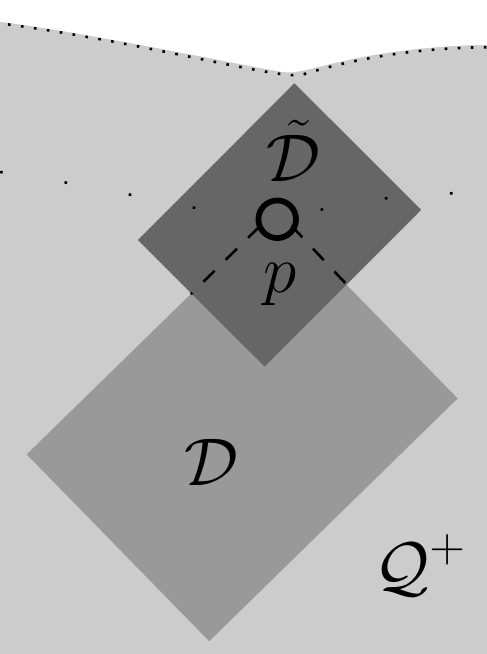}
\end{center}
As $\tilde{\Dm}\cup \Q^+$ is the quotient of a development of initial data, we must have, by maximality of $\Q^+$, that $\tilde{\Dm}\cup \Q^+ \subset \Q^+$.  Thus, in particular, $p\in \Q^+$.
\end{proof}

\subsection{Proof of Theorem \ref{thm:gext}: generalized extension principle}\label{sec:general_emkg/proof_gext}

Define
\beqn\label{Wbound}
W = \int_{\ve}^V\int_{\ue}^U\Omega^2~\dd u\dd v<\infty
\eeqn
and let $r_0$ and $R$ be constants such that
\beqn\nonumber
0<r_0 \leq r(u,v)\leq R<\infty
\eeqn
for all $(u,v)\in\Dm$.

Let us recall (\ref{defnbeta}) in which
\beqn\nonumber
\Qe =2r^2\Omega^{-2}F_{uv}. 
\eeqn
We introduce further the notation
\beqna
\lambda &=& \pv r\nonumber\\
\nu &= & \pu r\label{nu}\nonumber\\
\theta_A &=& r\pv \phi\label{theta}\nonumber\\
\zeta_A &=& r\D_u \phi\label{zeta}\nonumber.
\eeqna
By compactness, and the regularity of a solution as given in Theorem \ref{thm:gext}, there exists finite constants $N, \Lambda, \Phi, \Theta, Z, B, A, H$ such that along $X = [\ue, U] \times \{\ve\}\cup \{\ue\}\times [\ve, V]$ the following bounds hold: 
\beqn\nonumber
| r\nu| \leq N
\eeqn
\beqn\nonumber
|r\lambda| \leq \Lambda
\eeqn
\beqn\nonumber
|r\phi|\leq \Phi
\eeqn
\beqn\nonumber
|\theta_A|\leq\Theta
\eeqn
\beqn\nonumber
|\zeta_A| \leq Z
\eeqn
\beqn\nonumber
|\Qe|\leq B
\eeqn
\beqn\nonumber
|A_u|\leq A
\eeqn
\beqn\nonumber
|\pu \nu|, |\pv\lambda|, |\pu \Omega|, |\pv\Omega|, |\log\Omega^2| \leq H.
\eeqn
In view of the extension criterion in Proposition \ref{naext}, it suffices to show that similar uniform bounds hold in $\Dm$ in order to prove Theorem \ref{thm:gext}.
\label{sec:proof_gext}

\subsubsection{\emph{a priori} integral estimates}\label{sec:general_emkg/proof_gext/integral}
Recalling the notation of (\ref{tuv}), we begin by integrating (\ref{eqn:ruv}) in $u$ and $v$ to obtain
\beqna\nonumber
\int_{\ve}^{V} \int_{\ue}^{U} 4\pi r^2 T_{uv}~\dd u\dd v &=& \int_{\ve}^{V} \int_{\ue}^{U} \pu (r\lambda)~\dd u\dd v + \frac{1}{4}\int_{\ve}^{V} \int_{\ue}^{U} \Omega^2~\dd u\dd v\\
&\leq& (R^2 - r_0^2) + \frac{1}{4}W=: C_0\label{tuvest},
\eeqna
in view of the fact that $\pu (r\lambda) = \frac{1}{2}\pu\pv r^2$.

Using (\ref{tuvest}), we obtain \emph{a priori} spacetime integral bounds on the matter fields.  In particular, we have\footnote{It is here, and only here, where we use the dominant energy condition in the proof of Theorem \ref{thm:gext}, i.e.,~the assumption that $\m^2\geq0$.  In the case of the Einstein-Klein-Gordon system ($\e= F_{\mu\nu}=0$), we note that (\ref{totalb2}) holds on account of $T_{uv}$ having only one component (since $Q = 0$).\label{foot:dominant}}
\begin{subequations}
\beqna
\pi\m^2 \int_{\ve}^V \int_{\ue}^U \Omega^2 |r\phi|^2~\dd u\dd v&\leq& C_0\label{b2a}\\
\frac{1}{4}\int_{\ve}^V \int_{\ue}^U \Omega^2 r^{-2}Q^2~\dd u\dd v&\leq& C_0\label{b2}.
\eeqna\label{totalb2}
\end{subequations}
Note that (\ref{b2}) holds with $Q^2$ replaced by $\Qe^2$ on account of (\ref{totalcharge}).

We now integrate (\ref{eqn:ruv}) in $u$ to yield the pointwise estimate
\beqn\label{boundla}
\sup_{\ue \leq u\leq U}|r\lambda|\leq \Lambda + \frac{1}{4}\int_{\ue}^U \Omega^2~\dd u + \int_{\ue}^U 4\pi r^2T_{uv}~\dd u.
\eeqn
Upon integration in $v$, we then have
\beqn\label{c1}
\int_{\ve}^V \sup_{\ue \leq u\leq U} |r\lambda|~\dd v\leq \Lambda (V-V') + \frac{1}{4}W + C_0 =:C_1.
\eeqn
Similarly, we obtain the estimate
\beqn\label{c2}
\int_{\ue}^U \sup_{\ve \leq v\leq V} |r\nu|~\dd u \leq N(U-U') + \frac{1}{4}W + C_0 = : C_2.
\eeqn

A Cauchy-Schwarz inequality and (\ref{b2}) give
\beqna\nonumber
\int_{\ue}^{U}\int_{\ve}^{V} |F_{uv}|~\dd v\dd u &\leq& \left(\int_{\ue}^{U}\int_{\ve}^{V}\frac{1}{4} \Omega^2r^{-2}\Qe^2~\dd v\dd u\right)^{\frac{1}{2}}\left(\int_{\ue}^{U}\int_{\ve}^{V}\Omega^{2}r^{-2}~\dd v\dd u\right)^{\frac{1}{2}}\\
&\leq&\sqrt{C_0}r_0^{-1} \sqrt{W}\label{fuvbound}. 
\eeqna
We immediately then obtain
\beqna\nonumber
\int_{\ue}^{U}\sup_{\ve \leq v \leq V}|A_u|~\dd u &\leq& A(U-U')+ \int_{\ue}^{U}\int_{\ve}^{V} |\pv A_u|~\dd v\dd u\\
&\leq&A(U-U')+ \sqrt{C_0}r_0^{-1} \sqrt{W} =: C_3.\label{c3}
\eeqna
\subsubsection{Uniform bound on $r\phi$}\label{sec:general_emkg/proof_gext/rphi}
Given $\epsilon >0$, partition the region of spacetime $\Dm$ into smaller subregions $\Dm_{jk}$ given by
\beqn\nonumber
\Dm_{jk} = \left([u_j, u_{j+1}]\times[v_k, v_{k+1}]\right)\cap \Dm, \hspace{.3cm} j,k = 0, \ldots, I
\eeqn
with $u_0 = \ue$, $v_0 = \ve$, $u_{I+1}= U$ and $v_{I+1} = V$ such that 
\beqna\label{bnd1}
\int_{v_k}^{v_{k+1}}\int_{u_j}^{u_{j+1}}\Omega^2~\dd u\dd v &< &\epsilon\\
\int_{v_k}^{v_{k+1}}\int_{u_j}^{u_{j+1}}4\pi r^2 T_{uv}~\dd u\dd v &< &\epsilon\\
\int_{v_k}^{v_{k+1}}\sup_{u_j\leq u\leq u_{j+1}}|r\lambda|~\dd v &<& \epsilon\label{bnd4}\\
\int_{u_j}^{u_{j+1}}\sup_{v_k\leq v\leq v_{k+1}}|r\nu|~\dd u &<& \epsilon\label{bnd2}\\
\int_{u_j}^{u_{j+1}} \sup_{v_k\leq v\leq v_{k+1}}|A_u|~\dd u &<& \epsilon\label{bnd3}\\
\int_{v_k}^{v_{k+1}}\int_{u_j}^{u_{j+1}}|F_{uv}|~\dd u\dd v &< &\epsilon\label{bnd5}
\eeqna
for all $j$ and $k$.  This is clearly possible since we have shown that each quantity is \emph{uniformly} bounded in $\Dm$. Note that the cardinality of $I$ for the partition depends on $\epsilon$.

Re-write the wave equation (\ref{KG}) to read
\beqn\label{wavenew}
\pu\pv (r\phi) = \phi\pu\lambda - \e\ii r\Psi(A) - \frac{1}{4}\m^2\Omega^2 r\phi.
\eeqn
By integrating in $u$ and $v$ we now show that the right-hand side of (\ref{wavenew}) can be suitably bounded.

Begin by noting that (\ref{eqn:ruv}) gives
\beqn\label{spacerl}
|\pu \lambda|\leq \frac{1}{4}\Omega^2 r^{-1} + |r\nu||r\lambda| r^{-3} + 4\pi r^2 T_{uv} r^{-1}.
\eeqn
Upon integration, making use of (\ref{bnd1})--(\ref{bnd2}), we have
\beqn\label{pulam}
\int_{v_k}^{\vs}\int_{u_j}^{\us}|\pu \lambda|~\dd u\dd v\leq \frac{1}{4}r_0^{-1} \epsilon + r_0^{-3} \epsilon^2 + r_0^{-1} \epsilon
\eeqn
for all $(\us, \vs)\in \Dm_{jk}$.  
Define
\beqn\nonumber
P_{jk} = \sup_{\Dm_{jk}} |r\phi|.
\eeqn
These constants, with the possible exception of $P_{II}$, are finite. We use (\ref{spacerl}) to estimate
\beqna\nonumber
\left|\int_{v_k}^{\vs}\int_{u_j}^{\us} \phi\pu \lambda~\dd u\dd v\right|&\leq& r_0^{-1}P_{jk}\int_{v_k}^{\vs}\int_{u_j}^{\us}|\pu \lambda|~\dd u\dd v\\
&\leq& r_0^{-1}P_{jk}\left(\frac{1}{4}r_0^{-1} \epsilon + r_0^{-3} \epsilon^2 + r_0^{-1} \epsilon\right).\label{p1}
\eeqna
Using (\ref{bnd4}) and (\ref{bnd3}) we also have
\beqna\nonumber
\left|\int_{v_k}^{\vs}\int_{u_j}^{\us} A_u\phi \lambda~\dd u\dd v\right|&\leq& r_0^{-2}P_{jk}\int_{v_k}^{\vs}\sup_{u_j\leq u\leq \us} |r\lambda|~\dd v\int_{u_j}^{\us} \sup_{v_k\leq v\leq \vs}|A_u|~\dd u\\
&\leq&r_0^{-2}P_{jk} \epsilon^2\nonumber.
\eeqna
Integrating by parts gives
\beqn\nonumber
\int_{v_k}^{\vs} A_u \pv\phi (u,v) ~\dd v = \phi A_u (u, \vs) - \phi A_u(u, v_k) -  \int_{v_k}^{\vs} \phi \pv A_u(u,v)~\dd v.
\eeqn
Thus, upon integration in $u$ we obtain
\beqna\nonumber
\left|\int_{u_j}^{\us} \int_{v_k}^{\vs} A_u r\pv \phi~\dd v\dd u\right|&\leq& 2 r_0^{-1} P_{jk}\int_{u_j}^{\us} \sup_{v_k\leq v\leq \vs}|A_u|~\dd u + r_0^{-1}P_{jk} \int_{u_j}^{\us} \int_{v_k}^{\vs} |F_{uv}|~\dd v\dd u\\
&\leq& 3r_0^{-1}P_{jk}\epsilon,
\eeqna
where we made use of (\ref{bnd3}) and (\ref{bnd5}). Similarly, the bound (\ref{bnd5}) gives
\beqn
\left|\int_{v_k}^{\vs}\int_{u_j}^{\us} \frac{1}{2}\phi F_{uv}~\dd u\dd v\right|\leq \frac{1}{2}r_0^{-1}P_{jk} \epsilon.
\eeqn
Lastly, we note that
\beqn\label{pl}
\left|\int_{v_k}^{\vs}\int_{u_j}^{\us}\frac{1}{4}\m^2 \Omega^2 r\phi~\dd u\dd v \right|\leq \frac{1}{4}\m^2 P_{jk}\epsilon
\eeqn
follows from (\ref{bnd1}). Upon integration of (\ref{wavenew}), it then follows from (\ref{p1})--(\ref{pl}) that for a sufficiently small $\epsilon$, i.e.,~for a sufficiently fine partition of $\Dm$,
\beqn\nonumber
P_{jk}\leq 2 \left(\sup_{[u_j, u_{j+1}]\times\{v_k\}} |r\phi| + \sup_{\{u_j\}\times [v_k, v_{k+1}]}|r\phi|\right) \leq 2\left(P_{j k-1} + P_{j-1 k}\right).
\eeqn
By induction, we have that $P_{jk} \leq C_{jk} <\infty$, where the constant $C_{jk}$ depends only on the initial data.  In particular, there exists a constant $C({\epsilon})<\infty$ such that
\beqn\nonumber
\sup_{[u_I, u_{I+1}]\times\{v_{I}\}} |r\phi| + \sup_{\{u_I\}\times [v_I, v_{I+1}]}|r\phi| \leq C({\epsilon}) \Phi.
\eeqn
By a straightforward continuity argument, one has, for suitably small $\epsilon$, the bound $P_{II} \leq 5 C(\epsilon)\Phi$.  To see this, let $\Dm^*\subset \Dm_{II}$ be the set of $p = (u,v)\in D_{II}$ such that
\beqn\label{bootstrap2}
|r\phi|(\ut, \vt) < 5 C({\epsilon})\Phi
\eeqn
for all $(\ut, \vt)\in (J^-(p)\backslash \{p\})\cap \Dm_{II}$. It is clear that $\Dm^*$ is a connected non-empty open set (in the topology of $\Dm_{II}$).  The claim follows once we show that $\Dm^*$ is also closed (in the topology of $\Dm_{II}$).  That is, for $p\in \overline{\Dm^*}$ (in the topology of $\Dm_{II}$), we must show that $p\in \Dm^*$.  Note by continuity we have $|r\phi| \leq 5 C({\epsilon})\Phi$ in $J^-(p)\cap \Dm_{II}$. Upon integration of (\ref{pulam}), making use of the bootstrap estimate (\ref{bootstrap2}), we have, for some constant $C(r_0) <\infty$,
\beqn\nonumber
\sup_{\Dm_{II}} |r\phi| \leq 3 C({\epsilon})\Phi + C(r_0)\epsilon C(\epsilon)\Phi.
\eeqn
Taking $\epsilon$ such that $C(r_0)\epsilon <1$ gives
\beqn\nonumber
\sup_{\Dm_{II}} |r\phi|\leq 4 C(\epsilon)\Phi,
\eeqn
establishing the result. 

\subsubsection{Uniform bound on $\Omega^2$}\label{sec:general_emkg/proof_gext/omega}

Consider the evolution equation (\ref{logO}).  We recall (\ref{defmagn}) and  (\ref{totalcharge}) that the magnetic charge $\Qm\in \R$.  To uniformly bound $\Omega^2$, it suffices, given $\emph{a priori}$ estimates (\ref{b2}), (\ref{c1}), and (\ref{c2}), to bound
\beqn\label{sym}
\left|\int_{\ue}^{\us}\int_{\ve}^{\vs}\D_u\phi\left(\pv\phi\right)^{\dagger}+\pv\phi\left(\D_u\phi\right)^{\dagger}~\dd v\dd u \right|.
\eeqn
Integrating by parts, we note that
\beqn\nonumber
\int_{\ve}^{\vs}\D_u\phi \left(\pv \phi\right)^{\dagger} (u,v)~\dd v = \phi^{\dagger}\D_u\phi (u,\vs) - \phi^{\dagger}\D_u\phi(u, \ve) - \int_{\ve}^{\vs} \phi^{\dagger}\pv\D_u\phi(u,v)~\dd v.
\eeqn
Similarly, we have
\beqn\nonumber
\int_{\ve}^{\vs}\left(\D_u\phi\right)^{\dagger} \pv\phi(u,v)~\dd v = \phi\left(\D_u\phi\right)^{\dagger} (u,\vs) - \phi\left(\D_u\phi\right)^{\dagger}(u, \ve) - \int_{\ve}^{\vs} \phi\pv\left(\D_u\phi\right)^{\dagger}(u,v)~\dd v.
\eeqn
Since
\beqna
\phi\left(\D_u\phi\right)^{\dagger} &=& \phi \left(\pu\phi\right)^{\dagger} - \e\ii A_u |\phi|^2\nonumber\\
\phi^{\dagger}\D_u\phi &=& \phi^{\dagger} \pu\phi + \e\ii A_u |\phi|^2,\nonumber
\eeqna
it follows that
\beqn\nonumber
\phi^{\dagger}\D_u\phi +  \phi\left(\D_u\phi\right)^{\dagger} = \phi\left(\pu\phi\right)^{\dagger} + \phi^{\dagger}\pu\phi = \pu|\phi|^2.
\eeqn
Thus, the symmetrization in (\ref{sym}) yields the estimate
\beqna\nonumber
\left|\int_{\ue}^{\us}\int_{\ve}^{\vs}\D_u\phi\left(\pv\phi\right)^{\dagger}+\pv\phi\left(\D_u\phi\right)^{\dagger}~\dd v\dd u \right|\leq \hspace{5cm}\\
 2\sup_{V' \leq v\leq \vs}\left|\int_{\ue}^{\us}\pu|\phi|^2 ~\dd u\right|+\left|\int_{\ue}^{\us}\int_{\ve}^{\vs}\phi^{\dagger}\pv\D_u\phi + \phi\pv\left(\D_u\phi\right)^{\dagger}~\dd v\dd u\right| \label{boundcancel}.
\eeqna
Using the bound on $\phi$, we can immediately bound the first term on the right-hand side of (\ref{boundcancel}). Thus, it remains to consider the second term.

  Let us note that 
\beqna\nonumber
\phi^{\dagger}\pv\D_u\phi &=& \phi^{\dagger}\pv\pu\phi + \e\ii \phi^{\dagger}\pv(A_u\phi)\\
\phi\pv\left(\D_u\phi\right)^{\dagger} &=& \phi\pv\left(\pu\phi\right)^{\dagger} - \e\ii \phi\pv(A_u\phi^{\dagger}).\nonumber
\eeqna
Using the wave equation (\ref{KG}), we obtain
\beqna
\phi^{\dagger}\pu\pv \phi =  -r^{-1}\left(\nu \phi^{\dagger}\pv\phi + \lambda \phi^{\dagger}\pu\phi\right)\hspace{5.4cm}\nonumber\\
- \e\ii\left(A_u |\phi|^2 r^{-1}\pv r + A_u \phi^{\dagger}\pv\phi - \frac{1}{2}F_{uv}|\phi|^2\right)- \frac{1}{4}\m^2\Omega |\phi|^2\nonumber
\eeqna
and similarly, taking the complex conjugate of (\ref{KG}), we have
\beqna
\phi\pu\left(\pv \phi\right)^{\dagger} = -r^{-1}\left(\nu \phi\left(\pv\phi\right)^{\dagger} + \lambda \phi\left(\pu\phi\right)^{\dagger}\right) \hspace{5.5cm} \nonumber\\
+ \e\ii\left(A_u |\phi|^2 r^{-1}\pv r + A_u \phi\left(\pv\phi\right)^{\dagger} - \frac{1}{2}F_{uv}|\phi|^2\right)- \frac{1}{4}\m^2\Omega^2 |\phi|^2.\nonumber
\eeqna
In particular,
\beqna\label{cancel1}
\phi^{\dagger}\pu\pv \phi + \phi\pu\left(\pv \phi\right)^{\dagger} = \hspace{8cm}\nonumber\\
-r^{-1}\left(\nu \pv|\phi|^2 + \lambda \pu|\phi|^2\right) +\e\ii A_u\left(\phi\left(\pv\phi\right)^{\dagger} - \phi^{\dagger}\pv\phi\right)- \frac{1}{2}\m^2\Omega^2 |\phi|^2.
\eeqna
Lastly, we note that
\beqn\label{cancel2}
\phi^{\dagger}\pv(A_u\phi) -\phi \pv (A_u\phi^{\dagger})=   A_u \left(\phi^{\dagger} \pv \phi - \phi \left(\pv\phi\right)^{\dagger}\right).
\eeqn
It follows from (\ref{cancel1}) and (\ref{cancel2}) that
\beqn\nonumber
\phi^{\dagger}\pv\D_u\phi + \phi\pv\left(\D_u\phi\right)^{\dagger} = -r^{-1}\left(\nu \pv|\phi|^2 +\lambda \pu|\phi|^2\right) - \frac{1}{2}\m^2\Omega^2 |\phi|^2.
\eeqn
Using the uniform bound on $\phi$ and the finiteness of the spacetime volume, we can uniformly control
\beqn\nonumber
\left|\int_{U'}^{\us}\int_{V'}^{\vs}\frac{1}{2}\m^2\Omega^2 |\phi|^2~\dd v\dd u\right|.
\eeqn
We now bound
\beqn\nonumber
\left|\int_{U'}^{\us}\int_{V'}^{\vs}r^{-1}\left(\nu \pv|\phi|^2 + \lambda \pu|\phi|^2\right)~\dd v \dd u\right|.
\eeqn

Let us integrate by parts to first note that
\beqn\nonumber
\int_{V'}^{\vs}r^{-1}\nu \pv|\phi|^2(u,v)~\dd v = r^{-1}\nu |\phi|^2 (u, \vs) - r^{-1}\nu |\phi|^2(u,V') - \int_{V'}^{\vs} \pv (r^{-1}\nu) |\phi|^2(u,v)~\dd v.
\eeqn
Making use of the uniform bound on $\phi$ and the \emph{a priori} estimate (\ref{c2}), we now integrate in $u$ and obtain
\beqn\nonumber
\left|\int_{U'}^{\us}\int_{V'}^{\vs}r^{-1}\nu \pv|\phi|^2~\dd v\right|\leq C + C\int_{U'}^{\us}\int_{V'}^{\vs}\left|\pv (r^{-1}\nu)\right|~\dd v\dd u.
\eeqn
From (\ref{eqn:ruv}), since
\beqn\nonumber
|\pv \nu| \leq  \frac{1}{4}\Omega^2r^{-1} +|r\lambda||r\nu| r^{-3} + 4\pi r^2 T_{uv} r^{-1},
\eeqn
we note that
\beqn\nonumber
|\pv (r^{-1} \nu)| \leq r^{-1}|\pv\nu| + r^{-4} |r\nu||r\lambda| \leq \frac{1}{4}\Omega^2r^{-2} +2|r\lambda||r\nu| r^{-4} + 4\pi r^2 T_{uv} r^{-2}.
\eeqn
Using the \emph{a priori} estimates, we now have
\beqn\nonumber
\int_{U'}^{\us}\int_{V'}^{\vs}\left|\pv (r^{-1}\nu)\right|~\dd v\dd u \leq C,
\eeqn
allowing us to obtain
\beqn\nonumber
\left|\int_{U'}^{\us}\int_{V'}^{\vs}r^{-1}\nu \pv|\phi|^2~\dd v\dd u\right|\leq C.
\eeqn
Similarly, one can show that
\beqn\nonumber
\left|\int_{U'}^{\us}\int_{V'}^{\vs}r^{-1}\lambda \pu|\phi|^2~\dd v\dd u\right|\leq C.
\eeqn
We now immediately retrieve the bound we seek:
\beqn\nonumber
\left|\int_{\ue}^{\us}\int_{\ve}^{\vs}\phi^{\dagger}\pv\D_u\phi + \phi\pv\left(\D_u\phi\right)^{\dagger}~\dd v\dd u\right|\leq C.
\eeqn
In particular, integrating (\ref{logO}) we obtain
\beqn\nonumber
\sup_{\Dm}|\log\Omega^2|\leq C<\infty.
\eeqn
Thus, there are constants $c_0$ and $c_1$ such that
\beqn\label{pointwiseomega}
0 < c_0 \leq \Omega^2 \leq c_1  < \infty
\eeqn
uniformly in $\Dm$.

\subsubsection{One-sided estimates for $r\lambda$ and $r\nu$}
Note that (\ref{eqn:ruv}) gives
\beqn\nonumber
\pu (r\lambda) = \pv(r\nu)\geq -\frac{1}{4}\Omega^2.
\eeqn
Using the uniform estimate (\ref{pointwiseomega}) for $\Omega$, integration then yields
\beqna
-r\lambda (\us,\vs) &\leq& - r\lambda (U', \vs) + \frac{1}{4}\int_{U'}^{\us} \Omega^2~\dd u\nonumber\\
&\leq& \Lambda + \frac{1}{4}c_1 (U - U')=: \Lambda' < \infty\label{lambda0}.
\eeqna
Similarly, we note that
\beqn\nonumber
-r\nu (\us, \vs) \leq N + \frac{1}{4}c_1 (V-V') =: N' < \infty.
\eeqn
\subsubsection{$L^2$-estimates for $\pv\phi$ and ${\rm D}_u\phi$}
Integrating (\ref{eqn:convv}) we obtain
\beqna
\Omega^{-2}\lambda(\us,V') - \Omega^{-2}\lambda(\us, \vs) &=& \int_{V'}^{\vs}4\pi r\Omega^{-2} |\pv\phi|^2~\dd v\nonumber\\
&\geq& 4\pi r_0 c_1^{-1} \int_{V'}^{\vs}|\pv\phi|^2~\dd v.\nonumber
\eeqna
Thus, we have, using the one-sided estimate (\ref{lambda0}),
\beqna\nonumber
\int_{V'}^{\vs}|\pv\phi|^2~\dd v&\leq& \frac{c_1}{4\pi r_0}\left(\Omega^{-2}\lambda(\us,V') - \Omega^{-2}\lambda(\us, \vs)\right)\\
&\leq& \frac{c_1}{4\pi r_0^2 c_0}\left(\Lambda + \Lambda'\right)\label{raych}.
\eeqna
Similarly, we have
\beqn\label{raych2}
\int_{U'}^{\us}|\D_u\phi|^2~\dd u\leq \frac{c_1}{4\pi r_0^2 c_0}\left(N +N'\right).
\eeqn
\subsubsection{Uniform bound on $T_{uv}$}\label{sec:tuv}

With bounds on $\phi$ and $\Omega$, it remains to bound $\Qe$ in order to give, from (\ref{tuv}), an estimate for $T_{uv}$.

Using Cauchy-Schwarz, the bound on $\phi$, and (\ref{raych}), we note the bound
\beqn\nonumber
\left|\int_{\ve}^{\vs}r\left(\phi\left(\pv\phi\right)^{\dagger}-\phi^{\dagger}\pv\phi\right)
~\dd v\right|\leq 2 \left(\int_{V'}^{\vs}|\pv\phi|^2~\dd v\right)^{\frac{1}{2}}\left(\int_{V'}^{\vs}|r\phi|^2~\dd v\right)^{\frac{1}{2}} \leq C.
\eeqn
Upon integration of (\ref{beta}), we, therefore, obtain
\beqn\label{betabound}
|\Qe(\us, \vs)| \leq|\Qe(\us, \ve)| + \left|\int_{\ve}^{\vs} 2\pi\e \ii  r^2\left(\phi\left(\pv\phi\right)^{\dagger}-\phi^{\dagger}\pv\phi\right)
~\dd v\right| \leq B + 2\pi |\e| R C.
\eeqn
It thus follows that 
\beqn\nonumber
\sup_{\Dm} T_{uv} \leq C.
\eeqn
\subsubsection{Remaining uniform estimates}\label{sec:general_emkg/proof_gext/remain}

We can now bound all other quantities quite easily.

To estimate the first derivatives of $\phi$ let us use the wave equation (\ref{KG}) to compute
\beqna
\pu \theta_A + \e\ii A_u\theta_A&=&- \frac{\zeta_A}{r}\lambda  -\frac{1}{4}\Omega^2r\phi\left(\m^2 -\e\ii \frac{\Qe}{r^2}\right)\label{put}\\
\pv \zeta_A &=& - \frac{\theta_A}{r}\nu -\frac{1}{4}\Omega^2r\phi\left(\m^2  + \e\ii \frac{\Qe}{r^2}\right).\label{pvz}
\eeqna
From the uniform bounds on $T_{uv}$ and $\Omega^2$ we note that (\ref{boundla}) gives
\beqn\label{lambda1}
\sup_{\Dm}|r\lambda| \leq C.
\eeqn
Similarly, we have
\beqn\nonumber
\sup_{\Dm}|r\nu| \leq C.
\eeqn
Since (\ref{raych2}) and (\ref{lambda1}) give
\beqn\nonumber
\left|\int_{\ue}^{\us} -\frac{\zeta_A}{r}\lambda~\dd u\right|\leq C\int_{\ue}^{\us} \left|\D_u\phi\right|~\dd u\leq C \left(\int_{U'}^{\us}|\D_u\phi|^2~\dd u\right)^{\frac{1}{2}}\left(\int_{U'}^{\us}1~\dd u\right)^{\frac{1}{2}}\leq C,
\eeqn 
integrating (\ref{put}) we have
\beqn\nonumber
\sup_{\Dm} |\theta_A|\leq C
\eeqn
using the estimates on $\phi$, $\Omega^2$, and $\Qe$.  Using this bound, we integrate (\ref{pvz}) and similarly obtain
\beqn\nonumber
\sup_{\Dm} |\zeta_A|\leq C.
\eeqn
Integration of (\ref{fuv}) also gives, using (\ref{betabound}), the uniform estimate
\beqn\nonumber
\sup_{\Dm}|A_u|\leq C.
\eeqn
It remains now to show that $\pu \Omega$, $\pv\Omega$, $\pu \nu$, and $\pv \lambda$ are bounded quantities in $\Dm$.  By integrating (\ref{logO}), we obtain uniform estimates on $\pu\Omega$ and $\pv \Omega$ using already derived bounds.  Then, we can give uniform estimates for $\pu \nu$ and $\pv \lambda$ by using (\ref{eqn:conuu}) and (\ref{eqn:convv}), respectively.

  We have established that $N(\Dm) <\infty$ and Theorem \ref{thm:gext} follows from applying Proposition \ref{naext}. \qed 

\section{Proof of Theorem \ref{thm:main}: global characterization of spacetime}\label{sec:proof_main}

With the extension principle of Theorem \ref{thm:gext} having been proven, we are now able to give a characterization of first singularities that arise in the collapse of self-gravitating charged scalar fields. Together with monotonicity arguments, a consequence of the dominant energy condition (although, in some cases, monotonicity arising from the weaker null energy condition will suffice), we can then give a proof of Theorem \ref{thm:main}, establishing the global characterization of spacetime.

\subsection{Characterization of first singularities}\label{sec:proof_main/first}
We apply Theorem \ref{thm:gext} under the assumptions of initial data in Theorem \ref{thm:main} to give a characterization of first singularities $\mathcal{B}_1^+$.  Introducing the notation
\beqn\nonumber
r_{\inf}(p) = \lim_{q\rightarrow p}\inf_{J^-(p) \cap J^+(q)\cap \Q^+} r(q)
\eeqn
for $p\in \overline{\Q^+}$ and $q\in J^-(p)\cap \Q^+$, we have, as an immediate corollary of Theorem \ref{thm:gext}, recalling Proposition \ref{prop:finite_volume},
\begin{cor}\label{corext}
If $p\in \mathcal{B}_1^+$, then $r_{\inf}(p) = 0$.
\end{cor} 
As a consequence of Corollary \ref{corext}, if $p\in \overline{\mathcal{R}}\backslash{\overline{\Gamma}}\subset \overline{\Q^+}$ and $q\in \left( I^-(p)\cap\overline{\mathcal{R}}\right)\backslash \{p\}$ are such that
\beqn\nonumber
\Dm = \left(J^-(p)\cap J^+(q)\right)\backslash \{p\}\subset \mathcal{R}\cup \mathcal{A},
\eeqn
then $p\not\in \mathcal{B}_1^+$. Indeed, by taking $q$ close enough to $p$ such that $\Dm \cap \Gamma = \emptyset$, then $r\geq r_0 >0$ on $\Dm$ since $\lambda \geq 0$ and $\nu <0$ in $\mathcal{R}\cup \mathcal{A}$. Whence, we state
\begin{cor}If $p\in \mathcal{B}_1^+\cap \overline{\mathcal{R}}$, then $p=b_{\Gamma}$.
\end{cor}
In particular, for the Einstein-Maxwell-Klein-Gordon system the generalized extension principle implies the weak extension principle (cf.~\S\ref{sec:intro/general/weak}).  More generally, since Proposition \ref{prop:finite_volume} holds for any matter model that obeys the null energy condition, this establishes Proposition \ref{imply} of \S \ref{sec:intro/general/tame}, i.e.,~every strongly tame Einstein-matter system is also weakly tame.  

\subsection{The Hawking mass and its monotonicity properties}
It will be convenient here to recall that the Hawking mass $m$ is given by  
\beqn\label{hm}
1-\frac{2m}{r} = g^{ab}\partial_ar \partial_b r = -4\Omega^{-2}\lambda \nu,
\eeqn
and its evolution equations can be shown to satisfy (cf.~\cite{DC95})
\beqna
\pu m &=& 8\pi r^2 \Omega^{-2}\left(T_{uv}\nu - T_{uu}\lambda\right)\nonumber\\
\pv m &=& 8\pi r^2 \Omega^{-2}\left(T_{uv}\lambda - T_{vv}\nu\right).\nonumber
\eeqna
Because we assume that $\m^2\geq 0$, the Einstein-Maxwell-Klein-Gordon system satisfies the dominant energy condition (cf.~\S\ref{energymomentum}):
\beqn\nonumber
T_{uv}\geq 0, \hspace{.5cm}T_{uu}\geq 0,\hspace{.5cm} \textrm{and}\hspace{.5cm}T_{vv}\geq 0.
\eeqn
As a consequence of Proposition \ref{prop: antitrapped}, the Hawking mass, therefore, satisfies monotonicity properties
\beqn \label{hawkingmassv}
\pv m \geq 0 \hspace{.5cm}\textrm{and}\hspace{.5cm} \pu m \leq 0
\eeqn
 in $\mathcal{R}\cup \mathcal{A}$.  This monotonicity is only used in a handful of the proofs below. Where it is not exploited (keeping in mind the need for the dominant energy condition\footnote{see, however, footnote \ref{foot:dominant}} in Corollary \ref{corext}), the null energy condition
\beqn\nonumber
T_{uu}\geq 0\hspace{.5cm} \textrm{and}\hspace{.5cm}T_{vv}\geq 0
\eeqn
will suffice.  

\subsection{Statements I and II}\label{sec:proof_I_II}
Recall the notation of \S \ref{sec:rud}.  We will systematically refine the boundary characterization of Proposition \ref{rudb} so as to obtain finally that which is given in Theorem \ref{thm:main}.

\subsubsection{Null segment $\mathcal{N}_{\Gamma}$ emanating from $b_{\Gamma}$}\label{ngamma}
Let 
\beqn\nonumber
\mathcal{CH}_{\Gamma} =\{q \in \mathcal{N}_{\Gamma} : \exists p, p'\in \mathcal{N}_{\Gamma}~~\textrm{s.t.}~~ q\in (p, p']~~\textrm{and}~~r_{\inf}(q') \neq 0 ~\forall q'\in (p, p')\},
\eeqn
and define
\beqn\nonumber
\mathcal{S}_{\Gamma} = \mathcal{N}_{\Gamma}\backslash \chg.
\eeqn
We note that $r_{\inf}$ need not \emph{a priori} vanish on $\mathcal{S}_{\Gamma}$.  By definition, however, for every $(U,v)\in \mathcal{S}_{\Gamma}$ and every $\delta >0$, there exists $v-\delta < \vt\leq v$ such that $(U, \vt)\in \mathcal{S}_{\Gamma}$ and  $r_{\textrm{inf}}(U,\vt)=0$. 

We show $\chg$ is a connected set. Suppose, on the contrary, that $\chg$ is disconnected.  Then, there exists
$(U,v)\in \chg$ and $(U,v')\in\chg$, where without loss of generality $v < v'$, and $v< v''<v'$ such that $(U, v'')\in \mathcal{S}_{\Gamma}$. Moreover, there exists $v < v''' \leq v''$ such that $r_{\inf}(U,v''') = 0$.  Thus, there exists a sequence $(U_j, v'''_j)\rightarrow (U, v''')$ such that $r(U_j, v'''_j)\rightarrow 0$ with $U_j\leq U$ and $v'''_j\leq v'''$.
Define
\beqn\nonumber
r_0 = \min\{r_{\inf}(U,v), r_{\inf}(U,v')\} >0.
\eeqn
Choose $J$ sufficiently large so that $r(U_j, v'''_j) < r_0$ for all $j\geq J$.
Since $\nu<0$ in $\Q^+$, it follows that
\beqn\nonumber
r(u,v) > r_0 \hspace{1cm} \textrm{and}\hspace{1cm}r(u,v') > r_0
\eeqn
for all $u\geq U_J$. It follows that for each $U_j\in [U_J, U)$ there exists $\vt(U_j)\in (v, v'''_j)$ and $\hat{v}(U_j)\in (v'''_j, v')$ such that
\beqn\nonumber
\lambda(U_j, \vt(U_j))<0\hspace{1cm}\textrm{and}\hspace{1cm} \lambda(U_j,\hat{v}(U_j)) >0.
\eeqn
This, however, contradicts monotonicity (\ref{eqn:convv}) and no such $\hat{v}(U_j)$ can exist.  We conclude that $\chg$ is a single (possibly empty) half-open interval. From monotonicity (\ref{eqn:convv}), it then follows that $\mathcal{N}_{\Gamma}$ is given by
\beqn\nonumber
\mathcal{N}_{\Gamma} = \sgo\cup\chg\cup \sgt,
\eeqn
where $\sgo$ is a half-open (possibly empty) connected component of $\mathcal{S}_{\Gamma}$ that emanates from (but does not include) $b_{\Gamma}$ and $\sgt$ is a half-open (possibly empty) connected component of $\mathcal{S}_{\Gamma}$ that emanates from (but does not include) the future endpoint of $\chg$. We note that if $\mathcal{CH}_{\Gamma}=\emptyset$, then, necessarily, $\sgt = \emptyset$.
\subsubsection{Null segment $\mathcal{N}_{i^+}$ emanating from $i^{\square}$}
Let
\beqn\nonumber
\mathcal{CH}_{i^+} =\{q \in \mathcal{N}_{i^+} : \exists p, p'\in \mathcal{N}_{i^+}~~\textrm{s.t.}~~ q\in (p, p']~~\textrm{and}~~r_{\inf}(q') \neq 0 ~\forall q'\in (p, p')\}
\eeqn
and define
\beqn\nonumber
\mathcal{S}_{i^+} = \mathcal{N}_{i^+}\backslash \mathcal{CH}_{i^+}.
\eeqn
Again we note that $r_{\inf}$ need not \emph{a priori} vanish on $\mathcal{S}_{i^+}$. By definition, however, for every $(u,V)\in \mathcal{S}_{i^+}$ and every $\delta >0$ there exists $u-\delta < \ut\leq u$ such that $(\ut, V)\in \mathcal{S}_{i^+}$ and  $r_{\textrm{inf}}(\ut,V)=0$. 

We will establish that $\mathcal{S}_{i^+}$ is a connected set by showing that if $(u, V)\in \mathcal{S}_{i^+}$ and $r_{\textrm{inf}}(u,V)=0$, then $r_{\inf}(u',V)=0$ for all $u'\geq u$.  It will then follow that $\mathcal{CH}_{i^+}$ is a half-open (possibly empty) interval that, necessarily, emanates from (but does not include) $i^{\square}$. $\mathcal{S}_{i^+}$ is then a half-open (possibly empty) interval that emanates from (but does not include) the future endpoint of $\mathcal{CH}_{i^+}\cup i^{\square}$.

Suppose $(u,V)\in \mathcal{S}_{i^+}$.  Then, by definition, there exists some $u' \leq u$ such that $(u', V)\in \mathcal{S}_{i^+}$ and a sequence $(u'_j, V_j)\rightarrow (u',V)$ such that $r(u'_j, V_j)\rightarrow 0$ with $u_j\leq u'$ and $V_j< V$.  Since $\nu<0$ in $\Q^+$, $r(u'', V_j)\leq r(u'_j, V_j)$ for all $\{(u'', V_j): u''\geq u'\}\cap \Q^+$.  Thus, $(u'', V)\in \mathcal{S}_{i^+}$ and, in particular, \emph{a posteriori} $r_{\inf}$ vanishes on $\mathcal{S}_{i^+}$.

\subsubsection{The remaining achronal boundary} 

Let us define
\beqn\nonumber
\mathcal{S} =  \bigcup_{p\in \mathcal{B}_1^+\backslash b_{\Gamma}} \{p\}\cup \mathcal{N}_p^1 \cup \mathcal{N}_p^2.
\nonumber
\eeqn
We will show that $r_{\inf}=0$ on $\mathcal{S}$.  

Recall that Corollary \ref{corext} gives that $r_{\inf}(u,v)=0$ for all $(u,v) \in \mathcal{B}_1^+\backslash b_{\Gamma}$. In particular, there exists a sequence $(u_j, v_j)\rightarrow (u,v)$ such that $r(u_j, v_j) \rightarrow 0$ with $u_j\leq u$ and $v_j\leq v$.  Let $\mathcal{N}^1_p\subset \{v = v(p)\}$.  Since $\nu<0$ in $\Q^+$, $r(u', v_j)\leq r(u_j, v_j)$ for all $u'\geq u$. Thus, $r_{\inf}(u', v) = 0$. Let $\mathcal{N}^2_p\subset \{u = u(p)\}$.  Since $(u,v)$ is \emph{non-central}, there exists, by compactness, a $v'< v$ such that $(u,v')\in \Q^+$ and an $\epsilon >0$ such that $r(u', v')> \epsilon$ for all $\{(u', v'): u'\leq u\}\cap \Q^+$.  Moreover, there exists a sufficiently large $J$ such that $r(u_j, v_j) < \epsilon$ for all $j\geq J$.  It follows that there exists a $\hat{v}(j)\in (v', v_j)$ such that $\lambda(u_j, \hat{v}(j))<0$. Appealing to monotonicity (\ref{eqn:convv}), we conclude that $\lambda(u_j, v'') <0$ for all $v'' \geq \hat{v}(j)$. Whence $r(u_j, v'') \leq r(u_j, v_j)$ and $r_{\inf}(u, v'')=0$.
\subsubsection{$\mathcal{I}^+$ is an open set in the topology of $\mathcal{B}^+$}
If $b_{\Gamma} = i^+$, then the openness of $\mathcal{I}^+$ is obvious by definition.  If $b_{\Gamma}\neq i^{\square}$, it follows from Statement VI, which is proven below in \S \ref{statementvi}, that $i^{\square}\not\in \mathcal{I}^+$. This will establish the claim.

\subsubsection{Continuous extendibility of $r$}\label{r_ext}

Of the remaining properties of the boundary to establish is that of the extendibility properties of $r$. We do this in the sequel in a case-by-case analysis.

\subsubsection*{$r$ extends continuously to zero on $\left(\mathcal{S}\cup \mathcal{S}_{i^+}\right)\backslash b_{\Gamma}$}
For every $(u,v)\in \left(\mathcal{S}\cup \mathcal{S}_{i^+}\right)\backslash b_{\Gamma}$, there exists a neighborhood $\mathcal{U}\subset\R^{1+1}$ of $(u,v)$ such that $\mathcal{U}\cap \Q^+\subset \mathcal{T}$. Thus, there exists a sequence $(u_j, v_j)\rightarrow (u,v)$ such that $r(u_j, v_j)\rightarrow 0$ with $u_j\leq u$ and $v_j\leq v$.  Perturbing the sequence, we may, without loss of generality, assume that $v_j < v$. Given $\epsilon >0$, there exists a sufficiently large $J$ such that $r(u_j, v_j)< \epsilon$ for all $j\geq J$.  Since $\nu<0$ and $\lambda <0$ in $\mathcal{U}\cap \Q^+$, it follows that $r(u', v') <\epsilon$ in $\{u'\geq u_J\}\cap \{v'\geq v_J\}\cap \mathcal{U}\cap \Q^+$. 

\subsubsection*{$r$ extends continuously to zero on $b_{\Gamma}\in \mathcal{S}_{i^+}$}

Let $(U,V) = b_{\Gamma}\in \mathcal{S}_{i^+}$ and consider a neighborhood $\mathcal{U}\subset \R^{1+1}$ of $b_{\Gamma}$.

Fix $u' < U$ such that $(u', V)\in \mathcal{U}\cap\mathcal{S}_{i^+}$.  Since $r$ extends continuously to zero on $(u', V)$, there exists a sequence $(u'_j, V_j)$ such that $r(u'_j, V_j)\rightarrow 0$ with $u'_j\leq u'$ and $V_j< V$.  Given $\epsilon >0$, there exists a sufficiently large $J$ such that $r(u'_j, V_j) < \epsilon$ for all $j\geq J$. Since there exists a sufficiently small neighborhood $\widetilde{\mathcal{U}}$ of $(u',V)$ such that $\widetilde{\mathcal{U}}\cap\Q^+\subset\mathcal{T}$, we can assume, without loss of generality, that $(u'_J, V_J)\in \mathcal{T}$.  By monotonicity (\ref{eqn:convv}), it follows that $(u'_J, v') \in \mathcal{T}$ for all $\{v'\geq V_J \}\cap \Q^+$.  Since $\nu<0$ in $\Q^+$, we have that
\beqn\nonumber
r(u'',v)\leq \epsilon
\eeqn
for all $(u'',v)\in \{u''\geq u'_{J}\}\cap \{v\geq V_{J}\}\cap \mathcal{U}\cap \Q^+$.

\subsubsection*{$r$ extends continuously to zero on $b_{\Gamma}\backslash \left(\mathcal{S}_{i^+}\cup\mathcal{CH}_{i^+}\cup i^{\square}\right)$}

Let $(U,V) = b_{\Gamma}\backslash \left(\mathcal{S}_{i^+}\cup\mathcal{CH}_{i^+}\cup i^{\square}\right)$ and consider a neighborhood $\mathcal{U}\subset \R^{1+1}$ of $b_{\Gamma}$. Since, by assumption, $b_{\Gamma}$ does not coincide with the future endpoint of $\mathcal{S}_{i^+}\cup\mathcal{CH}_{i^+}\cup i^{\square}$, we note that there exists a sufficiently small $\delta >0$ such that $\{|v-V|\leq \delta\}\cap \Sigma$ is compact.  Hence, there exist $u_0 <U$ and a positive constant $C(u_0)<\infty$ such that
\beqn\label{monokappa}
-\frac{1}{4}\left(\Omega^{-2}\nu\right)^{-1}(u,v) = \frac{\lambda}{1-\frac{2m}{r}}(u, v)\leq \frac{\lambda}{1-\frac{2m}{r}}(u_0,v)\leq C(u_0)
\eeqn
in $\{ |v-V|\leq \delta\}\cap \{u\geq u_0\}\cap \mathcal{U}\cap \Q^+$ by (\ref{hm}) and monotonicity (\ref{eqn:conuu}). Since (\ref{hm}) gives $1 -\frac{2m}{r}\leq 1$ in $\mathcal{R}\cup\mathcal{A}$ (cf.~Statement V), we obtain
\beqn\label{lambdabound2}
0\leq \lambda(u, v)\leq \frac{\lambda}{1-\frac{2m}{r}}(u_0,v)\left(1-\frac{2m}{r}\right)(u,v)\leq C(u_0),
\eeqn
for all $(u,v)\in \{|v-V|\leq \delta\}\cap \{u\geq u_0\}\cap \mathcal{U}\cap \left(\mathcal{R}\cup\mathcal{A}\right)$.

Given $\epsilon >0$, let us choose $\delta' \leq\min \{\delta, (2C(u_0))^{-1} \epsilon\}$. Along each outgoing null segment in $\mathcal{U}\cap \Q^+$ emanating from $\Gamma = \{(u,v_{\Gamma}(u))\}$ let us define
\[ v_{\mathcal{A}}(u) = \left\{ \begin{array}{ll}
\inf\{v: \lambda (u,v) = 0\}, & \mbox{if $\mathcal{A}\cap \{ v \leq V+\delta'\}\neq\emptyset$;} \\
V + \delta', &\mbox{if $\mathcal{A}\cap \{ v \leq V+\delta'\}=\emptyset$} . \end{array} \right. \]
We note that $v_{\Gamma}(u) < v_{\mathcal{A}}(u)$ since $\Gamma \subset \mathcal{R}$, which is proven below in \S\ref{statement:iv}, and thus  $\{u\}\times[v_{\Gamma}(u), v_{\mathcal{A}}(u))\subset \mathcal{R}$.
Let $(u', v_{\Gamma}(u'))$ be the intersection point of $\{v = V- \delta'\}$ with $\Gamma$.
Then,
\beqn\nonumber
r(u, v ) - r(u,v_{\Gamma}(u)) = \int^{v}_{v_{\Gamma}(u)} \lambda(u, \overline{v})~\dd \overline{v}\leq C(u_0) (v- v_{\Gamma}(u))\leq 2C(u_0)\delta'\leq \epsilon
\eeqn
in $\{v\leq  v_{\mathcal{A}}(u)\}\cap \{u\geq u'\}\cap \mathcal{U}\cap\left(\mathcal{R}\cup \mathcal{A}\right)$.
Since $r$ is decreasing in $ \mathcal{T}$, there exists a neighborhood $\mathcal{U}'\subset \R^{1+1}$ satisfying $\mathcal{U}'\subset \{|v-V|\leq \delta'\}\cap\{u\geq u'\} \cap\mathcal{U}
$
so that
\beqn\nonumber
r(u,v)\leq \epsilon
\eeqn
for all $(u,v)\in \mathcal{U}'\cap \Q^+$.
\subsubsection*{$r$ extends continuously to zero on $\sgo\backslash \left(\mathcal{CH}_{i^+}\cup i^{\square}\right)$}
Let $(U,v)\in \sgo\backslash \left(\mathcal{CH}_{i^+}\cup i^{\square}\right)$ and consider a neighborhood $\mathcal{U}\subset \R^{1+1}$ of $(U,v)$. If $(U,v)\in \mathcal{S}_{i^+}$, then we establish continuity as in the case in which $b_{\Gamma}\in \mathcal{S}_{i^+}$.  Thus, without loss of generality, let us assume that $(U,v)\not\in \mathcal{S}_{i^+}$.  Since $\sgo$ does not coincide with the future endpoint of $\mathcal{S}_{i^+}\cup\mathcal{CH}_{i^+}\cup i^{\square}$, we note that there exists a sufficiently small $\delta >0$ such that $\{|v'-v|\leq \delta\}\cap \Sigma$ is compact.  Hence, there exist $u_0 <U$ and a positive constant $C(u_0)<\infty$ such that
\beqn\nonumber
-\frac{1}{4}\left(\Omega^{-2}\nu\right)^{-1}(u,v') = \frac{\lambda}{1-\frac{2m}{r}}(u, v')\leq \frac{\lambda}{1-\frac{2m}{r}}(u_0,v')\leq C(u_0)
\eeqn
in $\{ |v'-v|\leq \delta\}\cap \{u\geq u_0\}\cap \mathcal{U}\cap \Q^+$  by (\ref{hm}) and monotonicity (\ref{eqn:conuu}).  Since (\ref{hm}) gives $1 -\frac{2m}{r}\leq 1$ in $\mathcal{R}\cup\mathcal{A}$ (cf.~Statement V), we therefore obtain

\beqn\nonumber
0\leq \lambda(u, v')\leq \frac{\lambda}{1-\frac{2m}{r}}(u_0,v')\left(1-\frac{2m}{r}\right)(u,v')\leq C(u_0),
\eeqn
for all $(u,v')\in \{|v'-v|\leq \delta\}\cap \{u\geq u_0\}\cap \mathcal{U}\cap \left(\mathcal{R}\cup\mathcal{A}\right)$.

Let $\epsilon >0$ be given and choose $0<\delta' \leq \delta$ such that
\beqn\nonumber
C(u_0)\delta' < \frac{1}{4}\epsilon.
\eeqn
By definition, there exists some $v- \delta' < v'' \leq v$ such that $(U,v'')\in \sgo$ and a sequence $(U_j, v''_j)\rightarrow  (U, v'')$ such that $r(U_j, v''_j)\rightarrow 0$ with $U_j\leq U$ and $v''_j\leq v''$. In particular, there exists a sufficiently large $J$ such that $r(U_j, v''_j) < \frac{1}{2}\epsilon$ for all $j\geq J$.
We then have
\beqn\nonumber
r(U_j,v''')= r(U_j, v''_j) + \int_{v''_j}^{v'''} \lambda (U_j, \overline{v})~\dd \overline{v} \leq \frac{1}{2}\epsilon + 2C \delta' \leq \epsilon
\eeqn
for all $(U_j, v''')\in \{v''_j\leq v''' \leq v + \delta'\}\cap \{U_j\geq u_0\}\cap\mathcal{U} \cap \left(\mathcal{R}\cup\mathcal{A}\right)$.  Since $\nu < 0$ in $\Q^+$, it then follows that 
\beqn\nonumber
r(u,v''')<\epsilon
\eeqn
in $\{v'' \leq v'''\leq v+\delta'\}\cap \{u\geq U_J\}\cap \mathcal{U}\cap \left(\mathcal{R}\cup\mathcal{A}\right)$.
Since $r$ is decreasing in $\mathcal{T}$, there is a neighborhood $\mathcal{U}'\subset\R^{1+1}$ satisfying $\mathcal{U}'\subset \{v'' \leq v'\leq v+\delta'\}\cap \{u\geq U_J\}\cap \mathcal{U}$ such that
\beqn\nonumber
r(u,v')\leq \epsilon
\eeqn
for all $(u,v')\in \mathcal{U}'\cap \Q^+$.
\subsubsection*{$r$ extends continuously to zero on $\sgt$}
For every $(U,v)\in \sgt$, there exists a neighborhood $\mathcal{U}\subset \R^{1+1}$ such that $\mathcal{U}\cap \Q^+\subset\mathcal{T}$.  This follows by construction and monotonicity (\ref{eqn:convv}), for if $\sgt\neq \emptyset$, then $\mathcal{CH}_{\Gamma}\neq \emptyset$.  By definition, there exists a $v' \leq v$ such that $(U,v')\in \sgt\cap\mathcal{U}$ and a sequence $(U_j, v'_j)\rightarrow (U, v')$ such that $r(U_j, v'_j)\rightarrow 0$ with $U_j\leq U$ and $v'_j \leq v'$.  Given $\epsilon >0$, there exists a sufficiently large $J$ such that $r(U_j, v'_j)< \epsilon$ for all $j\geq J$.  Since $\lambda <0$ and $\nu<0$ in $\mathcal{U}\cap \Q^+$, we therefore have
\beqn\nonumber
r(u, v'') <\epsilon
\eeqn
in $\mathcal{U}'\cap \Q^+$ for any neighborhood $\mathcal{U}'$ satisfying $\mathcal{U}'\subset\{v'' \geq v' \}\cap \{u\geq U_J\}\cap \mathcal{U}$.
  
\subsubsection*{$r$ extends continuously on ${\rm int}\left(\mathcal{CH}_{\Gamma}\right)$}\label{b0cont}
Let $(U,v)\in\textrm{int}(\mathcal{CH}_{\Gamma})$.  Suppose $\mathcal{U}\subset\R^{1+1}$ is a neighborhood of $(U,v)$. To establish continuity, we will show that $\lambda$ is uniformly bounded in $\mathcal{U}\cap \Q^+$. Since (\ref{lambdabound2}) gives a bound on $\lambda$ in $\mathcal{U}\cap (\mathcal{R}\cup\mathcal{A})$, it suffices to obtain a bound in $\mathcal{U}\cap \mathcal{T}$.
  
Re-writing (\ref{eqn:ruv}) in terms of the Hawking mass and (\ref{tuv}), we note that in the trapped region
\beqn\label{newlambda}
\pu (-\lambda) = \frac{2}{r^2}\left(m-\frac{\sigma^2}{2r}\right)\frac{\lambda}{1-\frac{2m}{r}}(-\nu)\leq \frac{2}{r^2}\frac{m\lambda}{1-\frac{2m}{r}}(-\nu),
\eeqn
where
\beqn\nonumber
\sigma^2 = 16\pi \Omega^{-2}r^4T_{uv}\geq 0.
\eeqn
Note that by (\ref{hm}) we have  $\frac{2m}{r}>1$ in $\mathcal{T}$ (cf.~Statement V).

In view of the assumption that $(U,v)$ is in the interior of $\chg$, there exists a sufficiently small $\delta >0$ such that 
$\{|v'-v|\leq \delta\}\cap \Sigma$ is compact and, moreover,
\beqn\nonumber
r >r_0 = \inf_{\{|v'-v|\leq \delta\}\cap \Sigma}  r_{\inf}(U,v')>0.
\eeqn
Since $\nu <0$ in $\Q^+$, it follows that $r> r_0$ in $\{|v'-v|\leq \delta\}\cap \mathcal{U}\cap \Q^+$.

Define the sets $\mathcal{U}_1$ and $\mathcal{U}_2$ by
\beqn\nonumber
\mathcal{U}_1 = \left\{\frac{2m}{r}\geq 2\right\}\cap\{|v'-v|\leq \delta\}\cap \mathcal{U}\cap \mathcal{T}\hspace{.5cm}\textrm{and}\hspace{.5cm} \mathcal{U}_2 = \left\{\frac{2m}{r} < 2\right\}\cap\{|v'-v|\leq \delta\}\cap \mathcal{U}\cap\mathcal{T}.
\eeqn
Integrating the right-hand side of (\ref{newlambda}) along segments $\left([u_0, u]\times \{v'\}\right)\cap \mathcal{U}\cap \mathcal{T}$, we obtain
\beqn\nonumber
 \int_{u_0}^u \frac{2}{r^2}\frac{m\lambda}{1-\frac{2m}{r}}(-\nu)~\dd \overline{u}= \int_{[u_0, u]\cap \mathcal{U}_1}\frac{|\lambda| (-\nu)}{r}\frac{\frac{2m}{r}}{\frac{2m}{r}-1}~\dd\overline{u}+ \int_{[u_0, u]\cap \mathcal{U}_2}\frac{2m}{r}\frac{\lambda}{1-\frac{2m}{r}}\frac{-\nu}{r}~\dd\overline{u}.
\eeqn
On the set $\mathcal{U}_1$, we note that
\beqn\nonumber
0<\frac{\frac{2m}{r}}{\frac{2m}{r}-1} \leq 2.
\eeqn
Thus,
\beqn\nonumber
\int_{[u_0, u]\cap\mathcal{U}_1}\frac{|\lambda| (-\nu)}{r}\frac{\frac{2m}{r}}{\frac{2m}{r}-1}~\dd\overline{u}\leq 2\int_{[u_0, u]\cap\mathcal{U}_1}\frac{|\lambda| (-\nu)}{r}~\dd\overline{u}.
\eeqn
From compactness, there exist $u_0 < U$ and a positive constant $C(u_0)<\infty$ such that
\beqn\nonumber
0<\frac{\lambda}{1-\frac{2m}{r}}(u, v')\leq \frac{\lambda}{1-\frac{2m}{r}}(u_0,v')\leq C(u_0)
\eeqn
in $\{|v'-v|\leq \delta\}\cap\{u\geq u_0\} \cap\mathcal{U}\cap \Q^+$ by monotonicity (\ref{eqn:conuu}). Hence,
\beqn\nonumber
 \int_{[u_0, u]\cap\mathcal{U}_2}\frac{2m}{r}\frac{\lambda}{1-\frac{2m}{r}}\frac{-\nu}{r}~\dd\overline{u}\leq2 C(u_0)\int_{[u_0, u]\cap\mathcal{U}_2}\frac{-\nu}{r}~\dd\overline{u}\leq C(u_0, r_0)<\infty,
\eeqn
having used the lower bound on $r$.
It follows, from integration of (\ref{newlambda}), that
\beqn\nonumber
 |\lambda| (u, v')\leq C(u_0, r_0) + 2\int_{[u_0, u]\cap\mathcal{U}_1}\frac{|\lambda| (-\nu)}{r}~\dd\overline{u}.
\eeqn
Applying a Gr\"onwall inequality then yields
\beqn\nonumber
|\lambda|(u,v')\leq C(u_0, r_0)\exp\left(\int_{[u_0, u]\cap\mathcal{U}_1}\frac{-\nu}{r}~\dd\overline{u}\right),
\eeqn
whence we obtain
\beqn\nonumber
\sup_{ \{|v'-v|\leq \delta\}\cap \{u\geq u_0\}\cap\mathcal{U}\cap \mathcal{T}}|\lambda|(u, v') \leq C(u_0, r_0)< \infty.
\eeqn
The result follows.  

\subsubsection*{$r$ extends continuously to $\infty$ on $\mathcal{I}^+$} Let $p = (u,V)\in \mathcal{I}^+$ and consider an arbitrary sequence of points $q_j=(u_j,V_j)\in J^-(p)\cap\Q^+$ with $q_j\rightarrow p$. Since $\lambda >0$ and $\nu <0$ in $J^+(\mathcal{I}^+)\cap \Q^+$, it follows that
\beqn\nonumber
\inf_{J^-(p) \cap J^+(q_j)\cap \Q^+} r(q) = r(u,V_j).
\eeqn
Since $r(u,V_j)\rightarrow \infty$ as $V_j\rightarrow V$, we have $r_{\inf}(p) = \infty$, establishing the claim.

 \subsubsection*{$r$ extends continuously to $\infty$ on $ i^0$}
Continuity of $r$ on $i^0$ follows immediately from an appropriate notion of `asymptotically flat' in Theorem \ref{thm:main}.

\subsection{Statement III}
Follows from the results given by Dafermos in \cite{MD05b}. 
\subsection{Statement IV}\label{statement:iv}
\subsubsection{Claim 1}
Let us begin by re-normalizing the co-ordinates on $\Q^+$ such  that $\Gamma$ can be expressed as
\beqn\nonumber
\Gamma = \{(u, v=u)\}\cap \Q^+.
\eeqn
Define now a (timelike) vector field $\textbf{T}$ on $\Q^+$ by
\beqn\nonumber
\textbf{T} = \frac{\partial}{\partial u} + \frac{\partial}{\partial v}.
\eeqn
Since we have
\beqn\nonumber
0 = \textbf{T} r|_{\Gamma} = \left( \nu + \lambda \right)|_{\Gamma},
\eeqn
and since $\nu<0$ in $\Q^+$, it follows that
\beqn\nonumber
\lambda|_{\Gamma} >0.
\eeqn 
This establishes the claim.
\subsubsection{Claims 2--3}
The claims follow immediately from monotonicity (\ref{eqn:convv}). See \cite{DC95}.
\subsubsection{Claim 4}
Let $(u,V)\in \textrm{int}\left(\mathcal{CH}_{i^+}\right)$ and consider a neighborhood $\mathcal{U}\subset \R^{1+1}$ of $(u,V)$ such that either $\mathcal{U}\cap \Q^+\subset\mathcal{A}$ or $\mathcal{U}\cap \Q^+\subset\mathcal{T}$.  To establish continuity, it suffices to show that $\nu$ is uniformly bounded in $\mathcal{U}\cap \Q^+$.

If $\mathcal{U}\cap \Q^+\subset \mathcal{A}$, we simply note that $r$ is constant along outgoing null segments in $\mathcal{U}\cap \Q^+$. It follows that $\nu$ is uniformly bounded in this case.

From (\ref{hm}) and monotonicity (\ref{eqn:convv}), we note that $\nu$ and $1-\frac{2m}{r}$ have the same sign in the trapped region and, moreover, that
\beqn\label{mono2}
\pv \left(\frac{\nu}{1-\frac{2m}{r}} \right) \leq 0.
\eeqn
If $\mathcal{U}\cap \Q^+\subset \mathcal{T}$, we retrieve a uniform bound on $\nu$ by proceeding exactly as in the case in which we established the uniform bound on $\lambda$ in the trapped region when proving the extendibility of $r$ on $\textrm{int}(\chg)$.  We simply note that $\pu \lambda = \pv \nu$ and that monotonicity (\ref{mono2}) takes the role of (\ref{monokappa}). Moreover, we note that since $\lambda <0$ in $\mathcal{T}$, this monotonicity takes the role of $\nu<0$ in ensuring that $r$ is bounded from below away from zero in $\mathcal{U}\cap \Q^+$. 

\begin{remark} We do not say anything about the case in which $\mathcal{U}\cap\Q^+\subset \mathcal{A}\cup \mathcal{T}$, as there is no possible uniform control on $\frac{\nu}{1-\frac{2m}{r}}$.  
\end{remark}
\subsubsection{Claim 5}
\noindent(a)  If $\mathcal{A}\neq\emptyset$, then the limit points of $\mathcal{A}$ on $b_{\Gamma}\cup \sgo\cup\chg$ necessarily form a (possibly degenerate) connected closed interval.  This follows from monotonicity given in Claim 2.  On the other hand, there is \emph{a priori} no characterization of the limit points of $\mathcal{A}$ on $\mathcal{CH}_{i^+}\cup i^{+}$ (cf.~the diagram given in the statement of Theorem \ref{thm:main}). \\

\noindent(b)--(e) The claims follow from monotonicity given in Claim 2 (cf.~\S\ref{ngamma}).
\subsection{Statement V}
\subsubsection{Claim 1}
The claim is given by (\ref{hawkingmassv}).
\subsubsection{Claim 2}
The claim follows immediately from (\ref{hm}) since $\lambda \geq0$ and $\nu <0$ in $\mathcal{R}\cup \mathcal{A}$. 
\subsection{Statement VI}\label{statementvi}
The statement follows from the results given by Dafermos in \cite{MD05b}.

\subsection{Statement VII}

In spherically symmetric spacetimes, the Kretschmann scalar satisfies
\beqn\nonumber
R_{\mu\nu\alpha\beta}R^{\mu\nu\alpha\beta}= 4K^2 +  \frac{4}{r^4}\left(\frac{2m}{r}\right)^2 + \frac{12}{r^2}\nabla_a\nabla_b r\nabla^a\nabla^b r,
\eeqn 
where $K$ is the Gaussian curvature of the quotient metric $g_{ab}$, which is given by
\beqn\nonumber
K = 4\Omega^{-2}\left(\Omega^{-1}\pu\pv\Omega - \Omega^{-2}\pu\Omega\pv\Omega\right).
\eeqn
One can compute (see the appendix of \cite{MDAR07}), 
\beqn\nonumber
\nabla_a\nabla_b r\nabla^a\nabla^b r = 2\left(\frac{m}{r^2} + 2\pi r \textrm{tr}T\right)^2 + 8\pi^2r^2 (T_{ab} - \frac{1}{2}g_{ab}\textrm{tr}T) (T^{ab} - \frac{1}{2}g^{ab}\textrm{tr}T),
\eeqn
where $\textrm{tr}T = g^{ab}T_{ab}$. The last term on the right-hand side is manifestly non-negative if the energy momentum tensor obeys the null energy condition, for
\beqna
(T_{ab} - \frac{1}{2}g_{ab}\textrm{tr}T) (T^{ab} - \frac{1}{2}g^{ab}\textrm{tr}T)  &=& T_{ab}T^{ab} - \frac{1}{2}(\textrm{tr}T)^2\nonumber\\
&=& 2(g^{uv})^2T_{uu}T_{vv}\nonumber.
\eeqna
In this case, the Kretschmann scalar satisfies
\beqn\label{ks}
R_{\mu\nu\alpha\beta}R^{\mu\nu\alpha\beta}\geq \frac{4}{r^4}\left(\frac{2m}{r}\right)^2.
\eeqn
\subsubsection{Claim 1}\label{blowup}
For every $p\in \sgt\cup\mathcal{S}\cup \mathcal{S}_{i^+}$, there exists a neighborhood $\mathcal{U}\subset \R^{1+1}$ of $p$ such that $\mathcal{U}\cap \Q^+ \subset \mathcal{T}$. Recalling that $\frac{2m}{r}>1$ in the trapped region, we have from (\ref{ks}),
\beqn\nonumber
R_{\mu\nu\alpha\beta}R^{\mu\nu\alpha\beta}>\frac{4}{r^4}.
\eeqn
For every sequence $\{p_j\}_{j=1}^{\infty}\subset \mathcal{U}\cap \Q^+$ with $p_j\rightarrow p$, we, therefore, obtain 
\beqn\nonumber
\lim_{j\rightarrow \infty} R_{\mu\nu\alpha\beta}R^{\mu\nu\alpha\beta}(p_j)> \lim_{j\rightarrow\infty} \frac{4}{r^4}(p_j) = \infty.
\eeqn
This establishes the claim.

\subsubsection{Claim 2} 
\noindent (a) Let $(U,v)\in \sgo$ and consider a neighborhood $\mathcal{U}\subset \R^{1+1}$ of $(U,v)$. If $(U,v)$ is a limit point of $\{(U_j, v_j)\}_{j=1}^{\infty}\subset\mathcal{A}\cup\mathcal{T}$, then the result is obvious since 
\beqn\nonumber
\limsup_{j\rightarrow\infty}R_{\mu\nu\alpha\beta}R^{\mu\nu\alpha\beta} (U_j, v_j) \geq \limsup_{j\rightarrow \infty}\frac{4}{r^4}(U_j, v_j) = \infty.
\eeqn
Thus, without loss of generality, we can assume that $\mathcal{U}\cap \Q^+\subset\mathcal{R}$. Moreover, let us assume that $(U,v)$ does not lie on $\mathcal{CH}_{i^+}\cup i^{\square}$.
Fix $u_0 < U$ such that $(u_0, v)\in \mathcal{U}\cap\Q^+$. Since $r$ extends continuously to zero on $\sgo\backslash \left(\mathcal{CH}_{i^+}\cup i^{\square}\right)$, it follows that
\beqn\nonumber
\lim_{u\rightarrow U}\int_{v-\delta}^{v+\delta} \lambda (u, v')~\dd v'=0
\eeqn
in $\{u\geq u_0\}\cap \{|v'-v|\leq \delta\}\cap \mathcal{U}\cap \Q^+$, for all sufficiently small $\delta>0$. Integration of (\ref{newlambda})  gives
\beqn\nonumber
\int_{v-\delta}^{v+\delta} \lambda (u_0, v')\exp\left(-\int^{r(u_0, v')}_{r(u,v')} \frac{\frac{1}{r}\left(\frac{2m}{r}\right)}{1-\frac{2m}{r}}~\dd r\right)~\dd v'\leq \int_{v-\delta}^{v+\delta}\lambda(u, v')~\dd v'\rightarrow 0
\eeqn 
as $u\rightarrow U$ in $\{|v'-v|\leq \delta\}\cap \mathcal{U}\cap \Q^+$. By compactness, there exists a constant $c>0$ such that $\lambda(u_0, v') > c >0$ for all $\{|v'-v|\leq \delta\}\cap \mathcal{U}\cap \Q^+$.
Thus, it must be that
\beqn\nonumber
\lim_{u\rightarrow U}\int_{v-\delta}^{v+\delta}\exp\left(-\int^{r(u_0, v')}_{r(u,v')}\frac{\frac{1}{r}\left(\frac{2m}{r}\right)}{1-\frac{2m}{r}}~\dd r\right)~\dd v'= 0.
\eeqn
In particular, we obtain
\beqn\label{boundinf}
\sup_{v-\delta \leq v'\leq v+\delta}\int^{r(u_0, v')}_{r(u,v')}\frac{\frac{1}{r}\left(\frac{2m}{r}\right)}{1-\frac{2m}{r}}~\dd r = \infty.
\eeqn
Let $\epsilon>0$. Suppose for the moment that 
\beqn\nonumber
\frac{2m}{r}(u, v')\leq r^{\epsilon}(u, v')
\eeqn
in $\{u\geq U-\hat{\epsilon}\}\cap\{|v'-v|\leq \delta\}\cap \mathcal{U}\cap \Q^+$ for some $\hat{\epsilon}>0$. Since
\beqn\nonumber
\int \frac{x^{-1 + \epsilon}}{1 - x^\epsilon} ~\dd x = - \epsilon^{-1} \log (1- x^\epsilon),
\eeqn
we have, shrinking $\mathcal{U}$ if necessary to ensure that $r<1$ in $\mathcal{U}\cap \Q^+$,
\beqn\nonumber
\sup_{v-\delta \leq v'\leq v+\delta}\int^{r(u_0, v')}_{r(u,v')}\frac{\frac{1}{r}\left(\frac{2m}{r}\right)}{1-\frac{2m}{r}}~\dd r\leq C(\epsilon, \hat{\epsilon}) <\infty,
\eeqn
contradicting (\ref{boundinf}). Thus, we conclude that there exists a sequence of points
\beqn\nonumber
(U_j, v_j)\in \{u\geq U-\hat{\epsilon}\}\cap\{|v'-v|\leq \delta\}\cap \mathcal{U}\cap \Q^+
\eeqn such that
\beqn\label{seq1}
\frac{2m}{r}(U_j, v_j) \geq r^{\epsilon}(U_j,v_j).
\eeqn
Letting $\delta \rightarrow 0$, we can construct a sequence $(U_{j_k}, v_{j_k})\rightarrow (U,v)$ for which (\ref{seq1}) holds.
The result follows from (\ref{ks}) since
\beqn\nonumber
\limsup_{j\rightarrow \infty} R_{\mu\nu\alpha\beta}R^{\mu\nu\alpha\beta}(U_{j_k}, v_{j_k})\geq \limsup_{j\rightarrow \infty}  \frac{4}{r^{4-2\epsilon}}(U_{j_k}, v_{j_k})= \infty.
\eeqn

Lastly, we note that if $(U,v)$ does lie on $\mathcal{CH}_{i^+}\cup i^{\square}$, then there is a sequence of points $(U, v_j)\rightarrow (U,v)$ on $\sgo\backslash (\mathcal{CH}_{i^+}\cup i^{\square})$ for which the above argument applies. In particular, there exists a sequence of points $(U_i, v_j)$ such that the Kretschmann scalar blows up as $U_i\rightarrow U$.  Thus, we can construct a suitable subsequence of points $(U_{i_k}, v_{j_k})$ for which the Kretschmann scalar blows up as $(U_{i_k}, v_{j_k})\rightarrow (U,v)$.\\

(b) As in Claim 1, given $p\in\sgo$, if there exists a neighborhood $\mathcal{U}\subset \R^{1+1}$ of $p$ such that $\mathcal{U}\cap \Q^+ \subset \mathcal{A}\cup \mathcal{T}$, then the Kretschmann scalar will extend continuously to $\infty$ on $\sgo\cap \mathcal{U}$.
\subsubsection{Claim 3}
We deduce from Theorem \ref{thm:gext} that if $\mathcal{BH} \neq \emptyset$,
then $J^-(\mathcal{I}^+) \cap \Q^+$ has a future boundary in $\Q^+$, which we call $\mathcal{H}^+$. In fact, for $\mathcal{I}^+ \subset \{v = V\}$, we have
\beqn\nonumber
\mathcal{H}^+ =  \{u = U\} \cap\Q^+\subset\mathcal{R}\cup \mathcal{A}
\eeqn
for some $U$. Fix $v_0 <V$ and consider the set $\mathcal{H}^+\cap \{v\geq v_0\} =
\{U\}\times [v_0,V)$.  Without loss of generality, we can assume that $(U, v_0)\not\in \Gamma$.

Since $\lambda \geq 0$ along $\mathcal{H}^+$, the area-radius $r$ has a well-defined limit $r_+$, which is given by
\beqn\nonumber
r_+ = \sup_{\mathcal{H}^+} r.
\eeqn
Similarly, by monotonicity of the Hawking mass (\ref{hawkingmassv}), the Hawking mass  has a well-defined limit $m_+$, which is given by 
\beqn\nonumber
m_+ = \sup_{\mathcal{H}^+} m.
\eeqn
Moreover, as $\frac{2m}{r}\leq 1$ along $\mathcal{H}^+$, the constants $r_+$ and $m_+$ satisfy the bound
\beqn\nonumber
m_+ \leq \frac{1}{2}r_+.
\eeqn
By assumption, we consider the case of strict inequality.  To establish the claim, we will show that there exists $u_*> U$ such that $r\geq \frac{1}{2}r_0$ in
\beqn\nonumber
\Dm = \left([U, u_*] \times [v_0, V]\right)\cap \Q^+,
\eeqn
for $r_0 = r(U,v_0)$.
We proceed with a bootstrap argument.  

Define a region $\Dm^* \subset \Dm$ to be the set of $p= (u,v)\in \Dm$ such that
\beqn\label{bootstrap}
r(\ut,\vt) > \frac{1}{2}r_0
\eeqn
for all $(\ut,\vt)\in \left(J^-(p)\backslash\{p\}\right)\cap \Dm$. Since $r\geq r_0$ along $\mathcal{H}^+$, it is clear that $\Dm^*$ is a connected non-empty open set (in the topology of $\Dm$).  The claim follows once we show that $\Dm^*$ is also closed (in the topology of $\Dm$).  That is, for $p\in \overline{\Dm^*}$ (in the topology of $\Dm$), we must show that $p\in \Dm^*$.  Note by continuity we have $r \geq \frac{1}{2} r_0$ in $J^-(p)\cap \Dm$.

Define a function $\kappa$ by
\beqn\nonumber
\kappa = \frac{\lambda}{1-\frac{2m}{r}}.
\eeqn
Since
\beqn\nonumber
\lim_{v\rightarrow V} \left(1-\frac{2m}{r}\right)(U,v) \neq 0, 
\eeqn
monotonicity (\ref{eqn:convv}) implies that
\beqn\nonumber
\inf_{\mathcal{H}^+} \left(1 - \frac{2m}{r} \right)>0.
\eeqn
It follows that
\beqn\nonumber
\int_{v_0}^v \kappa (U,\overline{v})~\dd \overline{v} \leq C < \infty.
\eeqn
On the other hand,  from monotonicity (\ref{eqn:conuu}), we deduce that $\pu \kappa \leq 0$. Thus, for $u'\in [U, u]$,
\beqn\label{kappaint}
\int_{v_0}^{v} \kappa (u',\overline{v}) ~\dd \overline{v}\leq\int_{v_0}^{v}\kappa(U,\overline{v})~\dd \overline{v} \leq C<\infty.
\eeqn

By monotonicity (\ref{eqn:convv}), we have that $\lambda$ can change sign at most once along an outgoing null segment. Thus,  
\beqn\label{lambdaint}
\int_{v_0}^{v} |\lambda|(u',\overline{v})~\dd \overline{v} \leq C<\infty.
\eeqn 

Since (\ref{newlambda}) gives
\beqn\label{newni2}
\pv \log (-\nu) = \frac{2\kappa}{r^2}\left(m-\frac{\sigma^2}{2r}\right) \leq \frac{\kappa}{r}\left(\frac{2m}{r}\right),
\eeqn
let us consider the sets
\beqn\nonumber
\mathcal{V}_1 = \left\{\frac{2m}{r}\geq 2\right\}\cap \Dm^* \hspace{1cm}\textrm{and}\hspace{1cm} \mathcal{V}_2 = \left\{\frac{2m}{r} < 2\right\}\cap \Dm^*.
\eeqn
On the set $\mathcal{V}_1$, we note that
\beqn\nonumber
0< \frac{\frac{2m}{r}}{\frac{2m}{r}-1}\leq 2.
\eeqn
Using the bootstrap assumption, (\ref{kappaint}) and (\ref{lambdaint}) then give a uniform bound on
\beqna\nonumber
\int_{v_0}^{v} \frac{\kappa}{r}\left(\frac{2m}{r}\right)~\dd \overline{v} &=& \int_{[v_0, v]\cap\mathcal{V}_1} \frac{|\lambda|}{r}\left(\frac{\frac{2m}{r}}{\frac{2m}{r}-1}\right)~\dd \overline{v}  + \int_{[v_0, v]\cap\mathcal{V}_2} \frac{\kappa}{r}\left(\frac{2m}{r}\right)~\dd \overline{v}.
\eeqna
In particular, integration of (\ref{newni2}) gives the estimate
\beqn\nonumber
(-\nu)(u',v) \leq C (-\nu)(u',v_0)\leq C \sup_{U\leq u'\leq \us} (-\nu)(u',v_0)\leq C <\infty,
\eeqn
from which integration in $u$ then gives
\beqn\nonumber
r(u,v) \geq r(U, v) - C (u- U)\geq r_0 - C (u- U).
\eeqn
For sufficiently small $u_*-U$ we obtain
\beqn\nonumber
r(u,v) \geq \frac{3}{4}r_0,
\eeqn
establishing the result. 

\subsubsection{Claim 4}

Suppose $(\M, g_{\mu\nu})$ is future-extendible as a $C^2$-Lorentzian manifold $(\widetilde{\mathcal{M}}, \widetilde{g_{\mu\nu}})$.  Then, there exists a timelike curve $\gamma\subset\widetilde{\M}$ exiting $\M$ such that the closure of the projection of $\gamma|_{\M}$ to $\Q^+$ intersects the boundary $\mathcal{B}^+\backslash i^0$.  Let $\mathcal{E}$ denote the set of all such boundary points on $\mathcal{B}^+\backslash i^0$ satisfying the above, i.e.,~
\begin{equation}
\begin{split}
\mathcal{E} =& \left\{ p\in \mathcal{B}^+\backslash i^0: \exists~ \textrm{timelike}~\gamma: I \rightarrow \widetilde{\M} ~~\textrm{with} ~~\gamma(t^*) \in \partial \M\subset \widetilde{\M}~~\textrm{and}\right. \hspace{1cm}\\
  &\hspace{5.5cm} \left. \hspace{3.5cm} ~~\gamma(t)\in \M ~~\forall t< t^*~~\textrm{s.t.} ~~\{p\}\cap\overline{\pi\left(\gamma|_{\M}\right)}\neq \emptyset\right\}.
\end{split}\nonumber
\end{equation}
It is important to emphasize that extendibility is not formulated \emph{per se} with respect to the quotient manifold $\Q^+$.  If we were able, however, to assert that past set structure `upstairs' is preserved  `downstairs', then criteria can be given on $\Q^+$ that will imply inextendibility of $\M$. These criteria on $\Q^+$ will be given by $C^2$-compatible scalar invariants, which we now introduce.

\subsubsection*{$C^2$-compatible scalar invariants}  Let $\gamma$ be as above and let $x\in \partial \mathcal{M}\subset \widetilde{\M}$ be the point through which $\gamma$ exits the spacetime and consider $x'\in \gamma|_{\M}$ sufficiently close to $x$ such that $J^+(x')\cap J^-(x)\subset \widetilde{\M}$ is compact. We call $\Xi$ a $C^2$-compatible scalar invariant on $\widetilde{\M}$ if $\Xi$ remains uniformly bounded in the spacetime region $J^+(x')\cap J^-(x)\cap \M$. Two examples of $C^2$-compatible scalar invariants are given by the Kretschmann scalar $R_{\mu\nu\alpha\beta}R^{\mu\nu\alpha\beta}$ and the norm $g(X,X)$ of any Killing vector field $X$.  In the latter case, $X$, which satisfies
\beqn\nonumber
\nabla_{\nu}\nabla_\alpha X_{\beta} = R_{\mu\nu\alpha\beta}X^{\mu},
\eeqn 
is continuous on $\M\cup \partial\M$, cf.~\cite{MDAR05b}.
 
\subsubsection*{Extendibility criteria on $\Q^+$}

Let us begin by quoting a standard result, a proof of which can be found, for example, in \cite{PCJC10}.
\begin{prop}\label{prop:extstand}Let $\gamma$ be a causal curve in $(\M, g)$.  If $(\M, g)$ is $C^2$-future-extendible, then $\gamma|_{\M}$ is an incomplete causal curve in $\M$ and every $C^2$-compatible invariant $\Xi$ remains bounded along $\gamma|_{\M}$.
\end{prop}

From this result, one easily show the following:
\begin{lem}Let $P = I_\M^-(\gamma|_{\M})\subset\M$. If
\beqn\label{past}
\pi(P) = I^-(\pi(P))\cap \Q^+,
\eeqn
then for $p\in \mathcal{E}$ the following hold:
\begin{itemize}
\item [1.] $I^-(p)\cap \Q^+= \pi(P)$; and,
\item [2.] there exists a $p'\in \pi(\gamma|_{\M})$ such that every $C^2$-compatible scalar invariant $\Xi$ is uniformly bounded in the region $J^-(p)\cap J^+(p')\cap \Q^+$.\end{itemize}
\end{lem}
The hypothesis that (\ref{past}) holds asserts that past set structure is preserved under projection, for   
$P$ satisfies $P = I_\M^-(P)$.  Once this assumption is retrieved, then showing that there does not exist a timelike curve $\gamma$ that exits the spacetime through $x\in \partial \M$ reduces to finding a $C^2$-compatible scalar invariant $\Xi$ that is \emph{not} uniformly bounded in $J^-(p)\cap J^+(p')\cap \Q^+$ \emph{for every} $p'\in \pi(\gamma|_{\M})$.

\subsubsection*{Past set structure is preserved}
It is easy to show that since the metric $h$ on $\mathbb{S}^2$ is positive definite, timelike (resp.,~causal) vectors on $T_p\M$, for $p\in \M$, project to timelike (resp.,~causal) vectors on $T_{\pi(p)}\Q^+$.  On the other hand, null vectors need not project to null vectors, for if $\pi_h: \M\rightarrow \mathbb{S}^2$ denotes the standard projection map, then a null vector $V\in T_p\M$ will map to a timelike vector $\dd\pi (V) \in T_{\pi(p)}\Q^+$ unless $\dd\pi_h(V) =0$, i.e., unless $V$ is horizontal to the leaves $\Q^+\times\{x\}$ for $x\in \mathbb{S}^2$, in which case $\dd\pi(V)$ is also null.  Lastly, note that the horizontal lift of a timelike (resp.,~causal) vector in 
$T_q\Q^+$ is a timelike (resp.,~causal) vector in $T_p\M$, with $\pi(p) = q$. 

To establish (\ref{past}) we will show that
\beqn\label{past1}
\pi(P) = I^-(\pi (\gamma|_{\M}))\cap \Q^+,
\eeqn 
for then
\beqn\nonumber
I^-(\pi(P))\cap \Q^+ = I^-(I^-(\pi(\gamma|_{\M}))\cap \Q^+)= I^-(\pi(\gamma|_{\M}))\cap \Q^+ = \pi(P).
\eeqn
Consider $q\in \pi(P)$.  Then, there exists $p\in P$ such that $\pi(p) = q$.  Since $P$ is a past set, it follows that for some $p'\in \gamma|_{\M}$, there exists a timelike curve $\gamma_p:I \rightarrow \M$ such that $\gamma_p(0) = p$ and $\gamma_p(1) = p'$.  We then define a timelike vector field $X\in T\M$ along $\gamma_p$ such that
\beqn\nonumber
X(t) = \dot{\gamma}_p(t).
\eeqn
This vector field then projects to a timelike vector field $\dd\pi(X)\in T\Q^+$ along $\pi(\gamma_p)$.  In particular, there exists a future-directed timelike curve that connects $q$ and $\pi(p')\in \pi(\gamma|_{\M})$, i.e.,~$q\in I^-(\pi (\gamma|_{\M}))\cap \Q^+$.
To show the other inclusion, consider $q\in I^-(\pi (\gamma|_{\M}))\cap \Q^+$.  Then, by definition, there exists $q'\in \pi (\gamma|_{\M})$ and a timelike curve $\gamma_q: I \rightarrow \Q^+$ such that $\gamma_q(0) =q$ and $\gamma_q(1) = q'$.  We define a timelike vector field $Y\in T\Q^+$ along $\gamma_q$ such that
\beqn\nonumber
Y(t) = \dot{\gamma}_q(t).
\eeqn
The curve $\gamma_q$ has a unique horizontal lift $\hat{\gamma}_q$ in $\M$ through $p'\in \gamma|_{\M}$ such that $\pi(\hat{\gamma}_p) = \gamma_q$ and $\pi(p') = q'$. Since the horizontal lift of a timelike vector is also timelike, it follows that there exists a timelike vector field $\hat{Y}\in T\M$ along $\hat{\gamma}_q$ such that $\dd\pi(\hat{Y}) = Y$. We then follow along the lifted timelike curve from $p'$ to the point $p\in P$ such that $\pi(p) = q$. In particular, $q\in \pi(P)\cap \Q^+$.

\subsubsection*{Characterization of the set $\mathcal{E}$} In characterizing the extension, we consider three cases.   

\textbf{Case 1:} $\mathcal{E}\cap(\sgo\cup\sgt \cup\mathcal{S} \cup \mathcal{S}_{i^+})= \emptyset$\\
For the sake of contradiction, suppose that $\mathcal{E}\cap(\sgo\cup\sgt \cup\mathcal{S} \cup \mathcal{S}_{i^+})\neq \emptyset$ and let $p\in\mathcal{E}\cap \left(\sgo\cup\sgt \cup\mathcal{S} \cup \mathcal{S}_{i^+}\right)$.  By definition, there exists a timelike curve $\gamma\subset\widetilde{\M}$ exiting $\M$ such that $\{p\}\cap \overline{\pi(\gamma|_{\M})}\neq \emptyset$.  For all $p' \in \pi(\gamma|_{\M})$, Claims 1 and 2 assert that the Kretschmann scalar $R_{\mu\nu\alpha\beta}R^{\mu\nu\alpha\beta}$ is not uniformly bounded in $J^-(p) \cap J^+(p')\cap \Q^+$, contradicting the assumption that $p\in \mathcal{E}$. We conclude  that  $\mathcal{E}\cap(\sgo\cup\sgt\cup\mathcal{S} \cup \mathcal{S}_{i^+})= \emptyset$.

\textbf{Case 2:} $\mathcal{E}\cap \left(i^{\square}\cup\mathcal{I}^+\right)=\emptyset$\\
For the sake of contradiction, suppose that $\mathcal{E}\cap \left(i^{\square}\cup\mathcal{I}^+\right)\neq \emptyset$ and let $p\in \mathcal{E}\cap \left(i^{\square}\cup\mathcal{I}^+\right)$. By definition, there exists a timelike curve $\gamma\subset\widetilde{\M}$ exiting $\M$ such that $\{p\}\cap \overline{\pi(\gamma|_{\M})}\neq \emptyset$. We recall the fact that in spherical symmetry there are three Killing vector fields $X_1$, $X_2$, and $X_3$ such that 
\beqn\nonumber
 g(X_1, X_1) + g(X_2, X_2) +g(X_3, X_3) = 2 r^2.
\eeqn
Since $r$ is unbounded in $J^+(p'')\cap J^-(\mathcal{I}^+)\cap \Q^+$ for all $p'' \in \pi(\gamma|_{\M})$,  it follows that, without loss of generality, $g(X_1, X_1)$ is unbounded as well, contradicting the assumption that $p\in \mathcal{E}$. We conclude that $\mathcal{E}\cap \left(i^{\square}\cup\mathcal{I}^+\right)= \emptyset$.

\textbf{Case 3:} $\mathcal{E}\cap (\mathcal{CH}_{\Gamma}\cup\mathcal{CH}_{i^+})\neq \emptyset$\\
To establish the result, it suffices to show that
\beqn\nonumber
\mathcal{E} \neq \{b_{\Gamma}\},
\eeqn
i.e., one cannot simply extend through $b_{\Gamma}$, for in light of the above, it will follow that
\beqn\nonumber
\mathcal{E}\cap (\mathcal{CH}_{\Gamma}\cup\mathcal{CH}_{i^+})\neq \emptyset.
\eeqn 
Before we prove this claim, let us briefly summarize the main idea. Since the extension $\widetilde{\M}$ is a regular manifold, if $\gamma$ exits through $x\in \partial \M\subset\widetilde{\M}$ with the property that $\{b_{\Gamma}\}\cap \overline{\pi(\gamma|_{\M})}\neq \emptyset$, then there must be an open neighborhood $\widetilde{\mathcal{U}}\subset\widetilde{\M}$ of $x$ through which a suitable family of timelike curves (to be defined later) also exits into the extension. The worry is that this (and every) family of timelike curves when restricted to $\M$ will be crushed under the projection $\pi$ so as to always obtain a limit point on $b_{\Gamma}$.  For this to happen, however, both the $u$- and $v$-dimension of the tubular neighborhood swept out by this family must degenerate, i.e., become zero, as $\partial\M$ is approached. We demonstrate that this is not possible.  

Fix a curve $\gamma$ that exits $\M$ through $x\in \partial \M$  with the property that $\{b_{\Gamma}\}\cap \overline{\pi(\gamma|_{\M})}\neq \emptyset$.  We can choose $x\in \partial \M$, without loss of generality, such that there is an open neighborhood $\widetilde{\mathcal{U}}\subset\widetilde{\M}$ of $x$ such that $\{p\in \M: g'  \cdot p= p,~\forall g'\in SO(3)\}\cap \widetilde{\mathcal{U}} = \emptyset$ since the set of fixed points in $\M$ (the center $\Gamma$ `downstairs') is a curve in $\M$. Indeed, such a choice is possible because $\partial \M$ is a 3-dimensional Lipschitz submanifold of $\widetilde{\M}$.  At every point $x'\in \gamma|_{\M}$ there exists an open neighborhood $\widetilde{\mathcal{U}'}\subset\widetilde{\M}$ such that the exponential map $\exp_{x'}: T_{x'}\widetilde{\M} \rightarrow \widetilde{\M}$ induces a diffeomorphism
\beqn\nonumber
\exp_{x'}\left(B^4_{\epsilon}(0)\right)\rightarrow \widetilde{\mathcal{U}'},
\eeqn
where $B^4_{\epsilon}(0)\subset T_{x'}\widetilde{\M}$ is the 4-dimensional ball of radius $\epsilon$ centered around $0\in T_{x'}\widetilde{\M}$. Let $B^3_{\epsilon/2}(0)\subset B^4_{\epsilon}(0)$ be a 3-dimensional ball that intersects (under the image of the exponential map) $\gamma$ transversally.  Consider a family of disjoint timelike curves $\gamma_{\vec{s}} = \gamma_{(s_1, s_2, s_3)}$ that emanate from $\exp_{x'}\left(B^3_{\epsilon/2}(0)\right)$ such that
\begin{enumerate}
\item [a.] $\gamma_{\vec{s}}$ intersect $\exp_{x'}\left(B^3_{\epsilon/2}(0)\right)$ transversely; and,
\item [b.] $ \gamma_{(0,0,0)} = \gamma$.
\end{enumerate}    
For $x'\in \gamma|_{\M}$ suitably close to $x$, we can, moreover, require that for some suitably small $\epsilon_1$
\begin{enumerate}
\item [c.] $\gamma_{\vec{s}}(0)\in \M$ and $\gamma_{\vec{s}}(\frac{1}{2}\epsilon_1)\not\in \M$ for all $\vec{s}\in B^3_{\epsilon/2}(0)$.
\end{enumerate}
For each of these curves, let us define a pair of orthogonal (horizontal) null vectors $\mathbf{X}^1_{\vec{s}}$ and $\mathbf{X}^2_{\vec{s}}$ tangent to $ \gamma_{\vec{s}}(0)$ whose projection to $\Q^+$ is given by
\beqn\nonumber
\dd\pi(\mathbf{X}_{\vec{s}}^1)(\pi (\gamma_{\vec{s}}(0))) = a(\vec{s})\frac{\partial}{\partial u}\hspace{.7cm}\textrm{and}\hspace{.7cm}\dd\pi(\mathbf{X}^2_{\vec{s}})(\pi (\gamma_{\vec{s}}(0))) = b(\vec{s})\frac{\partial}{\partial v},
\eeqn
for smooth real-valued functions $a(\vec{s})=a(\pi (\gamma_{\vec{s}}(0)))$ and $b(\vec{s})=b(\pi (\gamma_{\vec{s}}(0)))$ normalized such that
\beqn\nonumber
2 = -g(\mathbf{X}_{\vec{s}}^1, \mathbf{X}_{\vec{s}}^2).
\eeqn

\begin{center}
\includegraphics[scale=.79]{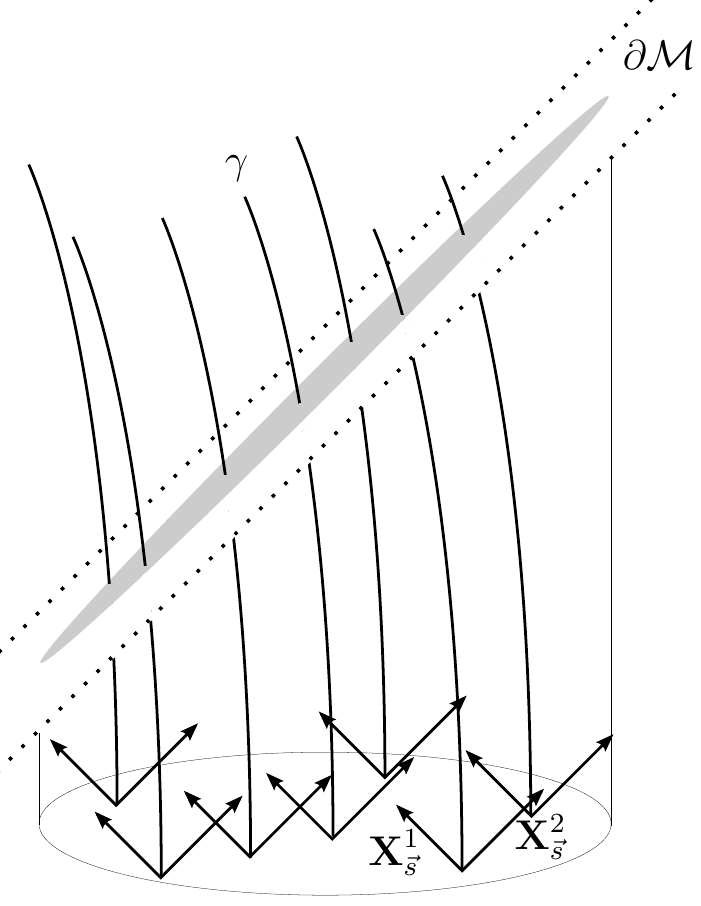}
\end{center}
Parallel transport the vectors $\mathbf{X}^1_{\vec{s}}$ and $\mathbf{X}^2_{\vec{s}}$ along each respective curve $\gamma_{\vec{s}}$.  This defines, in particular, two smooth vector fields $\mathbf{X}^1$ and $\mathbf{X}^2$ on $\widetilde{\mathcal{U}'}\cap \M$.  Because parallel transport preserves nullity, we note that
\beqn\nonumber
\mathbf{X}^1 = a  \frac{\partial}{\partial u}\hspace{.7cm}\textrm{and}\hspace{.7cm} \mathbf{X}^2 = b \frac{\partial}{\partial v}
\eeqn
for smooth real-valued functions $a$ and $b$ defined on $\widetilde{\mathcal{U}'}\cap \M$.
Moreover, because parallel transport preserves the inner product of vectors it follows that
\beqn\nonumber
2=-g(\mathbf{X}^1,\mathbf{X}^2) = ab\Omega^{2}.
\eeqn
Now, if the claim were false, then it must be that, in particular, this family of timelike curves when restricted to $\M$ projects to $J^-(b_{\Gamma})$.  We will show that there exists a choice of bounded co-ordinates on $\Q^+$ for which $\Omega^2$ is uniformly bounded in $J^-(b_{\Gamma})$.  Once this fact is established, it must be the case that there exists a constant $\xi >0$ such that \emph{either} $a > \xi$ or $b> \xi$. This latter fact will allow us to construct a tubular neighborhood of $\gamma$ whose $u$- or $v$-dimension can be uniformly controlled on $\widetilde{\mathcal{U}'}\cap \M$.  Consequently, we will show that this tubular neighborhood (upon restriction to $\M$) can not be entirely contained in $J^-(b_{\Gamma})$, providing a contradiction.  

\textbf{Construction of bounded co-ordinates}\\
Let $(U,V) = b_{\Gamma}$. Given a sequence $(U_j, V)\rightarrow (U,V)$, with $U_j\leq U$, we note that
\beqn\nonumber
\frac{2m}{r}(U_j, V) \not \rightarrow 1,
\eeqn
as $U_j\rightarrow U$.
For, if $\frac{2m}{r}(U_j, V)\rightarrow 1$, then $\limsup_{j\rightarrow \infty}R_{\mu\nu\alpha\beta}R^{\mu\nu\alpha\beta}(U_j, V)=\infty$, contradicting the assumption that $ b_{\Gamma}\in \mathcal{E}$. In particular, we conclude that $\mathcal{A}\cap J^-(b_{\Gamma})$ does not have a limit point on $b_{\Gamma}$.  As a result, we may assume, without loss of generality, that for fixed $v' < V$ and $u' < U$, the region
\beqn\nonumber
\mathcal{D} = \left([u', U] \times [v', V]\right)\cap \Q^+
\eeqn
satisfies $\mathcal{D}\cap \left(\mathcal{A}\cup\mathcal{T}\right)= \emptyset$.  Moreover, without loss of generality, we can assume that $b_{\Gamma}$ does not coincide with $\mathcal{S}_{i^+}\cup\mathcal{CH}_{i^+}\cup i^{\square}$. (If $b_{\Gamma}\in \mathcal{S}_{i^+}\cup i^{\square}$, then $b_{\Gamma}\not\in \mathcal{E}$ as argued above.\footnote{In fact, we have already implicitly ruled out the possibility that $b_{\Gamma}$ coincides with $\mathcal{S}_{i^+}$, as $\mathcal{D}\cap \left(\mathcal{A}\cup\mathcal{T}\right)= \emptyset$.}  If $b_{\Gamma}\in \mathcal{CH}_{i^+}$, then there is nothing to show.) In particular, on the compact set $\left(\{u'\}\times[v', V] \cup [u', U]\times\{v'\}\right)\cap \Q^+$, there exists a constant $c>0$ such that $1-\frac{2m}{r}> c$. Thus, if we re-normalize the co-ordinates on $\mathcal{D}$ such that
\beqn\nonumber
\frac{-\nu}{1-\frac{2m}{r}}(u,V) = 1 \hspace{1cm}\textrm{and}\hspace{1cm}\frac{\lambda}{1-\frac{2m}{r}}(u',v) = 1,
\eeqn
then these new co-ordinates $u$ and $v$ will have finite range. Indeed, for example, 
\beqn\nonumber
U - u' = \int_{u'}^U \frac{-\nu}{1-\frac{2m}{r}}(\overline{u},V)~\dd \overline{u} \leq c^{-1}\int_{u'}^U -\nu(\overline{u}, V)~\dd \overline{u} =c^{-1} r(u', V) < \infty.
\eeqn
By monotonicity (\ref{monokappa}) and (\ref{mono2}), it follows that
\beqn\nonumber
\frac{-\nu}{1-\frac{2m}{r}}(u,v) \leq 1 \hspace{1cm}\textrm{and}\hspace{1cm}\frac{\lambda}{1-\frac{2m}{r}}(u,v) \leq 1,
\eeqn
for all $(u,v)\in \mathcal{D}$.  In particular, we have that $\Omega^2$ is uniformly bounded in $\mathcal{D}$, for (\ref{hm}) yields
\beqn\nonumber
\frac{1}{4}\Omega^2 = \frac{-\lambda\nu}{1-\frac{2m}{r}} =\frac{\lambda}{1-\frac{2m}{r}}\frac{-\nu}{1-\frac{2m}{r}} \left(1-\frac{2m}{r}\right)\leq 1.
\eeqn

\textbf{Construction of the tubular neighborhood}\\
Suppose $a > \xi$ (the case where $b> \xi$ is similar). Consider a 2-dimensional ball $B^2_{\epsilon/4}(0)\subset B^3_{\epsilon/2}(0)$ that is transversal (under the image of the exponential map) to the integral curves of $\mathbf{X}^1$. Let $(s_1, s_2, 0) \in B^2_{\epsilon/4}(0)$ and consider a new family of curves $\hat{\gamma}^+_{(s_1, s_2, 0)}$ defined by following the point $\gamma_{(s_1, s_2, 0)}(t)$ along the future-directed integral curves of $\mathbf{X}^1$ for time $0\leq\tau \leq \delta$. Similarly, one can define $\hat{\gamma}_{(s_1, s_2, 0)}^-$ by flowing along the past-directed integral curves of $\mathbf{X}^1$. For $\delta$ suitably small, this flow, call it $\Phi_1$, defines a diffeomorphism onto its image. In particular, since $\textbf{X}^1$ is a smooth vector field, $\Phi_1$ defines a tubular neighborhood $\widetilde{\mathcal{V}_{\delta}}$ of $\gamma|_{\M}$ in $\M$. In this neighborhood, having flowed time $0<\delta_1\leq \delta$,
\beqn\label{ulower}
u\left(\hat{\gamma}^{+}_{(s_1,s_2, 0)}(\delta_1)\right)- u\left(\hat{\gamma}^{-}_{(s_1,s_2, 0)}(\delta_1)\right)\geq \delta_1 \xi >0.
\eeqn
 We now claim that $\pi(\widetilde{\mathcal{V}_{\delta}}) \not\subset \mathcal{D}$.  Suppose, on the contrary, that $\pi(\widetilde{\mathcal{V}_{\delta}}) \subset\Dm$.  Since
\beqn\nonumber
\pi\left(\widetilde{\mathcal{V}_{\delta}}\right) \cap J^+( \pi \left(\gamma(0)\right))\neq \emptyset,
\eeqn
let $p' = \pi \left(\gamma|_{\M}(t)\right)$.  We note that the $u$-dimension of $J^+(p')\cap \Dm$ tends to zero as $p'\rightarrow b_{\Gamma}$, i.e.,~for $p'', p'''\in J^+(p')\cap \Dm$ we have
\beqn\nonumber
\lim_{p'\rightarrow b_{\Gamma}}\sup_{J^+(p')\cap \Dm}|u(p'') - u(p''')|  = 0.
\eeqn
This then contradicts (\ref{ulower}).  We conclude that $\pi(\widetilde{\mathcal{V}_{\delta}}) \not\subset\Dm$.  Thus, there is a curve $\gamma_{(s_1, s_2, 0)}$ entering the extension for which $\{b_{\Gamma}\}\cap \overline{\pi\left(\gamma_{(s_1, s_2, 0)}|_{\M}\right)} = \emptyset$. Since we have shown that
\beqn\nonumber
\mathcal{E}\cap(\sgo\cup\sgt \cup\mathcal{S} \cup \mathcal{S}_{i^+}\cup i^{\square}\cup \mathcal{I}^+) = \emptyset,
\eeqn
it must be the case that 
\beqn\nonumber
 \mathcal{E}\cap \left(\chg\cup \mathcal{CH}_{i^+}\right) \neq \emptyset,
 \eeqn
establishing the result.

\section{Proof of Theorems \ref{thm:s_t_general} and \ref{thm:w_t_general}: global structure of strongly and weakly tame Einstein-matter systems}\label{sec:proof_general}

We note that only in \S\ref{sec:general_emkg}, in which we established the generalized extension principle for the Einstein-Maxwell-Klein-Gordon system, did we exploit model-specific structure.  Elsewhere, assertions were proven with the aid of simply the dominant (or the weaker null energy) condition.  Accordingly, since both Theorems \ref{thm:s_t_general} and \ref{thm:w_t_general}  presume that a suitable extension principle hold, we may reproduce the proof of either Theorem exactly as in \S \ref{sec:proof_main}.

\bibliographystyle{acm}
\bibliography{bibdesk}
\end{document}